\documentclass[11pt]{article}
\usepackage[utf8]{inputenc}

\usepackage{amsmath,amsfonts,amssymb,amsthm}
\usepackage[letterpaper,margin=1in]{geometry}
\usepackage{graphicx,xcolor}
\usepackage{float}
\usepackage{physics}
\usepackage[colorlinks]{hyperref}
\usepackage[capitalize]{cleveref}

\usepackage{changes}
\usepackage{appendix}
\usepackage{hyperref}
\usepackage{hhline}

\usepackage{mathrsfs}
\usepackage{mathtools}
\mathtoolsset{centercolon}
\usepackage[numbers,comma,sort&compress]{natbib}
\usepackage{algorithm}
\usepackage{algorithmic}
\newcommand{\algname}[1]{\textup{\texttt{#1}}}
\usepackage{bm}
\usepackage{multirow}
\usepackage{makecell}
\usepackage{bbm}
\usepackage{subcaption}

\usepackage{enumitem} % 提供列表参数设置功能
\usepackage{environ}  % 提供 \NewEnviron 和 \BODY 功能

\usepackage{tikz} 
\usetikzlibrary{arrows, shapes.gates.logic.US, calc}
\usetikzlibrary{shapes}
\usetikzlibrary{plotmarks}
\usetikzlibrary{quantikz2}
\usetikzlibrary{patterns}
\usetikzlibrary{shapes.geometric, arrows.meta, positioning, calc, shadows}
\usepackage{tabularx}
\usepackage{booktabs}
\usepackage{subcaption}
\usepackage{array}
\usepackage{colortbl}

\usepackage[final]{pdfpages}

\newcommand{\E}{{\mathbb{E}}}

\def\E{{\mathbb E}}

\def\le{\leqslant}
\def\ge{\geqslant}

\renewcommand{\d}{\,\mathrm{d}}

\newtheorem{theorem}{Theorem}
\newtheorem{lemma}{Lemma}
\newtheorem{corollary}{Corollary}
\newtheorem{proposition}{Proposition}
\newtheorem{assumption}{Assumption}
\newtheorem*{assumption*}{Assumption}
\newtheorem*{theorem*}{Theorem}

\theoremstyle{definition}
\newtheorem{definition}{Definition}

\newtheorem{remark}{Remark}

\allowdisplaybreaks

\newenvironment{tightitemize}{
  \begin{itemize}[topsep=0pt, partopsep=0pt, itemsep=0pt, parsep=0pt]
}{\end{itemize}}

\NewEnviron{tightalign*}{
  \begingroup
  \setlength{\abovedisplayskip}{1pt}%
  \setlength{\belowdisplayskip}{3pt}%
  \begin{align*}
    \BODY
  \end{align*}
  \endgroup
}

\renewcommand{\sec}[1]{\hyperref[sec:#1]{Section~\ref*{sec:#1}}}
\newcommand{\app}[1]{\hyperref[app:#1]{Appendix~\ref*{app:#1}}}
\newcommand{\thm}[1]{\hyperref[thm:#1]{Theorem~\ref*{thm:#1}}}
\newcommand{\lem}[1]{\hyperref[lem:#1]{Lemma~\ref*{lem:#1}}}
\newcommand{\cor}[1]{\hyperref[cor:#1]{Corollary~\ref*{cor:#1}}}
\newcommand{\prb}[1]{\hyperref[prb:#1]{Problem~\ref*{prb:#1}}}
\newcommand{\fgr}[1]{\hyperref[fgr:#1]{Figure~\ref*{fgr:#1}}}
\newcommand{\tab}[1]{\hyperref[tab:#1]{Table~\ref*{tab:#1}}}

\newcommand{\beq}{\begin{equation}}
\newcommand{\eeq}{\end{equation}}
\newcommand{\beqa}{\begin{eqnarray}}
\newcommand{\eeqa}{\end{eqnarray}}
\title{Quantum Algorithms for Gibbs Expectation of Non-log-concave and Heavy-tailed Distributions}
\author{Xinmiao Li$^{1,2}$, Jin-Peng Liu$^{1,3,\thanks{Corresponding author: liujinpeng@tsinghua.edu.cn}}$ \\ 
\footnotesize $^{1}$ Yau Mathematical Sciences Center, Tsinghua University\\
\footnotesize $^{2}$ Qiuzhen College, Tsinghua University\\
\footnotesize $^{3}$ Beijing Institute of Mathematical Sciences and Applications\\
}
\date{}

\begin{document}

\maketitle

\begin{abstract}
    We establish a systematic framework of unbiased quantum sampling and estimation protocols for the classical Gibbs expectation. This framework generalizes existing approaches to the partition function estimation and has broader applications in various fields. We consider sampling and estimation for a wide class of non-log-concave distributions, particularly heavy-tailed ones, under relaxed assumptions beyond strong convexity, such as dissipativity.
    We develop an unbiased extension of quantum-accelerated multilevel Monte Carlo (QA-MLMC) to eliminate all biases from discretization and time truncation, together with introducing a change-of-measure approach and the Girsanov theorem via Radon–Nikodym derivatives. As a result, our approach achieves quantum complexity $\widetilde{\mathcal{O}}(\epsilon^{-1})$ within error $\epsilon$, whereas the classical MLMC requires $\widetilde{\mathcal{O}}(\epsilon^{-2})$ and existing quantum algorithms yield biased estimators under stronger assumptions. Furthermore, our unified framework enables unbiased quantum sampling and estimation for certain heavy-tailed distributions after transformation. We provide several concrete applications of our approach in statistics, machine learning, and finance, towards more practical scenarios of the quantum acceleration of stochastic processes.
    
\end{abstract}

\tableofcontents

\newpage

\section{Introduction}\label{sec:introduction}
With the emerging breakthroughs 
in quantum information theory and quantum technologies, quantum computation has attracted growing 
attention from a wide range of scientific fields.
   Quantum computing offers significant potential for
    accelerating computational tasks in many areas,
    such as quantum many-body simulation, numerical linear algebra, dynamical simulation, and stochastic simulation.

One widely explored direction is the quantum acceleration of stochastic processes,  
which are central to many applications in physics~\cite{Lindblad1976OnTG}, finance~\cite{Rebentrost2018QuantumCF,An2020QuantumacceleratedMM}, and machine learning~\cite{Alonso2025QuantumML}.
While classical algorithms often suffer from high computational complexity, 
especially in high-dimensional or finely discretized settings,  
quantum algorithms based on techniques such as quantum walks~\cite{Aharonov2000QuantumWO,Szegedy2004QuantumSO,Childs2002ExponentialAS} 
and quantum fast-forwarding~\cite{apers2019quantumfastforwardingmarkovchains} can offer quadratic or even exponential speedups 
for these problems.

In recent years, quantum algorithms for stochastic processes 
have seen remarkable progress, particularly through 
the lens of quantum walks and their applications in search, 
simulation, and sampling. 
For instance, both discrete-time and continuous-time quantum walks 
have demonstrated provable quadratic speedups in spatial search 
problems compared to classical random walks, 
even under general settings with multiple marked states and arbitrary 
Markov chains~\cite{Chakraborty2018FindingAM,Apers2021QuadraticSF,Ambainis2019QuadraticSF}. 
While various quantum walk-based search algorithms have been developed~\cite{Szegedy2004QuantumSO,MNRS,Electric,Dohotaru2017ControlledQA}, 
a unified framework was proposed in~\cite{Apers2019AUF} to encompass 
and generalize these approaches, 
offering clearer insights into their performance guarantees 
and the structural conditions required for achieving speedups.

Meanwhile, in the domain of quantum-enhanced sampling and estimation,  
quantum amplitude estimation has led to substantial advances.  
A number of works~\cite{Brassard2000QuantumAA,Heinrich2001QuantumSW,Montanaro2015QuantumSO,Kothari2022MeanEW}  
have developed quantum mean estimation methods that achieve quadratic improvements in sample complexity  
compared to classical Monte Carlo methods.  
Building on this foundation, \cite{An2020QuantumacceleratedMM} proposed a quantum-accelerated multilevel Monte Carlo (QA-MLMC) —  
which combines quantum mean estimation with multilevel variance reduction  
to further reduce the overall computational cost by leveraging hierarchical couplings.  
Recent advancements, such as those in~\cite{Blanchet2025NonlinearQM}, 
continue to extend these ideas in nonlinear and high-dimensional settings.

In the study of sampling and expectation estimation with respect to 
complex probability distributions, underdamped Langevin diffusion (ULD)~\cite{Dalalyan2017UserfriendlyGF,Durmus2018AnalysisOL}, 
the randomized midpoint method for ULD (ULD-RMM)~\cite{Roy2022StochasticZD}, 
and the Metropolis-adjusted Langevin algorithm (MALA)~\cite{Chen2019FastMO,Chewi2020OptimalDD} 
are several commonly used techniques. 
Quantum algorithms have also shown strong potential for accelerating these methods.  
In~\cite{Childs2022QuantumAF}, the authors implemented quantum speedups 
for fundamental problems such as log-concave sampling and estimating normalizing constants.  
They developed Multilevel Quantum Inexact ULD, Multilevel Quantum Inexact ULD-RMM, 
and a quantum annealing approach combined with Quantum MALA, 
achieving quadratic speedups in complexity.  
Moreover,~\cite{Ozgul2025QuantumSF} demonstrates quantum speedups for
gradient estimation via stochastic evaluation oracles and for zeroth-order
sampling, while~\cite{ozgul2023stochastic} studies quantum sampling for
non-log-concave finite-sum potentials using stochastic gradient oracles.

\paragraph{\textbf{\textsf{Main Task}}}
This work focuses on the design and analysis of quantum algorithms 
for sampling from probability distributions and estimating expectations arising in stochastic processes using QA-MLMC methods. 
   More precisely, the goal is to sample from the (classical) \emph{Gibbs distribution}
$$ \pi(x) = \frac{e^{-f(x)}}{\int e^{-f(x)} \mathrm{d}x}$$
and to estimate the \emph{Gibbs expectation}
$$
    \mathbb{E}_{\pi}[\varphi] = \int_{\mathbb{R}^d} \varphi(x) \pi(x) \mathrm{d} x,
$$
where $f : \mathbb{R}^d \to \mathbb{R}$ is the potential function and can be non-convex.
Throughout this work, unless otherwise stated, we assume that the observable $\varphi$ is bounded and $K$-globally Lipschitz, i.e.,
$$
\|\varphi(x) - \varphi(y)\| \leq K \|x - y\| \quad \text{for all } x, y \in \mathbb{R}^d.
$$

This task generalizes the partition function estimation and has broader applications in various fields. In statistical physics, $\mathbb{E}_{\pi}[\varphi]$ describes an ensemble average of a stationary physical system. Particularly, when $\varphi = f$, $\mathbb{E}_{\pi}[f]$ is known as the entropy~\cite{Jaynes1957InformationTA}. In machine learning and generative artificial intelligence, $\mathbb{E}_{\pi}[\varphi]$ estimates an expectation of an energy-based model~\cite{LeCun2006ATO,Song2019GenerativeMB,Andrieu2004AnIT}. In Bayesian statistics, $\mathbb{E}_{\pi}[\varphi]$ characterizes a posterior expectation of a log-likelihood distribution~\cite{GelmanEtAl2013,RobertCasella2004}. 

This is approached by simulating the overdamped Langevin diffusion process,
which is described by the stochastic differential equation (SDE)
\begin{align}\label{SDE0}
        \mathrm{d} X_t = -\nabla f(X_t) \mathrm{d} t + \sqrt{2} \mathrm{d} W_t,
\end{align}
where $(W_t)_{t\geq 0}$ denotes a standard $d$-dimensional Brownian motion with $W_0 = 0$.

In most mean estimation methods for the stochastic process~\eqref{SDE0}, 
a sufficiently large terminal time $T$ is chosen so that the distribution of $X_T$ 
is close to the target distribution $\pi$, 
and then $\mathbb{E}[\varphi(X_T)]$ is used as an approximation.
 Although this approach can effectively control the error, 
 it results in a biased estimator.
 The bias may accumulate and significantly impact downstream computations, 
 particularly in iterative algorithms or when used as input to further estimations.

 \paragraph{Why exact unbiasedness can matter.}
For a single $\epsilon$ accurate expectation, choosing sufficiently large $T$
and small $h$ can already retain the leading
$\widetilde{\mathcal O}(\epsilon^{-1})$ quantum precision scaling.
Exact unbiasedness is especially useful when Gibbs expectations are
repeatedly used by another stochastic algorithm.
A concrete example is stochastic approximation for maximum-likelihood
learning of an energy-based model
$\pi_\theta(x)\propto e^{-f_\theta(x)}$, where one gradient term is
$$
\nabla_\theta\log Z_\theta
=
-\mathbb{E}_{\pi_\theta}[\nabla_\theta f_\theta(X)].
$$
Thus the stationary Gibbs expectation must be estimated repeatedly during
the outer optimization.
Suppose
$$
\theta_{k+1}=\theta_k-\eta_k\widehat g_k,
$$
where $\eta_k$ is the step size, $g(\theta_k)$ is the exact gradient,
and $\widehat g_k$ replaces its stationary expectation term by an estimator.
Our construction gives
$$
\mathbb{E}[\widehat g_k\mid\theta_k]=g(\theta_k),
$$
so the estimation error has zero conditional mean.

In contrast, a fixed-$(T,h)$ Langevin estimator gives
$$
\mathbb{E}[\widehat g_k\mid\theta_k]
=
g(\theta_k)+b_{T,h}(\theta_k),
$$
where $b_{T,h}$ is the finite-time and discretization bias.
This bias enters every outer update, contributing a cumulative perturbation
$\sum_k\eta_k b_{T,h}(\theta_k)$ that repeated sampling cannot remove.
For standard stochastic approximation convergence guarantees, one must
therefore control $b_{T_k,h_k}$ across iterations, typically by increasing
$T_k$ and decreasing $h_k$.
 For instance,
$f_\theta(x)=x^2/2+A\cos x+\theta\tanh x$ with $A>1$ and bounded
$\theta$ is nonconvex in $x$, while the score
$\partial_\theta f_\theta(x)=\tanh x$ is bounded and Lipschitz.
Thus an exactly unbiased Gibbs expectation oracle removes this source
of gradient drift without requiring iteration-dependent refinement of
$T$ and $h$. This does not, by itself, establish an end-to-end training
speedup over well-tuned biased estimators.

To deal with this issue, an \emph{unbiased} estimation method is developed in this work.
 We construct a suitable coupling such that the partial sum $\sum\limits_{i=0}^\ell P_i$ over 
 the first $\ell$ levels serves as an estimator of $X_{T_\ell}$, 
 where the terminal time $T_\ell$ tends to infinity as $\ell$ increases.
By adopting a multilevel approach, we can obtain an unbiased estimation of 
$\mathbb{E}_\pi[\varphi]$ with a cost of 
$\widetilde{\mathcal{O}}(\epsilon^{-1})$.

On the other hand, most quantum estimation algorithms impose strong assumptions 
on the convexity and smoothness of the potential function $f$~\cite{Childs2022QuantumAF}, 
typically requiring that for any $x\neq  y\in\mathbb{R}^d$,
\begin{align*}
2\cdot \frac{f(y) - f(x) - \langle \nabla f(x), y - x \rangle}{\|x - y\|_2^2} \in [\mu, L].
\end{align*}

In this work, 
we relax this condition by introducing a spring term into the evolution
of the stochastic process~\eqref{SDE0},
which effectively controls the coupling distance between different paths. 
As a result, multilevel mean estimation can be performed under dissipativity and twice smoothness assumptions only, i.e.,
  \begin{align*}
    \|\nabla f(x)-\nabla f(y)\|\leq L\|x-y\|,\quad
     &\|\nabla^2 f(x)-\nabla^2 f(y)\|\leq L\|x-y\|\\
     \text{ and }
    \langle x,\nabla f(x)\rangle&\geq a\|x\|^2-\beta.
  \end{align*}
 Under our improved quantum approach for Gibbs expectation,
an unbiased estimator of $\mathbb{E}_\pi[\varphi]$ can be constructed
with variance at most $\epsilon^2$,
at a computational cost of $\widetilde{\mathcal{O}}(\epsilon^{-1})$.
 
Moreover, \emph{heavy-tailed} scenarios have received limited attention
 in current quantum sampling methods.
Since a class of heavy-tailed distributions can be effectively
 addressed by transforming them into distributions with more favorable analytical
  properties,
we leverage this transformation in the second method introduced 
in Section~\ref{sec:heavy-tailed}  to construct a QA-MLMC 
algorithm specifically adapted to the heavy-tailed setting, 
enabling efficient estimation while preserving theoretical guarantees.

\paragraph{\textbf{\textsf{Contributions}}}
This work develops an unbiased QA-MLMC framework for Gibbs
expectation estimation, eliminating both discretization and
finite-time truncation biases while retaining the leading
$\widetilde{\mathcal{O}}(\epsilon^{-1})$ expected quantum complexity.
Our results extend this guarantee beyond strong convexity to
dissipative, non-log-concave distributions and certain heavy-tailed
targets. These results rely on extensions of unbiased estimation and QA-MLMC,
together with weighted coupling analysis and an adapted transformation
for heavy-tailed distributions. The specific advances relative to earlier
methods are detailed below.

 \begin{itemize}
    \item \textbf{General construction of unbiased estimators}\\
$\quad$ We extend the unbiased estimator construction framework  and provide a more general
approach for converting biased estimation procedures into unbiased
estimators. This construction will later be used to obtain an unbiased
QA-MLMC algorithm.
\begin{theorem*}[Informal of Theorem~\ref{thm:unbiased estimator construct}]
Let $P$ be a random variable.
If there exists an algorithm that estimates $\mathbb{E}[P]$
with mean-squared error at most $\sigma^2$
using cost $\widetilde{\mathcal{O}}(C\sigma^{-p})$
for some $p\in(0,2)$,
then one can construct an unbiased estimator of $\mathbb{E}[P]$
with variance at most $\sigma^2$
and expected cost $\widetilde{\mathcal{O}}(C\sigma^{-p})$.
\end{theorem*}

    \item \textbf{Unbiased quantum-accelerated multilevel Monte Carlo}\\
    $\quad$ We strengthen the original QA-MLMC framework introduced in \cite{An2020QuantumacceleratedMM} by achieving an unbiased estimator, establishing variance control in the quantum setting, 
    and explicitly characterizing the computational complexity in terms of the dimension 
$r$ and other parameters of interest.
\begin{theorem*}[Informal of Theorem~\ref{thm:QMLMC2}]
Let $P_\ell$ be multilevel approximations of a $r$-dimensional random variable $P$, 
whose bias, variance decay, and sampling cost satisfy
$$|\E[P_\ell-P]| =\widetilde{\mathcal{O}} \left( K_1 2^{-\alpha \ell}\right),
\quad \operatorname{Var}(P_\ell-P_{\ell-1}) =\widetilde{\mathcal{O}}\left( K_2 2^{-\beta \ell}\right),
\quad \mathcal{C}_\ell  =\widetilde{\mathcal{O}} \left(K_3 2^{\gamma \ell}\right),$$
and $\alpha\geq\min\{\beta/2,\gamma\}$.
Then QA-MLMC produces an unbiased estimator of $\mathbb{E}[P]$ 
with variance at most $\epsilon^2$ and  expected total cost 
    \begin{align*}
        \begin{cases} 
       \widetilde{\mathcal{O}}\left(r^{\frac{1}{2}}K_2^{\frac{1}{2}}K_3 \cdot\epsilon^{-1} \right), & \beta \geq2\gamma, \\
    \widetilde{\mathcal{O}}\left(r^{\frac{1}{2}}K_1^{(\gamma-\frac{1}{2}\beta)/\alpha}K_2^{\frac{1}{2}}K_3\cdot \epsilon^{-1-(\gamma-\frac{1}{2}\beta)/\alpha} \right)  , & \beta < 2\gamma.
        \end{cases}
    \end{align*}
\end{theorem*}

    \item \textbf{Gibbs expectation under dissipativity}\\
$\quad$ We first give a straightforward (biased) estimation for the Gibbs expectation within the QA-MLMC framework.
Extending \cite{Fang2018MultilevelMC}, we introduce a change-of-measure approach by incorporating a spring term 
into the dynamics~\eqref{SDE0} and applying the Girsanov theorem via Radon–Nikodym derivatives, 
which enables effective control of the coupling distance between sample paths and relaxes the original one-sided Lipschitz condition to dissipativity.
\begin{theorem*}[Informal of Theorem~\ref{thm:QA-MLMC under weaker one-sided Lipschitz}]
Assume that $ f $ is $ L $-smooth, $L$-Hessian smooth, dissipative,  
and satisfies the weaker one-sided Lipschitz condition with constant $ \lambda $.  
Then we can construct an estimator $\widehat\mu$ for
$\mathbb{E}\pi[\varphi]$ such that
$\mathbb{E}\bigl[(\widehat\mu-\mathbb{E}\pi[\varphi])^2\bigr]
\le \epsilon^2$.  
The expected total complexity is $$  \widetilde{\mathcal{O}}\left(\frac{SLd(L+\sqrt{d})}{2S-\lambda}\epsilon^{-1} \right), $$
where $S$ denotes the spring coefficient with $S > \lambda/2$.
\end{theorem*}

    \item \textbf{Unbiased Gibbs expectation}\\
     $\quad$ By eliminating the time truncation bias via a suitable coupling strategy, 
we construct two unbiased estimators for the Gibbs expectation: a structured extension of \cite{Giles2016MultilevelMC} 
under the one-sided Lipschitz condition, and a more general approach requiring only dissipativity.
Combined with the unbiased QA-MLMC framework, 
this approach further removes the discretization bias and yields a fully unbiased estimator.
% \jpl{Add a sentence to highlight generalized unbiased QA-MLMC.}
% \lxm{I added a sentence to highlight the generalized unbiased QA-MLMC framework and clarify the two constructions.}
  
  \begin{theorem*}[Informal of Theorem~\ref{thm:unbiased QA-MLMC}]
Assume that $ f $ is  $ L $-smooth, $L$-Hessian smooth, and satisfies the one-sided Lipschitz condition with constant $m$.
Then we can obtain an unbiased estimator $\widehat{\mu}$ of $ \mathbb{E}_\pi[\varphi] $
with $\operatorname{Var}(\widehat{\mu})\le \epsilon^2$. The  expected total complexity is
$$\widetilde{\mathcal{O}}\left(\frac{\,L^2\sqrt{d}+Ld\,}{m^2}\epsilon^{-1}\right).$$
\end{theorem*}

\begin{theorem*}[Informal of Theorem~\ref{thm:unbiased A-MLMC under weaker one-sided Lipschitz}]
     Assume that $ f $ is $ L $-smooth, $L$-Hessian smooth, 
    dissipative,  and satisfies the weaker one-sided Lipschitz 
    condition with constant $ \lambda $.  
Then we can obtain an unbiased estimator $\widehat{\mu}$ of $ \mathbb{E}_\pi[\varphi] $
with $\operatorname{Var}(\widehat{\mu})\le \epsilon^2$. The  expected total complexity is
$$ \widetilde{\mathcal{O}}\left(\frac{SLd(L+\sqrt{d})}{2S-\lambda}\epsilon^{-1} \right),\quad S > \lambda/2.$$

\end{theorem*}

\begin{remark}
In our framework, the dependence on $d$ mainly arises from the upper bound on the variance of the distance between the fine and coarse paths. In addition, when $\beta \geq 2\gamma$ in QA-MLMC (Theorem~\ref{thm:QMLMC2}),  polynomial dependence on $d$ in the bias term $|\mathbb{E}[P_\ell - P]|$ does not affect the overall complexity.
In the one-sided Lipschitz setting, this yields an overall complexity of $\mathcal{O}(d)$. 
In contrast, in the dissipative setting, the change-of-measure technique introduces an additional $d^{1/2}$ factor, resulting in the overall complexity of $\mathcal{O}(d^{3/2})$.
Here we assume that $\|\widehat{Y}^c_0\|$, $\|\widehat{Y}^f_0\|$, and $\|\nabla f(0)\|$ are $\mathcal{O}(1)$; see  Section~\ref{sec:summary} and Appendix~\hyperref[Appendix A: Strong Error Analysis]{A} for details.

\end{remark}

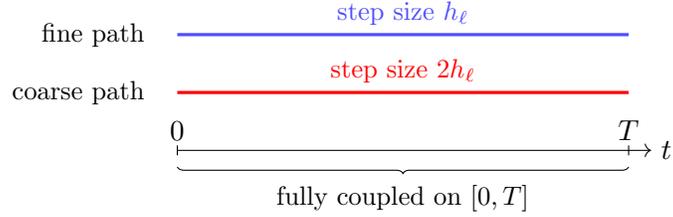
\begin{figure}[ht]
\centering
\begin{tikzpicture}[x=6cm,y=1.1cm]

% time axis
\draw[->] (0,0) -- (1.05,0) node[right] {$t$};
\draw (0,0.06) -- (0,-0.06) node[above=2pt] {$0$};
\draw (1,0.06) -- (1,-0.06) node[above=2pt] {$T$};

% fine path row
\node[left] at (-0.05,1.4) {\small fine path };
\draw[very thick,blue!70] (0,1.4) -- (1,1.4);
\node[above,blue!70] at (0.5,1.4) {\small step size $h_\ell$};

% coarse path row
\node[left] at (-0.05,0.7) {\small coarse path };
\draw[very thick,red] (0,0.7) -- (1,0.7);
\node[above,red] at (0.5,0.7) {\small step size $2h_\ell$};

% coupling 
\draw[decorate,decoration={brace,mirror}] (0,-0.2) -- (1,-0.2)
    node[midway,below=3pt] {\small fully coupled on $[0,T]$};

\end{tikzpicture}
\caption{Level-$\ell$ sampler for Gibbs expectation with fixed terminal time $T$(Theorem~\ref{thm:QA-MLMC of one-sided lipschitz})} \label{fig:1}
\end{figure}

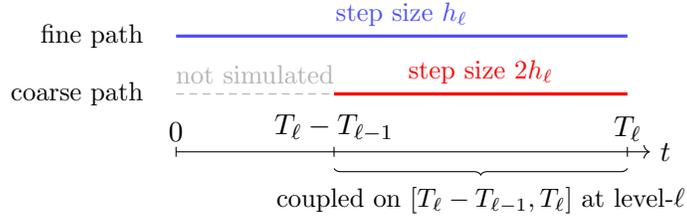
\begin{figure}[ht]
\centering
\begin{tikzpicture}[x=6cm,y=1.1cm]
% time axis
\draw[->] (0,0) -- (1.05,0) node[right] {$t$};
\draw (0,0.06) -- (0,-0.06) node[above=2pt] {$0$};
\draw (0.35,0.06) -- (0.35,-0.06) node[above=2pt] {$T_{\ell}-T_{\ell-1}$};
\draw (1,0.06) -- (1,-0.06) node[above=2pt] {$T_\ell$};

% fine path row
\node[left] at (-0.05,1.4) {\small fine path };
% fine: full 0--T_J
\draw[very thick,blue!70] (0,1.4) -- (1,1.4);
\node[above,blue!70] at (0.5,1.4) {\small step size $h_\ell$};

% coarse path row
\node[left] at (-0.05,0.7) {\small coarse path };
% "not simulated" before T_{l-1}
\draw[densely dashed,gray!60] (0,0.7) -- (0.35,0.7);
\node[above,gray!60] at (0.175,0.7) {\small not simulated};
% coupled tail
\draw[very thick,red] (0.35,0.7) -- (1,0.7);
\node[above,red] at (0.675,0.7) {\small step size $2h_\ell$};

% coupling 
\draw[decorate,decoration={brace,mirror}] (0.35,-0.2) -- (1,-0.2)
    node[midway,below=3pt] {\small coupled on $[T_{\ell}-T_{\ell-1},T_\ell]$ at level-$\ell$};

\end{tikzpicture}
\caption{Level-$\ell$ sampler for unbiased Gibbs expectation with one-sided Lipschitz(Theorem~\ref{thm:unbiased QA-MLMC})} \label{fig:2}
\end{figure}

\begin{figure}[H]
\centering
\begin{tikzpicture}[x=6cm,y=1.1cm]
% time axis
\draw[->] (0,-0.2) -- (1.05,-0.2) node[right] {$t$};
\draw (0,-0.14) -- (0,-0.26) node[above=2pt] {$0$};
\draw (0.33,-0.14) -- (0.33,-0.26) node[above=2pt] {$T_0$};
\draw (0.75,-0.14) -- (0.75,-0.26) node[above=2pt] {$T_{J-1}$};
\draw (1,-0.14) -- (1,-0.26) node[above=2pt] {$T_J$};

% Base term row
%\node[left] at (-0.05,1.9) {base term};

\draw[very thick,blue] (0,2.3) -- (0.33,2.3);
\draw[very thick,red] (0,2.1) -- (0.33,2.1);
\node[left] at (-0.05,2.35) {\scriptsize fine path of $T_0$ };
\node[left] at (-0.05,2.05) {\scriptsize coarse path  of $T_0$};

\node[above=2pt] at (0.23,2.2) { $\varphi(X_{T_0})$ };

% Random increment row
%\node[left] at (-0.05,0.7) {random increment};

% random
\draw[very thick,blue] (0,1.6) -- (0.75,1.6);
\draw[very thick,red] (0,1.4) -- (0.75,1.4);
\node[left] at (-0.05,1.65) {\scriptsize fine path of $T_{j-1}$ };
\node[left] at (-0.05,1.35) {\scriptsize coarse path  of $T_{j-1}$};
\node[above=2pt] at (0.65,1.5) { $\varphi(X_{T_{j-1}})$ };

\draw[very thick,blue] (0,0.9) -- (1,0.9);
\draw[very thick,red] (0,0.7) -- (1,0.7);
\node[left] at (-0.05,0.95) {\scriptsize fine path of $T_{j}$ };
\node[left] at (-0.05,0.65) {\scriptsize coarse path  of $T_{j}$};
\node[above=2pt] at (0.9,0.8) { $\varphi(X_{T_j})$ };

% Extra horizon label
\draw[decorate,decoration={brace,mirror}] (0.75,-0.4) -- (1,-0.4)
node[midway,below=4pt] {\small random increment};
\node[below=2pt] at (0.81,-0.9) {\small $2^j\big(\varphi(X_{T_j})-\varphi(X_{T_{j-1}})\big)$ };

\draw[decorate,decoration={brace,mirror}] (0.0,-0.4) -- (0.33,-0.4)
node[midway,below=4pt] {\small base term};
\node[below=2pt] at (0.16,-0.9) {\small $\varphi(X_{T_0})$ };

\end{tikzpicture}
\caption{Generalized unbiased Gibbs expectation(Theorem~\ref{thm:unbiased A-MLMC under weaker one-sided Lipschitz})} \label{fig:3}

\end{figure}
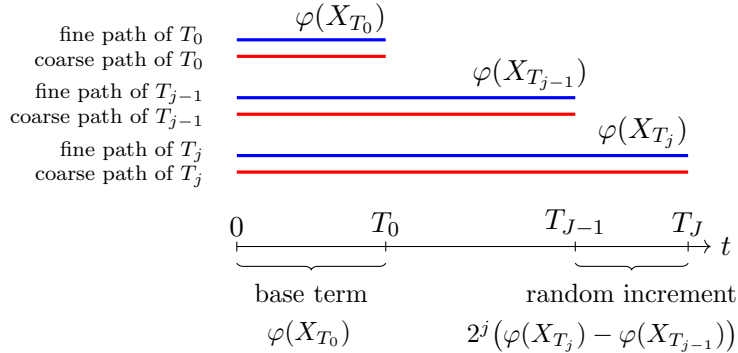

\begin{remark}
Figures~\ref{fig:1}--\ref{fig:3} compare the three constructions for estimating the Gibbs expectation.
Figure~\ref{fig:1} corresponds to a commonly used MLMC-based construction, where the coarse and fine paths are coupled over a fixed interval $[0,T]$.
Figure~\ref{fig:2} illustrates the core idea of Theorem~\ref{thm:unbiased QA-MLMC}, in which the fine and coarse paths target
$\mathbb{E}[\varphi(X_{T_\ell})]$ and $\mathbb{E}[\varphi(X_{T_{\ell-1}})]$, respectively, and are evolved over
$[T_\ell - T_{\ell-1},\, T_\ell]$, thereby enabling access to $X_\infty$.
In the more general setting of Theorem~\ref{thm:unbiased A-MLMC under weaker one-sided Lipschitz}, as illustrated in Figure~\ref{fig:3},
unbiased QA-MLMC is applied at each truncation time, and a geometric random variable
$j \sim \mathrm{Geom}(\frac{1}{2})$ together with the correction term
$2^{j}\big(\varphi(X_{T_j})-\varphi(X_{T_{j-1}})\big)$ removes the truncation bias, yielding $X_\infty$.
\end{remark}

\item \textbf{Heavy-tailed sampling and estimation at light-tailed efficiency}\\
$\quad$ By carefully designing a suitable transformation building on~\cite{He2024AnAO}, we develop a unified framework that enables sampling and estimation for both heavy-tailed and light-tailed distributions with comparable efficiency, 
extending the applicability beyond the original heavy-tailed setting.
\begin{theorem*}[Informal of Theorem~\ref{thm:QA-MLMC of heavy-tailed}]
Suppose that the potential function \(f\), after transformation, gives rise
to a transformed potential \(f_h\) that is dissipative, \(L_h\)-smooth,
\(L_{h}\)-Hessian smooth, and satisfies the weaker one-sided Lipschitz
condition with constant \(\lambda_h\),
then an unbiased estimator of
$\mathbb{E}_{\pi}[\varphi]$
can be constructed with variance at most $\epsilon^2$.
The  expected total query complexity is 
$$\widetilde{\mathcal O}\left(
\frac{S L_h d(L_h+\sqrt d)}{2S-\lambda_h}\epsilon^{-1}
\right),
\qquad S>\lambda_h/2.$$

\end{theorem*}
 \end{itemize}

\begin{table}[H]
\centering
\renewcommand{\arraystretch}{1.25}
\begin{tabularx}{\textwidth}{c|c|c|c}
\noalign{\global\arrayrulewidth=1.5pt}
\cline{1-4}
\noalign{\global\arrayrulewidth=0.4pt}
\textbf{Method} & \textbf{Assumption}& \textbf{Biased} &  \textbf{Cost}  \\
\noalign{\global\arrayrulewidth=1.5pt}
\cline{1-4}
\noalign{\global\arrayrulewidth=0.4pt}
Multilevel ULD~\cite{Ge2019EstimatingNC}
&One-sided Lipschitz, $L$-Smooth
&biased
&$\widetilde{\mathcal{O}}(\epsilon^{-2})$\\
  \cline{1-4} 
MALA~\cite{Ge2019EstimatingNC} 
&One-sided Lipschitz, $L$-Smooth
& biased
& $\widetilde{\mathcal{O}}(\epsilon^{-2})$\\
\cline{1-4}
\makecell{Multilevel Quantum\\ Inexact ULD~\cite{Childs2022QuantumAF}}
&One-sided Lipschitz, $L$-Smooth& 
biased
 & $\widetilde{\mathcal{O}}(\epsilon^{-1})$\\
  \cline{1-4} 
 Quantum MALA~\cite{Childs2022QuantumAF}
 & One-sided Lipschitz, $L$-Smooth& biased 
 & $\widetilde{\mathcal{O}}(\epsilon^{-1})$\\
\noalign{\global\arrayrulewidth=1.5pt}
\cline{1-4}
\noalign{\global\arrayrulewidth=0.4pt}
\multirow{2}{*}{Theorem~\ref{thm:QA-MLMC under weaker one-sided Lipschitz}}
 & Dissipative, $L$-Smooth
 &\multirow{2}{*}{biased}
 & $\widetilde{\mathcal{O}}(\epsilon^{-1.5})$  \\
   \cline{2-2}   \cline{4-4}
 &Dissipative, $L$-Smooth, $L$-Hessian-Smooth&&$\widetilde{\mathcal{O}}(\epsilon^{-1})$\\
\cline{1-4}
\multirow{2}{*}{Theorem~\ref{thm:unbiased QA-MLMC}}
 & One-sided Lipschitz, $L$-Smooth
 & \multirow{2}{*}{unbiased}
 & $\widetilde{\mathcal{O}}(\epsilon^{-1.5})$\\
   \cline{2-2}   \cline{4-4}
 &One-sided Lipschitz, $L$-Smooth, $L$-Hessian-Smooth&&$\widetilde{\mathcal{O}}(\epsilon^{-1})$\\
  \cline{1-4} 
\multirow{2}{*}{Theorem~\ref{thm:unbiased A-MLMC under weaker one-sided Lipschitz} }
 &Dissipative, $L$-Smooth
 & \multirow{2}{*}{unbiased}
 & $\widetilde{\mathcal{O}}(\epsilon^{-1.5})$\\
   \cline{2-2}   \cline{4-4}
 &Dissipative, $L$-Smooth, $L$-Hessian-Smooth&&$\widetilde{\mathcal{O}}(\epsilon^{-1})$\\
\noalign{\global\arrayrulewidth=1.5pt}
\cline{1-4}
\noalign{\global\arrayrulewidth=0.4pt}
\end{tabularx}
\caption{Comparison of different algorithms for Gibbs expectation.}
\label{tab:comparison}
\end{table}
\begin{remark}
More precisely, Theorem~\ref{thm:QA-MLMC under weaker one-sided Lipschitz}
and Theorem~\ref{thm:unbiased A-MLMC under weaker one-sided Lipschitz}
also require the weaker one-sided Lipschitz condition.
However, since the weaker one-sided Lipschitz property
is directly implied by the $L$-smooth assumption,
its effect on the conclusions of the theorems
amounts only to a multiplicative constant.
For the sake of brevity, this condition is therefore omitted here.
\end{remark}

\begin{remark}[Comparison with quantum sampling algorithms]
Table~\ref{tab:comparison} concerns Gibbs expectation estimation rather
than approximate sampling or partition-function estimation. In particular,
it should be distinguished from the quantum sampling results of
\cite{ozgul2023stochastic}, where the accuracy parameter $\epsilon$
refers to the total-variation error of the output distribution, or to the
relative error in partition-function estimation. For bounded $\varphi$, a total-variation guarantee can control the bias of the induced expectation, but it does not directly provide an unbiased estimator with variance at most $\epsilon^2$.
In contrast, the present QA-MLMC framework directly constructs an estimator $\widehat\mu$ satisfying \[ \mathbb E[\widehat\mu]=\mathbb E_\pi[\varphi], \qquad \operatorname{Var}(\widehat\mu)\le \epsilon^2 . \] Thus, the corresponding complexity
bounds are associated with different computational tasks and are not
directly comparable. The algorithms in \cite{ozgul2023stochastic} work
under smoothness and dissipativity assumptions, and their stated
$\epsilon$-dependence is $\widetilde{\mathcal O}(\epsilon^{-1})$ for
sampling with a full-gradient oracle, but the full complexity also depends
on isoperimetric constants such as Cheeger and log-Sobolev constants. In
the present QA-MLMC framework, Hessian smoothness is used to obtain the
strong order-one multilevel variance decay needed for the optimal
$\widetilde{\mathcal O}(\epsilon^{-1})$ expectation-estimation scaling;
without it, the same framework yields
$\widetilde{\mathcal O}(\epsilon^{-3/2})$.
\end{remark}

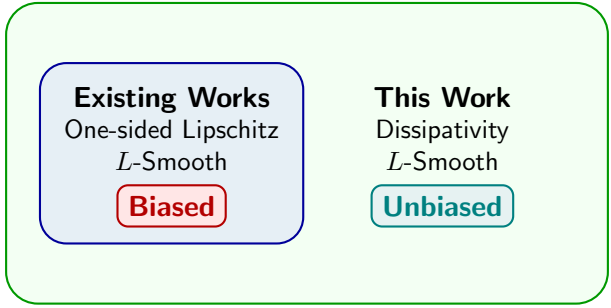
\begin{figure}[ht]
    \centering    
\begin{tikzpicture}[
    font=\sffamily,
    % This work
    set/.style={
        draw, thick, rounded corners=12pt,
        minimum width=8cm, minimum height=4cm,
        fill opacity=0.15
    },
    % Existing works
    subset/.style={
        draw, thick, rounded corners=10pt,
        minimum width=3.5cm, minimum height=2.4cm,
        fill opacity=0.2
    },
]

%This Work
\node[set, fill=green!30, draw=green!60!black] (big) at (0,0.8) {};

\node[align=center, line width=0pt, inner sep=2pt, 
      text depth=0pt, text height=1.5ex,
      execute at begin node=\setlength{\baselineskip}{1.1em}]
at (1.8, 0.7) {
    \textbf{This Work}\\
    \small   Dissipativity\\  
   \small $L$-Smooth
};

\node[draw=teal, fill=teal!10, thick, rounded corners=4pt,
      font=\sffamily\bfseries, text=teal, inner sep=4pt]
      at (1.8,0.1) {Unbiased};

%Existing Work

\node[subset, fill=blue!30, draw=blue!60!black] (small) at (-1.8,0.8) {};

\node[align=center, line width=0pt, inner sep=2pt, 
      text depth=0pt, text height=1.5ex,
      execute at begin node=\setlength{\baselineskip}{1.1em}]
at (-1.8, 0.7) {
 \textbf{Existing Works}\\
    \small One-sided Lipschitz\\
   \small $L$-Smooth
};

\node[draw=red!70!black, fill=red!10, thick, rounded corners=4pt,
      font=\sffamily\bfseries, text=red!70!black, inner sep=4pt]
      at (-1.8,0.1) {Biased};

\end{tikzpicture}
\caption{Extension of Applicable Regimes in This Work}
\end{figure}

\paragraph{\textbf{\textsf{Techniques}}}
The technical ingredients of our quantum algorithms are highlighted below.
\begin{itemize}
    \item \textbf{Unbiased quantum mean estimation}(Theorem~\ref{thm:Unbiased Quantum Mean Estimation})\\
   $\quad$ Unbiased quantum mean estimation forms the foundation of this work,
     providing the quadratic speedup that enables the quantum acceleration of all subsequent algorithmic developments.
    \item \textbf{Transforming biased estimators into unbiased ones}(Theorem~\ref{thm:unbiased estimator construct})\\
   $\quad$ If the original biased estimator satisfies certain conditions, we can convert it into an unbiased estimator with essentially no change in computational complexity.
       \item \textbf{Unbiased quantum-accelerated multilevel Monte Carlo}(Theorem~\ref{thm:QMLMC2})\\
    $\quad$ A randomization-based debiasing mechanism is incorporated into the QA-MLMC framework to remove the truncation bias of the multilevel estimator. This yields unbiased estimators while retaining quantum acceleration for costly level-wise sampling.
    \item \textbf{Spring coupling SDE simulation and change-of-measure approach}(Algorithm~\ref{alg:classical level-l-mlmc with change of measure})\\
   $\quad$  By adding a “spring” between the fine path and the coarse path to draw them closer to each other, one can effectively prevent the difference between the two paths from growing exponentially. Moreover, introducing suitable changes of measure and applying the Girsanov theorem via Radon–Nikodym derivatives provides an equivalent representation of the original expectation of interest under these new measures.
     \item \textbf{Transformation of heavy-tailed distribution}(Section~\ref{section Transformation Map})\\
   $\quad$  By constructing a suitable transformation map, heavy-tailed distributions can be transformed into light-tailed ones with only a constant-factor increase in cost. Moreover, under certain conditions on the initial distribution, the transformed light-tailed distribution satisfies the desired properties.
    
\end{itemize}

\paragraph{\textbf{\textsf{Open questions}}}
Building on the results presented in this work, 
several promising directions remain open for further exploration:
\begin{itemize}

    \item \textbf{Extension to broader heavy-tailed distributions.}  
    The current estimation procedures for heavy-tailed distributions 
    still impose nontrivial conditions on the initial potential function~$ f $.  
    A promising direction is to develop more flexible transformation mappings 
    to accommodate a broader class of heavy-tailed distributions, 
    while preserving the quantum speedup.

    \item \textbf{Towards more general and practical settings.}  
    In many practical scenarios, 
    the oracle~$ O_{\nabla f} $ for evaluating the exact gradient 
    of the potential function may not be available.  
    Instead, we may only have access to a stochastic evaluation oracle, 
    which allows simultaneous queries at two input points 
    under the same stochastic realization.  
    Moreover, the potential function~$ f $ itself may exhibit 
    poor smoothness or even be non-differentiable.  
    These challenges motivate the integration of stochastic gradient estimation 
    and gradient approximation techniques to further enhance and generalize 
    the algorithmic framework.
\end{itemize}

\paragraph{\textbf{\textsf{Structure}}}
The structure of this paper is as follows.

Section~\ref{sec:introduction} provides an introduction to the overall motivation
 and objectives of the work.
 
  Section~\ref{sec:Preliminaries} introduces the notation and definitions
used throughout the paper, and presents the quantum oracles that will
be employed in the quantum algorithms developed later.

Section~\ref{sec:Multilevel Monte Carlo with Quantum Acceleration} 
 provides a brief overview of 
  Monte Carlo and multilevel Monte Carlo methods, 
  together with their quantum accelerations. 
 It further develops an unbiased quantum-accelerated multilevel Monte Carlo framework by extending existing constructions of unbiased
estimators to eliminate the bias introduced in multilevel methods. The section also includes a brief introduction to Markov chain Monte Carlo (MCMC) methods.

Section~\ref{sec:dissipativity} develops a quantum approach for estimating Gibbs expectation under the one-sided Lipschitz condition using the QA-MLMC framework. 
Furthermore, by introducing a spring coupling term via a change-of-measure technique, 
the requirement is relaxed to dissipativity.
We also provide illustrative non-log-concave examples of functions $f$ satisfying these conditions.

Section~\ref{sec:unbiased} builds on Section~\ref{sec:dissipativity} to construct an unbiased estimator, 
which allows us to directly estimate $\mathbb{E}[\varphi(X_\infty)]$ 
without truncating at a finite time $T$ and approximating 
$\mathbb{E}[\varphi(X_T)]$. 
In particular, under the one-sided Lipschitz condition, the unbiased 
construction can be incorporated directly into the multilevel Monte 
Carlo path coupling. In contrast, under the weaker one-sided Lipschitz 
condition, an additional application of the MLMC 
technique is required to achieve unbiasedness.

Section~\ref{sec:heavy-tailed} extends the analysis to heavy-tailed distributions. 
It shows how such distributions can be transformed into ones with more
 favorable properties, and applies the most general result 
 from Section~\ref{sec:unbiased}  
 to construct a quantum algorithm for Gibbs expectation tailored to the
  heavy-tailed setting. Illustrative heavy-tailed examples of functions $f$ under assumptions are provided.

 Section~\ref{sec:application} presents several practical applications, 
including Bayesian mixture modeling, robust regression with correntropy loss, 
and the cosine-modulated quadratic well. 
It also includes a heavy-tailed example arising from truncated losses and capped payoffs in finance.

Section~\ref{sec:summary} summarizes the main results of the paper, discusses the dependence on the dimension, and outlines the limitations and future directions.

The appendices provide additional details: Appendix~\hyperref[Appendix A: Strong Error Analysis]{A} presents the strong error analysis, 
while Appendix~\hyperref[Appendix B: Technical Details for heavy tailed]{B} contains technical details for Section~\ref{sec:heavy-tailed}.

\section{Preliminaries}\label{sec:Preliminaries}

\subsection{Notation}
\begin{itemize}
    \item For a vector $v=\left(v_0,v_1,\cdots,v_{N-1}\right)$, we use $\|v\|$ to denote the 2-norm, i.e.,
    $\|v\|=\sqrt{v_0^2+v_1^2+\cdots+v_{N-1}^2}$;
  \item For vectors $u, v \in \mathbb{C}^N$, the standard inner product is defined as $$\langle u | v \rangle = \sum\limits_{i=0}^{N-1} \overline{u_i} v_i;$$
  \item For an operator $A$, $\|A\|=\mathrm{sup}_{\langle\psi|\psi\rangle=1}\|A|\psi\rangle\|$
is the operator norm (or the spectral norm) of $A$;   
  \item    $f=\mathcal{O}(g)$ means that there exists a constant $M > 0$ such that $f(n) \leq M g(n)$ for all $n$;    \item $f = \Omega(g)$ means that $g =\mathcal{O}(f)$; 
\item $f = \Theta(g)$ means that $g =\mathcal{O}(f)$ and $f = \mathcal{O}(g)$; 
\item In particular, in this work, unless otherwise specified, $f = \widetilde{\mathcal{O}}(g)$  means $f=\mathcal{O}(g\cdot\mathrm{polylog}(\epsilon^{-1}))$,
where $\epsilon$ denotes the target accuracy.
When the accuracy is expressed in terms of a variance parameter $\sigma$, 
the $\widetilde{\mathcal{O}}$ notation is understood analogously.
  \end{itemize}

\subsection{Preparatory Concepts}
Before introducing the main framework, 
we first present several important remarks and definitions 
that will be useful for the subsequent development.

In this paper, we use the Euler-Maruyama method to approximate the SDE~\eqref{SDE0}. 
Given a step size $h$, we define the sequence of random variables $\left\{\widehat{X}_n\right\}_{n\geq 0}$ by
\begin{align}\label{SDE:discrete}
    \widehat{X}_0=X_0,\quad \widehat{X}_{n+1}=\widehat{X}_n-\nabla f(\widehat{X}_n)h+\sqrt{2}\Delta W_n,
\end{align}
where $\Delta W_n=W_{(n+1)h}-W_{nh}$ are the Brownian increments. 
However, in fact, since the coefficient in front of the Brownian motion term is 
constant, the Euler-Maruyama method and the Milstein method coincide in this case.

In the multilevel method introduced later, 
we couple the fine and coarse paths by driving them with the same Brownian motion. 
The fine path is simulated with step size $h$, while the coarse path uses step size $2h$. 
More precisely, let $\widehat{X}^f_0 = X^f_0$ and $\widehat{X}^c_0 = X^c_0$. Then
\begin{align}\label{SDE:discrete couple}
 \widehat{X}^f_{2n+1}=&\widehat{X}^f_{2n}-\nabla f(\widehat{X}^f_{2n})h
 +\sqrt{2}\Delta W^f_{2n},\notag\\
  \widehat{X}^f_{2n+2}=&\widehat{X}^f_{2n+1}-\nabla f(\widehat{X}^f_{2n+1})h
 +\sqrt{2}\Delta W^f_{2n+1},\\
  \widehat{X}^c_{2n+2}=&\widehat{X}^c_{2n}-\nabla f(\widehat{X}^c_{2n})2h
 +\sqrt{2}\Delta W^c_{2n},\notag
\end{align}
where $\Delta W^c_{2n} = \Delta W^f_{2n} + \Delta W^f_{2n+1}$.
In the standard setting, we take $\widehat{X}_0^f = \widehat{X}_0^c$.
However, in Subsection~\ref{subsec:Unbiased Estimation under One-sided Lipschitz}, 
we introduce a method with $\widehat{X}_0^f \neq \widehat{X}_0^c$ 
to eliminate the bias arising from time truncation.

Moreover, lowercase letters are used to represent individual realizations obtained in the sampling procedure, i.e.,
\begin{align}\label{SDE:sample}
x_{n+1} = x_n - h \nabla f(x_n) + \sqrt{2h} \xi_n, 
\quad \xi_n \sim \mathcal{N}(0, I_d).    
\end{align}

In the discretization process, 
the final time horizon $T$ may not be an exact multiple of the step size $h$. 
    This issue is easily resolved by taking $N = \left\lfloor T / h \right\rfloor$, 
evolving the process for $N$ steps with step size $h$, 
and then  performing one additional step with step size $T - Nh$. 
This adjustment does not affect the resulting estimations. 
Therefore, in the subsequent discussion, 
we always assume that $T / h = N \in \mathbb{N}$.

 %To distinguish from $X_T$ in (\ref{SDE}), 
 %let $\tilde{\rho}_n$ denote the distribution of $x_n$ that evolves 
 %following  ULA as Algorithm \ref{alg:ULA}.

 %As $\gamma \rightarrow 0$, ULA recovers the Langevin diffusion (\ref{SDE}) in continuous-time.

\begin{definition}[$L$-Smoothness]\label{def:L-Smooth}
A function $ f : \mathbb{R}^d \to \mathbb{R} $ is said to be $ L $-smooth for some constant $ L > 0 $ if its gradient $ \nabla f $ is $ L $-Lipschitz continuous, i.e., for any $ x, y \in \mathbb{R}^d $,
$$
\|\nabla f(x) - \nabla f(y)\| \leq L \|x - y\|.
$$
\end{definition}

It follows from the definition of $ L $-smoothness that
$$
\|\nabla f(x)\| \leq \|\nabla f(0)\| + L\|x\| \quad \Rightarrow \quad \|\nabla f(x)\|^2 \leq 2\|\nabla f(0)\|^2 + 2L^2 \|x\|^2.
$$
The $L$-smoothness condition
plays a crucial role in the analysis of Langevin dynamics. 
It ensures the Lipschitz continuity of the drift term, 
which guarantees the existence and uniqueness of strong solutions to the SDE. 

Moreover,  
$L$-smoothness is also essential for establishing error bounds of discretization 
schemes like the Euler–Maruyama method, where it ensures that the approximation 
error remains controlled and helps bound
the discrepancy between fine 
and coarse paths in multilevel Monte Carlo.

\begin{definition}[$L$-Hessian-Smoothness]\label{def:L-Hessian-smooth}
A function $ f : \mathbb{R}^d \to \mathbb{R} $ is called 
$L$-Hessian smooth if it is twice continuously differentiable and, 
for all $ x, y \in \mathbb{R}^d $,
\begin{align*}
    \|\nabla^2 f(x) - \nabla^2 f(y)\| \leq L \|x - y\|,
\end{align*}
where $\|\cdot\|$ denotes the operator norm of a matrix.
Equivalently, the Hessian of $f$ is $L$-Lipschitz continuous.
\end{definition}

\begin{definition}[One-sided Lipschitz / Strong convexity]\label{def:One-sided Lipschitz}
A differentiable function $ f : \mathbb{R}^d \to \mathbb{R} $ is said to be
one-sided Lipschitz with constant $ m > 0 $ if, for all $ x, y \in \mathbb{R}^d $,
\begin{align*}
\langle \nabla f(x) - \nabla f(y), x - y \rangle \geq m \|x - y\|^2.
\end{align*}
Equivalently, $f$ is $m$-strongly convex,  i.e., for any $ x, y \in \mathbb{R}^d$,
$$
f(x) \geq f(y) + \langle \nabla f(y), x - y \rangle + \frac{m}{2} \|x - y\|^2.
$$
\end{definition}

The one-sided Lipschitz condition above is in fact equivalent to strong convexity.
We refer to it as a “one-sided Lipschitz condition” 
to maintain consistency with the weaker one-sided Lipschitz condition introduced later.

Moreover, Definition~\ref{def:One-sided Lipschitz} also implies that
\begin{align}\label{eq:ing1}
    \langle \nabla f(x), x \rangle \geq m' \|x\|^2 - 2b \|\nabla f(0)\|^2, \quad \forall x \in \mathbb{R}^d,
\end{align}
for some constants $m' > 0$ and $b \geq 0$, which implies that $\nabla f$ is dissipative, as defined below.

\begin{definition}[Weaker One-sided Lipschitz]\label{def:Weaker One-sided Lipschitz}
A function $ f : \mathbb{R}^d \to \mathbb{R} $ is said to satisfy the weaker one-sided Lipschitz condition with parameter $ \lambda \geq 0 $ if, for all $ x, y \in \mathbb{R}^d $,
\begin{align*}
\langle \nabla f(x) - \nabla f(y), x - y \rangle \geq -\lambda \|x - y\|^2.
\end{align*}
\end{definition}

Note that the weaker one-sided Lipschitz condition can be directly implied by $L$-smoothness. Therefore, when $L$-smoothness holds, we may sometimes omit the weaker one-sided Lipschitz condition.

\begin{definition}[Dissipativity]\label{def:dissipative}
A function $ f : \mathbb{R}^d \to \mathbb{R} $ is said to satisfy the dissipativity condition  
if there exist constants $ \alpha, \beta > 0 $ such that for all $ x \in \mathbb{R}^d $,
$$
\langle x, \nabla f(x) \rangle \geq \alpha \|x\|^2 - \beta.
$$
\end{definition}

This condition plays a crucial role in guaranteeing the long-time stability of the associated Langevin diffusion.  
In particular, \cite{MATTINGLY2002185}
proves that the dissipativity condition is sufficient to guarantee 
the existence of a unique stationary distribution 
and ensures that the process  is geometrically ergodic with respect to this invariant measure.
\begin{lemma}[Ergodicity under Dissipativity]\label{lem:ergodic under dissipative}
Assume that $f$ is dissipative and $L$-smooth, 
then the Langevin diffusion~\eqref{SDE0} is geometrically ergodic
 with invariant measure $\pi \propto e^{-f(x)}$.
\end{lemma}
\begin{remark}[Geometric ergodicity]
Geometric ergodicity means that the distribution of the Langevin diffusion
converges to its invariant measure at an exponential rate. More precisely,
there exist constants \(C>0\) and
\(\kappa>0\), depending on the parameters of the model and
on the initial distribution, such that for a suitable class of test
functions \(\varphi\),
\[
\left|
\mathbb E[\varphi(X_T)]-\mathbb E_{\pi}[\varphi]
\right|
\le
C e^{-\kappa T}.
\]
Equivalently, the finite-time bias can be made smaller than
\(\epsilon\) by choosing
\[
T
=
\mathcal O\!\left(
\kappa^{-1}
\log\frac{C}{\epsilon}
\right).
\]
In the strongly convex case, the contraction rate can often be expressed
explicitly in terms of the strong convexity parameter. Under
dissipativity alone, however, the constants \(C\) and
\(\kappa\) generally depend on the Lyapunov and
minorization constants and are not simply given by a convexity parameter.
\end{remark}

Since the one-sided Lipschitz condition also implies dissipativity,
it follows as a corollary that the combination of the one-sided
Lipschitz condition and \(L\)-smoothness ensures the existence of a
unique invariant measure and geometric ergodicity.

\begin{figure}[ht]
\centering
    
    \begin{tikzpicture}[x=0.75pt,y=0.75pt,yscale=-1,xscale=1]

    \draw (75,20) -- (215,20) -- (215,45) -- (75,45) -- cycle;
    \draw (82,24) node [anchor=north west, inner sep=0.75pt] {One-sided Lipschitz};
    
    \draw (290,20) -- (410,20) -- (410,45) -- (290,45) -- cycle;
    \draw (295,24) node [anchor=north west, inner sep=0.75pt] {Strong convexity};
    
    \draw (50,85) -- (240,85) -- (240,110) -- (50,110) -- cycle;
    \draw (55,89) node [anchor=north west, inner sep=0.75pt] {Weaker one-sided Lipschitz};
    
    \draw (290,85) -- (410,85) -- (410,110) -- (290,110) -- cycle;
    
    \draw (110,150) -- (180,150) -- (180,175) -- (110,175) -- cycle;
    \draw (115,154) node [anchor=north west, inner sep=0.75pt] {$L$-smooth};
    
    \coordinate (C_osl)  at ($(75,20)!0.5!(215,45)$);
    \coordinate (C_sc)   at ($(290,20)!0.5!(410,45)$);
    \coordinate (C_wosl) at ($(50,85)!0.5!(240,110)$);
    \coordinate (C_diss) at ($(290,85)!0.5!(410,110)$);
    \coordinate (C_ls)   at ($(110,150)!0.5!(180,175)$);
    
    \node at (C_diss) {Dissipativity};

    \coordinate (M_top) at ($($(215,20)!0.5!(215,45)$)!0.5!($(290,20)!0.5!(290,45)$)$);
    \coordinate (M_mid) at ($($(240,85)!0.5!(240,110)$)!0.5!($(290,85)!0.5!(290,110)$)$);
    
    \coordinate (V_left1)  at ($(C_osl)!0.5!(C_wosl)$); % between OSL and WOSL
    \coordinate (V_left2)  at ($(C_wosl)!0.5!(C_ls)$);  % between WOSL and L-smooth
    \coordinate (V_right1) at ($(C_sc)!0.5!(C_diss)$);  % between SC and Diss
    
    \node at (M_top) {$\Longleftrightarrow$};
    
    \node at (M_mid) {$\iff$};
    \node at ($(M_mid)+(0,0pt)$) {$\times$}; 

    \node at (V_left1)  [rotate=90] {$\Longleftarrow$};
    \node at (V_left2)  [rotate=90] {$\Longrightarrow$};
    \node at (V_right1) [rotate=90] {$\Longleftarrow$};
    
    \end{tikzpicture}
    \caption{Relationships among the above several definitions}
\end{figure}
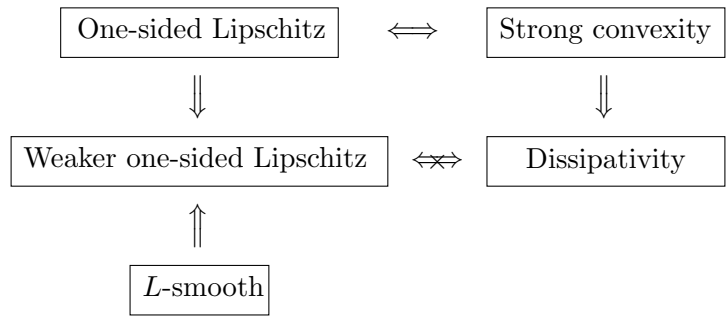

\subsection{Oracles in Quantum Algorithms}~\label{sec:(sub)oracles}
In quantum algorithms, an oracle is a black-box unitary operator that provides query access,
and the computational cost is often measured in terms of the number of oracle queries.

Throughout this paper, the complexity of a randomized-level debiasing
procedure denotes its expected ideal-oracle query cost, not a
per-run deterministic upper bound or a gate-complexity guarantee.

In Algorithms~\ref{alg:Level l sample with change-of-measure QA-MLMC},
\ref{alg:l-sample for QA-MLMC(unbiased)}, and
\ref{alg:Quantum Transformed Langevin Algorithm},
several types of oracles are employed.
For clarity, we introduce and specify these oracles here before presenting the algorithms.

\paragraph{\textbf{\textsf{Common Oracles}}}
The following oracles are shared by all algorithms considered in this paper:
\begin{itemize}
    \item \textbf{Sampling oracle} $O_I : |0\rangle \mapsto \sum\limits_{x} \sqrt{\rho_0(x)} |x\rangle$, which prepares a quantum state encoding the initial distribution $\rho_0$;
    \item \textbf{Gradient oracle} $O_{\nabla f} : |x\rangle |0\rangle \mapsto |x\rangle |\nabla f(x)\rangle$, which outputs the gradient of $f$ at input $x$;
    \item \textbf{Query oracle} $O_\varphi :|x\rangle|0\rangle\mapsto |x\rangle|\varphi(x)\rangle$, which returns $\varphi(x)$, the quantity of interest;
   \item \textbf{Randomness oracle} $O_W : |s\rangle |0\rangle \mapsto |s\rangle |\xi_s\rangle$, where $|\xi_s\rangle$ encodes a pseudo-random Gaussian vector in $\mathbb{R}^d$, generated deterministically from the classical seed $s$.
    Note that in classical MLMC, 
we drive both the fine and coarse paths using the same random variable, 
which induces a strong correlation between them and helps reduce the variance.
In a reversible quantum implementation, the same seed can be used to
compute the increment for both paths without cloning an unknown state.
For coherent quantum mean estimation, the seed register must be prepared
coherently and the generated values subsequently uncomputed.
\end{itemize}

\paragraph{\textbf{\textsf{Algorithm-Specific Oracles}}}
In addition to the common oracles above, each algorithm employs
certain algorithm-specific oracles:

\begin{itemize}
  \item Algorithm~\ref{alg:Level l sample with change-of-measure QA-MLMC}
  additionally uses the Radon-Nikodym derivative oracle $$O_R : |y_1\rangle|y_2\rangle|v\rangle|h\rangle|0\rangle \mapsto |y_1\rangle|y_2\rangle|v\rangle|h\rangle|R(y_1,y_2,v,h)\rangle,$$
which returns the single-step Radon--Nikodym derivative induced by the change of measure;
  \item Algorithm~\ref{alg:Quantum Transformed Langevin Algorithm}
employs the following oracles:
\begin{itemize}
  \item Query oracle $O_h$ for the transformation function $h$: $O_h|x\rangle|0\rangle = |x\rangle|h(x)\rangle$, which provides query access to the transformation function $h$;
  \item Correction oracle $O_a:O_a|x\rangle|y\rangle = |x\rangle|y - \nabla \mathrm{log}|\det J_h(x)|\rangle;$, which computes the additional correction term arising from the transformation;
  \item Jacobian--vector product oracle $O_m$: $O_m|x\rangle|y\rangle = |x\rangle|J_h(x)^\top y\rangle$, which returns the product of the Jacobian of the transformation with a given vector.
\end{itemize}
The details of the transformation function, the associated correction term,
and the Jacobian are given in Section~\ref{sec:(subsec)quantum sampling-heavy}, where the quantum transformed Langevin
dynamics is introduced.

\end{itemize}

\section{ Multilevel Monte Carlo with Quantum Acceleration}
\label{sec:Multilevel Monte Carlo with Quantum Acceleration}

The main techniques used in this work for estimating Gibbs expectations are the
quantum mean estimation method and the quantum-accelerated multilevel
Monte Carlo method. Therefore, in this section we first review the classical Monte Carlo method and its quantum counterpart, namely quantum mean estimation, as well as the classical multilevel Monte Carlo method and its quantum acceleration.

Beyond this review, we further extend the unbiased estimator construction
framework proposed in \cite{Ozgul2025QuantumSF,Sidford2023QuantumSF}. In
particular, in Subsection~\ref{subsec:Quantum-accelerated unbiased multilevel Monte Carlo method},
we present a more general construction of unbiased estimators and introduce an
unbiased quantum-accelerated multilevel Monte Carlo method.

Finally, to prepare for the Gibbs expectation problem studied later, we
introduce Markov Chain Monte Carlo in
Subsection~\ref{subsec:Markov Chain Monte Carlo}.

\subsection{Quantum Mean Estimation Method}

Monte Carlo method is a broad class of computational algorithms 
that rely on repeated random sampling to obtain numerical results~\cite{Metropolis1949TheMC}.
It is widely used for numerical integration, optimization, and simulating complex systems across physics, finance, statistics, and engineering.

In classical Monte Carlo methods,
our goal is to estimate $\E[P]$, 
where $P$ is a random variable defined on the probability space $(\Omega,\mathcal{F},{\mathbb{P}})$ 
with the assumption that $\operatorname{Var}(P)<\infty$. 
We can use the following Monte Carlo Simulation to estimate it.
 \begin{algorithm}[H]
  \begin{algorithmic}[1]
    \caption{\algname{Monte Carlo Simulation}}
    \label{alg:classical monte carlo}
    \item \textbf{Input}: Random variable $P(\omega)$, number of samples $N$\\
    \textbf{Output}: An estimate of $\mathbb{E}[P]$
    \STATE \textbf{For} $i = 1$ \textbf{to} $N$
     \begin{itemize}
       \item Sample $\omega^{(n)} \sim \mathbb{P}$ from probability space $(\Omega, \mathcal{F}, \mathbb{P})$ independently 
     \item  Compute $P(\omega^{(n)})$
     \end{itemize}   
    \textbf{End for}
    \STATE Compute estimate: $\tilde{\mu} =\frac{1}{N} \sum\limits_{i=1}^{N} P(\omega^{(n)})$
    \STATE \textbf{Return} $\tilde{\mu}$
  \end{algorithmic}
\end{algorithm}
Since the $\omega^{(n)}$ are all independent, we have
\begin{align*}
    \operatorname{Var}\left(\frac{1}{N} \sum\limits_{i=1}^{N} P(\omega^{(n)})\right)=\frac{\operatorname{Var}(P)}{N},
\end{align*}
and by Chebyshev's inequality:
\begin{align*}
    {\mathbb{P}}\left[\left|\frac{1}{N} \sum\limits_{i=1}^{N} P(\omega^{(n)})-\E[P]\right|\geq \epsilon\right]\leq \frac{1}{\epsilon^2}\E\left[\left|\frac{1}{N} \sum\limits_{i=1}^{N} P(\omega^{(n)})-\E[P]\right|^2\right]=\frac{\operatorname{Var}(P)}{\epsilon^2 N}.
\end{align*}
So when $N=\mathcal{O}\left(\frac{\operatorname{Var}(P)}{\epsilon^2}\right)$, 
we can estimate $\E[P]$ with additive error $\epsilon$ and success probability at least $0.99$.

In the above algorithm, the success probability is lower bounded by 0.99. 
Sometimes we desire a higher success probability, 
in which case we can use the powering lemma below to boost it arbitrarily close to 1 by repetition.
\begin{lemma}[Powering Lemma \cite{Jerrum1986RandomGO}]
\label{lem:powering lemma}
Let $\mathcal{A}$ be a (classical or quantum) algorithm that outputs an estimate 
 $\tilde{\mu}$ of a target quantity $\mu$,
and satisfies $|\mu - \tilde{\mu}| \leq \epsilon$ 
 with success probability at least $1-\gamma$, 
 for some fixed $\gamma < \frac{1}{2}$. Then for any $\delta > 0$, 
 it suffices to repeat $\mathcal{A}$ $\mathcal{O}(\mathrm{log} 1/\delta)$ times 
 and take the median to obtain an estimate of $\mu$ 
 with additive error $\epsilon$ and success probability at least $1-\delta$.
\end{lemma}

Quantum computing offers a more efficient approach for estimating stochastic quantities. Using a quantum computer, the number of operations can be reduced almost quadratically compared to the classical bound.
\cite{Brassard2000QuantumAA} proposed quantum amplitude estimation, 
laying the foundation for the subsequent development of  quantum mean estimation. 
Early research mainly focused on approximating the mean of a uniform distribution;
for example, Heinrich proposed an asymptotically optimal quantum algorithm for this problem~\cite{Heinrich2001QuantumSW},
which uses $\mathcal{O}(\epsilon^{-1})$ queries to achieve an additive error of $\epsilon$. 
Montanaro extended Heinrich's method, 
 enabling quantum mean estimation to be applied to general variables with finite variance~\cite{Montanaro2015QuantumSO}. 
Furthermore, Kothari and O'Donnell proposed a new quantum mean estimation method~\cite{Kothari2022MeanEW},
which allows for estimation with a cost of $\mathcal{O}(\epsilon^{-1})$
in the case of ``having the code,'' i.e., when we have access to
a unitary circuit ${\mathbb{P}}$ for the distribution $p$ such that
\begin{align*}
{\mathbb{P}}|\vec{0}\rangle
= \sum_{\omega\in\Omega}\sqrt{p(\omega)}\,|\omega\rangle|\text{garbage}_\omega\rangle,
\end{align*}
as well as controlled-$\mathcal{Y}$ and controlled-$\mathcal{Y}^\dagger$, where
$\mathcal{Y}$ is a unitary circuit satisfying
\begin{align*}
\mathcal{Y}|\omega\rangle |0^b\rangle = |\omega\rangle|y(\omega)\rangle
\end{align*}
for all $\omega\in \Omega$. Here \(y:\Omega\to\mathbb R\) denotes the
random variable whose mean $\mathbb E_{\omega\sim p}[y(\omega)]$
is to be estimated.

 In this paper, we primarily adopt the following unbiased mean estimation approach, as it provides an unbiased estimator.

\begin{theorem}[Unbiased Quantum Mean Estimation~\cite{Sidford2023QuantumSF}]
    \label{thm:Unbiased Quantum Mean Estimation}
For a  $r$-dimensional random variable  $X$ with 
$\operatorname{Var}(X) \leq \sigma^2$  and some  $\hat{\sigma} \geq 0$,  
suppose we are given access to its quantum sampling oracle  $O_X$.  
Then, there is a procedure that uses 
 $\widetilde{O}\left(\frac{r^{\frac{1}{2}}\sigma}{\hat{\sigma}}\right)$  
  queries to $O_X$  and outputs an unbiased estimate  $\hat{\mu}$ 
  of the expectation $\mu$  satisfying  
  $\operatorname{Var}(\hat{\mu}) \leq \hat{\sigma}^2.$

\end{theorem}

In classical Monte Carlo,  suppose that
$\operatorname{Var}(X)\leq \sigma^2$ and we aim to construct an estimator
$\hat{\mu}_c$ of $\mu=\mathbb{E}[X]$ whose variance is bounded by  $\hat{\sigma}^2$. 
The empirical mean with $N$ independent samples is $$\hat{\mu}_c=\frac{1}{N}\sum_{i=1}^N X_i,$$ and 
$$
\operatorname{Var}(\hat{\mu}_c)=\frac{\operatorname{Var}(X)}{N}
\leq \frac{\sigma^2}{N}.
$$
Therefore, achieving $\operatorname{Var}(\hat{\mu}_c)\leq \hat{\sigma}^2$
requires
$$
N=\Theta\left(\frac{\sigma^2}{\hat{\sigma}^2}\right)
$$
samples.

In contrast to Theorem~\ref{thm:Unbiased Quantum Mean Estimation}, the classical
procedure exhibits a quadratic dependence on the target variance level
$\hat{\sigma}^2$, whereas the quantum estimator achieves the same variance
guarantee using only
$\widetilde{\mathcal{O}}\left(r^{\frac{1}{2}}\sigma/\hat{\sigma}\right)$ oracle queries.

\subsection{Quantum-accelerated Multilevel Monte Carlo Method}

Heinrich first introduced the multilevel variance reduction technique in~\cite{Heinrich1998MonteCC} and then developed it to 
multilevel Monte Carlo (MLMC) method in~\cite{Heinrich2001MultilevelMC}
and apply to stochastic solution of integral equations.
Then Giles introduced MLMC to simulate SDEs~\cite{41cc5c22-cf88-3432-8b5a-b09edd90ee7a} and 
made further developments~\cite{Heinrich2001MultilevelMC,Lester2014ExtendingTM,Fang2020AdaptiveEM,Giles20192AA} .

The control variate is a classical variance reduction technique 
used to improve the efficiency of Monte Carlo methods~\cite{Glasserman2003MonteCM}.
 Suppose our goal is to estimate $\mathbb{E}[f]$,
  and we are given a control variate $g$ that is well correlated with 
  $f$ and whose expectation $\mathbb{E}[g]$ is known. 
  In this case, we can construct the following unbiased estimator of 
  $\mathbb{E}[f]$ using $N$ independent samples $\omega^{(n)}$:
$$
N^{-1} \sum_{n=1}^N
\left\{ f^{(n)} - \lambda \left( g^{(n)} - \mathbb{E}[g] \right) \right\}.
$$
The optimal value for $\lambda$ is $\rho \sqrt{\operatorname{Var}(f) / \operatorname{Var}(g)}$,
where $\rho$ is the correlation between $f$ and $g$, and the variance 
of the control variate estimator is reduced by factor $1-\rho^2$ 
compared to the standard estimator.

The idea behind MLMC is similar and can be seen as a generalization of the control variate method.
Consider a sequence of random variables $P_0, P_1, \ldots, P_{L-1}$ that provide increasingly accurate approximations of $P_L$,
with computational cost also increasing with the level. Then we have
$$
\mathbb{E}[P_L] = \mathbb{E}[P_0] + \sum_{\ell=1}^L \mathbb{E}[P_\ell-P_{\ell-1}].
$$
This leads to the following unbiased estimator for 
$\mathbb{E}[P_L]$,
$$
N_0^{-1} \sum_{n=1}^{N_0} P_0^{(0,n)} \ + \ 
 \sum_{\ell=1}^L \left[
N_\ell^{-1} \sum_{n=1}^{N_\ell} \left(P_\ell^{(\ell,n)} - P_{\ell-1}^{(\ell,n)}\right)
\right].
$$
Here, the level index $\ell$ in the superscript $(\ell,n)$ 
indicates that the samples used for each level of correction are independent.

By choosing $N_\ell \propto \sqrt{\frac{V_\ell}{\mathcal{C}_\ell}} $,
the total cost is minimized subject to the variance being fixed and bounded by $\epsilon^2$. 
Here $C_0$ and $V_0$ denote the cost and variance of $P_0$, 
while $\mathcal{C}_\ell$ and $V_\ell$ denote the cost and variance of $P_\ell - P_{\ell-1}$.

More precisely, the algorithm and theorem of MLMC can be formally stated as follows.

\begin{algorithm}[H]
  \begin{algorithmic}[1]
    \caption{\algname{Multilevel Monte Carlo (MLMC)}}
    \label{alg:classical MLMC}

    \item \textbf{Input}:  Sequence of random variables $P_0, P_1, \ldots, P_L$,
    number of samples $N_0, N_1, \ldots, N_L$

    \textbf{Output}: An estimate of $\mathbb{E}[P]$

    \STATE \textbf{For} $\ell = 0$ \textbf{to} $L$
    \begin{itemize}
        \item  Generate $N_\ell$ independent samples
   $ \{ P_\ell^{(\ell, n)} - P_{\ell-1}^{(\ell, n)} \}_{n=1}^{N_\ell},$
    with $P_{-1}= 0$
 \item Compute     $$
    Y_\ell = \frac{1}{N_\ell} \sum_{n=1}^{N_\ell} \left( P_\ell^{(\ell,n)} - P_{\ell-1}^{(\ell,n)} \right)
    $$
    \end{itemize}
\textbf{End for}

    \STATE \textbf{Return} $Y=\sum\limits_{\ell=0}^L Y_\ell$
  \end{algorithmic}
\end{algorithm}

\begin{theorem}[Multilevel Monte Carlo]\label{thm:MLMC}
    Let $P$ denote a random variable whose expectation we want to estimate.
    Suppose $P_0, P_1,\ldots$ is a sequence that  approximates $P$.
   
    If there exist independent estimators $Y_\ell$ based on $N_\ell$ 
Monte Carlo samples,  each with expected cost $\mathcal{C}_\ell$, and positive constants 
$\alpha, \beta, \gamma$ such that 
$\alpha\geq{\textstyle \frac{1}{2}}\mathrm{min}(\beta,\gamma)$ and
    \begin{itemize}
        \item $\left| \mathbb{E}[P_\ell - P] \right|\ = \mathcal{O}(2^{-\alpha \ell})$;
        \item $\mathbb{E}[Y_\ell]\ = \left\{ \begin{array}{ll}
\mathbb{E}[P_0],                     & \ell=0, \\
\mathbb{E}[P_\ell - P_{\ell-1}], & \ell>0;
\end{array}\right.
$
        \item $\operatorname{Var}(Y_\ell)\ =\mathcal{O}\left(  2^{-\beta \ell}\right)$;
        \item $\mathcal{C}_\ell\ =\mathcal{O}\left(  2^{\gamma \ell}\right)$.
    \end{itemize}
    Then for any fixed $\epsilon < 1/e$,
    there are values $L$ and $N_\ell$ for which the multilevel estimator
$$
Y = \sum_{\ell=0}^L Y_\ell,
$$
has a mean-square-error with bound
$$
MSE \equiv \mathbb{E}\left[ \left(Y - \mathbb{E}[P]\right)^2\right] < \epsilon^2
$$
with a computational complexity:
        $$
    \begin{cases}
        \mathcal{O}(\epsilon^{-2}), & \gamma < \beta,\\
        \mathcal{O}(\epsilon^{-2}(\mathrm{log}\epsilon)^2), & \gamma = \beta,\\
        \mathcal{O}(\epsilon^{-2 - (\gamma - \beta)/\alpha}), & \gamma > \beta.
    \end{cases}
    $$
\end{theorem}

MLMC can also be accelerated using quantum techniques. 
In particular,  \cite{An2020QuantumacceleratedMM} proposed a quantum-accelerated multilevel Monte Carlo (QA-MLMC) method and 
applied it to compute expectation values derived from classical solutions of
SDEs.

\begin{theorem}
\label{thm:QA-MLMC}
Let $P$ be a random variable, and let $P_\ell~(\ell=0,1,\ldots,L)$ 
be a sequence of random variables approximating $P$ at level $\ell$.
 Further define $P_{-1}=0$.
Let $\mathcal{C}_\ell$ be the cost of sampling from $P_\ell$, 
and let $V_\ell$ be the variance of $P_\ell-P_{\ell-1}$.
If there exist positive constants $\alpha,
\beta,\gamma$
such that 
\begin{itemize}
\item $|\E[P_\ell-P]| = \mathcal{O}(2^{-\alpha \ell})$,
\item  $V_\ell= \mathcal{O}(2^{-\beta \ell})$,
\item $\mathcal{C}_\ell = \mathcal{O}(2^{\gamma \ell})$.
\end{itemize}
Then for any $\epsilon< 1/e$ 
there is a quantum algorithm that 
estimates $\E[P]$ up to additive error
$\epsilon$ with probability at least 0.99, and with cost
    \begin{align*}
        \begin{cases} 
        \mathcal{O}\left(\epsilon^{-1}(\mathrm{log}\,\epsilon^{-1})^{\frac{3}{2}} (\mathrm{log} \mathrm{log} \,\epsilon^{-1})^2\right), & \beta > 2\gamma, \\
        \mathcal{O}\left(\epsilon^{-1} (\mathrm{log} \,\epsilon^{-1})^{\frac{7}{2}} (\mathrm{log} \mathrm{log} \,\epsilon^{-1})^2\right), & \beta =2 \gamma, \\
        \mathcal{O}\left(\epsilon^{-1 - (\gamma - \frac{1}{2}\beta)/\alpha}(\mathrm{log}\,\epsilon^{-1})^{\frac{3}{2}} ( \mathrm{log} \mathrm{log} \,\epsilon^{-1})^2\right), & \beta < 2\gamma.
        \end{cases}
    \end{align*}
\end{theorem}

\subsection{Unbiased Quantum-accelerated Multilevel Monte Carlo Method}
\label{subsec:Quantum-accelerated unbiased multilevel Monte Carlo method}
Recall the MLMC decomposition
\begin{align*}
\mathbb{E}[P]=\mathbb{E}[P_0]+\sum_{\ell=1}^{\infty}\mathbb{E}[P_\ell - P_{\ell-1}].
\end{align*}
In practical computations, 
one truncates the infinite sum at a suitably large level $L$, yielding the approximation
\begin{align*}
\mathbb{E}[P]\approx\mathbb{E}[P_L]
=\mathbb{E}[P_0]+\sum_{\ell=1}^{L}\mathbb{E}[P_\ell - P_{\ell-1}] .
\end{align*}
Such a truncation inevitably introduces a bias given by
$\mathbb{E}[P_L] - \mathbb{E}[P]$.

Fortunately, existing works have shown that this bias can be eliminated
by introducing an additional random variable together with a suitable construction.
In \cite{McLeish2010AGM}, Don McLeish introduced an additional randomization
scheme that enables the construction of an unbiased estimator of
$x_\infty$ from a deterministic sequence $\{X_i\}$ whose expectations satisfy
$\mathbb{E}[X_i] \to x_\infty$.
Rhee and Glynn have also contributed a number of related 
works~\cite{Rhee2012ANA,Rhee2015UnbiasedEW,glynn2014exactestimationmarkovchain},
in which they proposed exact estimation algorithms that provide
unbiased estimators for equilibrium expectations.

The key idea underlying these works is similar.
We take \cite{glynn2014exactestimationmarkovchain} as a representative example.
Given a sequence of biased approximations $\{Y_k\}$ with
$\mathbb{E}[Y_k] \to \mathbb{E}[Y]$
and increments $\{\Delta_k\}$ satisfying
$\mathbb{E}[\Delta_k] = \mathbb{E}[Y_k - Y_{k-1}]$,
let $N$ be a $\mathbb{Z}_+$-valued random variable
independent of $\{\Delta_k\}$, and define
$$
Z = \sum_{k=0}^{N}
\frac{\Delta_k}{\mathbb{P}[N \geq k]}.
$$
Provided that
\begin{align*}
   \sum_{k=0}^{\infty} \mathbb{E}\left[ |\Delta_k|\right] < \infty,
\end{align*}
we obtain
\begin{align*}
    \mathbb{E}[Z]=\mathbb{E}\left[\sum_{k=0}^{N}
\frac{\Delta_k}{\mathbb{P}[N \geq k]}\right]
=\sum_{n=0}^{\infty}\mathbb{P}[N = n]\sum_{k=0}^{n}\frac{\mathbb{E}[\Delta_k]}{\mathbb{P}[N \geq k]}
=\sum_{k=0}^{\infty} \mathbb{E}\left[ \Delta_k\right]=\mathbb{E}[Y],
\end{align*}
Hence, $Z$ is an unbiased estimator for $\mathbb{E}[Y]$.

Similar techniques have also been employed in quantum algorithms.
Following the bias-reduced methods in~\cite{Asi2021StochasticBG,Blanchet2019UnbiasedMM},
the work of \cite{Sidford2023QuantumSF}
leveraged this idea to develop quantum variance reduction techniques,
while \cite{Ozgul2025QuantumSF} applied it to stochastic optimization,
obtaining unbiased gradient estimators.

Building upon these two works, we generalize the unbiased construction
to a broader class of biased procedures $\mathcal{A}$.
More specifically, let $\mathbb{E}[P]$ denote the target expectation.
For any $\sigma > 0$, suppose there exists an algorithm
$\mathcal{A}(\sigma)$ that returns a (possibly biased) estimator
$\mu$ such that
$\mathbb{E}\left[
\left(\mu - \mathbb{E}[P]\right)^2
\right]\leq\sigma^2,$
with query complexity
$\widetilde{\mathcal{O}}\left(\sigma^{-p}\right).$
While the approaches in
\cite{Ozgul2025QuantumSF,Sidford2023QuantumSF}
are restricted to the case $p=1$,
we generalize the unbiased construction to cover the full range
$p \in (0,2)$.

\begin{theorem}
  \label{thm:unbiased estimator construct}
Let $P$ be a random variable. Suppose that an algorithm $\mathcal{A}({\sigma})$
returns a random vector ${\mu}$ such that 
$\mathbb{E}\left[\left({\mu}-\mathbb{E}[P]\right)^{2}\right] \leq {\sigma}^{2},$
with a cost $\widetilde{\mathcal{O}}\left(C{\sigma}^{-p}\right)$
for any ${\sigma}>0$, where $C$ contains some coefficients we are interested in
and $p\in(0,2)$ is a constant.
Then there exists an algorithm $\widetilde{\mathcal{A}}({\tilde{\sigma}})$
that outputs an unbiased estimator $\tilde{\mu}$ of $P$ satisfying
$\operatorname{Var}(\tilde{\mu}) \leq {\tilde{\sigma}}^{2}$,
with expected cost 
$\widetilde{\mathcal{O}}\left(C{\tilde{\sigma}}^{-p}\right)$.
\end{theorem}

\begin{proof}
  Choose $\rho\in\left(\frac{1}{2},\frac{1}{p}\right)$, $M^2>2+\frac{16}{1-2^{1-2\rho}}$ and consider the following algorithm.
\begin{algorithm}[H]
  \begin{algorithmic}[1]
\caption{\algname{$\widetilde{\mathcal{A}}(\cdot)$: Unbiased Estimator Construction}}
\label{alg:unbiased estimator}
\item \textbf{Input}: A biased estimator $\mathcal{A}(\cdot)$,
target variance ${\tilde{\sigma}}^{2}$

\textbf{Output}: An unbiased estimator $\tilde{\mu}$ of $P$ with 
        $\operatorname{Var}(\tilde{\mu}) \leq {\tilde{\sigma}}^{2}$

\STATE  Set $\tilde{\mu}_0 \leftarrow \mathcal{A}({\tilde{\sigma}}/M)$

\STATE Randomly sample $j \sim \mathrm{Geom}\left(\frac12\right)$
\STATE Compute 
      $ \tilde{\mu}_{j} \leftarrow 
      \mathcal{A}\left(2^{-\rho j}\frac{{\tilde{\sigma}}}{M}\right)$
 and
   $   \tilde{\mu}_{j-1} \leftarrow 
      \mathcal{A}\left(2^{-\rho(j-1)}\frac{{\tilde{\sigma}}}{M}\right)$
\STATE \textbf{Return}  $\tilde{\mu}= 
      \tilde{\mu}_{0} + 2^{j}(\tilde{\mu}_{j}-\tilde{\mu}_{j-1})$
\end{algorithmic}
\end{algorithm}

Below, we show that the algorithm constructed in this way satisfies the conclusion.

First, for any $j \geq 1$,
\begin{align*}
\mathbb{E}\!\left[|\tilde{\mu}_j-\tilde{\mu}_{j-1}|\right]
&\leq\left(\mathbb{E}\!\left[|\tilde{\mu}_j-\tilde{\mu}_{j-1}|^{2}\right]\right)^{\frac{1}{2}}  \\
&\leq\left(2\,\mathbb{E}\!\left[(\tilde{\mu}_j-\mathbb{E}[P])^{2}\right]
+2\,\mathbb{E}\!\left[(\tilde{\mu}_{j-1}-\mathbb{E}[P])^{2}\right]\right)^{\frac{1}{2}} \\
&\leq\left(2^{-2\rho j}\frac{\tilde{\sigma}^2}{M^2}+2^{-2\rho (j-1)}\frac{\tilde{\sigma}^2}{M^2}\right)^{\frac{1}{2}}
=\sqrt{1+2^{2\rho}}\,\frac{\tilde{\sigma}}{M}\,2^{-\rho j}.
\end{align*}

Hence,
$$
\sum_{j=1}^{\infty}\mathbb{E}\!\left[|\tilde{\mu}_j-\tilde{\mu}_{j-1}|\right]< \infty.
$$

If $J \sim \mathrm{Geom}\!\left(\frac12\right)$,
then $\mathbb{P}[J=j]=2^{-j}$.
It follows that
\begin{align*}
\mathbb{E}[\tilde{\mu}]
&=\mathbb{E}\!\left[\tilde{\mu}_{0}+2^{J}(\tilde{\mu}_{J}-\tilde{\mu}_{J-1})\right]  \\
&=\mathbb{E}[\tilde{\mu}_{0}]+\sum_{j=1}^\infty\mathbb{P}[J=j]2^j
\left(\mathbb{E}[\tilde{\mu}_{j}]-\mathbb{E}[\tilde{\mu}_{j-1}]\right)  \\
&=\mathbb{E}[\tilde{\mu}_{0}]+\sum_{j=1}^\infty
\left(\mathbb{E}[\tilde{\mu}_{j}]-\mathbb{E}[\tilde{\mu}_{j-1}]\right)\\
&=\lim_{j\to\infty}\mathbb{E}[\tilde{\mu}_j]
=\mathbb{E}[P].
\end{align*}

Therefore, the estimator $\tilde{\mu}$ is unbiased.

As for the variance,
\begin{align*}
  \operatorname{Var}(\tilde{\mu})
  =\mathbb{E}\left[\left(\tilde{\mu}-\mathbb{E}[P]\right)^{2}\right] 
  \leq 2\mathbb{E}\left[\left(\tilde{\mu}-\tilde{\mu}_0\right)^{2}\right]
  +2\mathbb{E}\left[\left(\tilde{\mu}_0-\mathbb{E}[P]\right)^{2}\right].
\end{align*}
By definition,
\begin{align*}
  \mathbb{E}\left[\left(\tilde{\mu}-\tilde{\mu}_0\right)^{2}\right]
  =&\sum_{j=1}^\infty \mathbb{P}[J=j] 2^{2j}\mathbb{E}\left[\left(\tilde{\mu}_j-\tilde{\mu}_{j-1}\right)^{2}\right]\\
\leq&\sum_{j=1}^\infty 2^{j+1}\left(2^{-2\rho j}\frac{{\tilde{\sigma}}^2}{M^2}+2^{-2\rho{(j-1)}}\frac{{\tilde{\sigma}}^2}{M^2}\right)\\
\leq& \frac{4\tilde{\sigma}^2}{M^2}\frac{2^{-2\rho}+1}{1-2^{1-2\rho}}.
\end{align*}
  Hence
     \begin{align*}
       \operatorname{Var}(\tilde{\mu})\leq  \frac{2\tilde{\sigma}^2}{M^2}+ \frac{8\tilde{\sigma}^2}{M^2}\frac{2^{-2\rho}+1}{1-2^{1-2\rho}}\leq \tilde{\sigma}^2.
     \end{align*}
    The expected cost is 
    \begin{align*}
      \widetilde{\mathcal{O}}\left(C{\tilde{\sigma}}^{-p}\cdot
      \left(1+\sum_{j=1}^\infty \mathbb{P}[J=j] \left(2^{\rho p j}+2^{\rho p(j-1)}\right)\right)\right)
      =  \widetilde{\mathcal{O}}\left(C{\tilde{\sigma}}^{-p}\right).
    \end{align*}
\end{proof}

Building on the above construction,
we further strengthen Theorem~\ref{thm:QA-MLMC}
by constructing an unbiased estimator.
In addition, we explicitly characterize the sampling complexity
for $r$-dimensional random variables,
clarify its dependence on relevant parameters,
and derive a bound on the mean square error.

\begin{theorem}
\label{thm:QMLMC2}
Let $P$ be a $r$-dimensional random variable, and let 
$P_\ell~(\ell=0,1,\ldots)$ be a sequence of $r$-dimensional random variables 
approximating $P$ at level $\ell$. 
Further define $P_{-1}=0$.
Assume that obtaining one sample from $P_\ell$ requires $\mathcal{C}_\ell$ queries to $O_P$.
Suppose there exist positive constants $\alpha,\beta,\gamma$ and $K_1,K_2,K_3$
such that $\alpha\geq\min\{\beta/2,\gamma\}$ and
\begin{itemize}
\item $|\E[P_\ell-P]| =\widetilde{\mathcal{O}} \left( K_1 2^{-\alpha \ell}\right)$,
\item  $V_\ell=\operatorname{Var}(P_\ell-P_{\ell-1}) =\widetilde{\mathcal{O}}\left( K_2 2^{-\beta \ell}\right)$,
\item $\mathcal{C}_\ell  =\widetilde{\mathcal{O}} \left(K_3 2^{\gamma \ell}\right)$,
\end{itemize}
where $K_1,K_2,K_3$ are parameters of interest.
Then we can obtain an estimator $\hat{\mu}$ of the expectation $\E[P]$ such that  
 $\E\left[(\hat{\mu}-\E[P])^2\right] \leq \epsilon^2$
with  expected total cost
    \begin{align*}
        \begin{cases} 
       \widetilde{\mathcal{O}}\left(r^{\frac{1}{2}}K_2^{\frac{1}{2}}K_3 \cdot\epsilon^{-1} \right), & \beta \geq 2\gamma, \\
    \widetilde{\mathcal{O}}\left(r^{\frac{1}{2}}K_1^{(\gamma-\frac{1}{2}\beta)/\alpha}K_2^{\frac{1}{2}}K_3\cdot \epsilon^{-1-(\gamma-\frac{1}{2}\beta)/\alpha} \right)  , & \beta < 2\gamma.
        \end{cases}
    \end{align*}

Furthermore, if $\beta + 2\alpha > 2\gamma$, there exists an unbiased estimator
$\tilde{\mu}$ of $P$ satisfying $\operatorname{Var}(\tilde{\mu}) \leq \epsilon^2$, 
and its expected computational cost remains identical to the one stated above.

In this theorem, the $\widetilde{\mathcal{O}}$ notation may additionally hide polynomial factors in $\ell$.
\end{theorem}

\begin{proof}
For a truncation level \(L\) to be chosen later, let
\(\hat{\mu}\) denote the QA-MLMC estimator for \(\mathbb E[P_L]\).
Assume that there exists a constant $C(\ell)$, which may depend polynomially on $\ell$, such that
$|\E[P_\ell]-\E[P]| \leq C(\ell) K_1 2^{-\alpha \ell}.$
Then
   \begin{align*}
    \E\left[ (\hat{\mu}-\E[P])^2  \right]
    =\operatorname{Var}(\hat{\mu})+\left(\E[P_L]-\E[P]\right)^2
    \leq \operatorname{Var}(\hat{\mu})+C(\ell)^2K_1^22^{-2\alpha \ell},
   \end{align*}
so we may choose $L=\widetilde{\mathcal{O}}\left(\mathrm{log} \frac{2K_1^2}{\epsilon^2}\right) $ 
    such that $C(L)^2K_1^22^{-2\alpha L}\leq \frac{\epsilon^2}{2}$.
    Thus, it remains to construct an unbiased estimator $\hat{\mu}$ for
    $P_L$ satisfying $\operatorname{Var}(\hat{\mu})\leq \frac{\epsilon^2}{2}.$

 Let $\{\epsilon_\ell\}_{\ell=0}^L$ be positive parameters, to be
chosen below, which specify the target standard deviations of the
level-wise quantum mean estimators.
For $\ell = 0,1,\dots,L$, by Theorem~\ref{thm:Unbiased Quantum Mean Estimation},
we obtain an unbiased estimator $\hat{\mu}_\ell$ of $P_\ell - P_{\ell-1}$ using 
$\widetilde{\mathcal{O}}\left(r^{\frac{1}{2}}\sqrt{V_\ell}/\epsilon_\ell\right)$  
queries to $O_P$, and it satisfies
 $\operatorname{Var}(\hat{\mu}_\ell)\leq \epsilon_\ell^2$.
Then $\hat{\mu} = \sum\limits_{\ell=0}^{L} \hat{\mu}_\ell$ is an unbiased estimator of $\mathbb{E}[P_L]$.

Moreover,
\begin{align*}
    \mathbb{E}\left[(\hat{\mu}-\mathbb{E}[P_L])^2\right]^{\frac{1}{2}}
    = 
    \mathbb{E}\left[\left(\sum_{\ell=0}^{L}
    (\hat{\mu}_\ell-\mathbb{E}[P_\ell-P_{\ell-1}])\right)^2\right]^{\frac{1}{2}}
    \leq 
    \sum_{\ell=0}^{L}
    \mathbb{E}\left[(\hat{\mu}_\ell-\mathbb{E}[P_\ell-P_{\ell-1}])^2\right]^{\frac{1}{2}}
    \leq 
    \sum_{\ell=0}^{L} \epsilon_\ell,
\end{align*}
that is, the variance of $\hat{\mu}$ is bounded by 
$\left(\sum\limits_{\ell=0}^{L} \epsilon_\ell\right)^2$.

The total number of queries to $O_P$ is

\begin{align*}
    \widetilde{\mathcal{O}}\left(
    \sum_{\ell=0}^{L} \frac{r^{\frac{1}{2}}\sqrt{V_\ell}}{\epsilon_\ell} \mathcal{C}_\ell
    \right)
    =
    \widetilde{\mathcal{O}}\left(
    r^{\frac{1}{2}} K_2^{\frac{1}{2}} K_3 
    \sum_{\ell=0}^{L} 
    \frac{1}{\epsilon_\ell} 2^{(\gamma - \beta/2)\ell}
    \right).
\end{align*}

    (1)$\beta\textgreater 2\gamma$
    
    Choose $\epsilon_\ell=\frac{\epsilon}{2}(1-2^{-(\frac{1}{2}\beta-\gamma)/2} ) 2^{-(\frac{1}{2}\beta-\gamma)\ell/2}$. 
   Then the standard deviation of $\hat{\mu}$ can be bounded as
    \begin{align*}
       \sum_{\ell=0}^{L} \epsilon_\ell
       =\sum\limits_{\ell=0}^L \frac{\epsilon}{2}(1-2^{-(\frac{1}{2}\beta-\gamma)/2} ) 2^{-(\frac{1}{2}\beta-\gamma)\ell/2}
       =\frac{\epsilon}{2}(1-2^{-(\frac{1}{2}\beta-\gamma)L/2})\leq\frac{\epsilon}{2}.
    \end{align*}
   Hence, the variance of $\hat{\mu}$ is at most $\frac{\sigma^2}{2}$.

    The total cost is
    \begin{align*}
        C=&\widetilde{\mathcal{O}}\left(r^{\frac{1}{2}}K_2^{\frac{1}{2}}K_3\sum_{\ell=0}^L 2^{-(\frac{1}{2}\beta-\gamma) \ell}/\epsilon_\ell \right)\\
        =&\widetilde{\mathcal{O}}\left(r^{\frac{1}{2}}K_2^{\frac{1}{2}}K_3\sum_{\ell=0}^L 2^{-(\frac{1}{2}\beta-\gamma) \ell} \cdot \sigma^{-1}(1-2^{-(\frac{1}{2}\beta-\gamma)/2} ) ^{-1}2^{(\frac{1}{2}\beta-\gamma)\ell/2} \right)\\
        =&\widetilde{\mathcal{O}}\left(r^{\frac{1}{2}}K_2^{\frac{1}{2}}K_3\epsilon^{-1}\sum_{\ell=0}^L  2^{-(\frac{1}{2}\beta-\gamma)\ell/2} \right)\\
        =&\widetilde{\mathcal{O}}\left(r^{\frac{1}{2}}K_2^{\frac{1}{2}}K_3\epsilon^{-1}\frac{1-2^{-(\frac{1}{2}\beta-\gamma)L/2}}{1-2^{-(\frac{1}{2}\beta-\gamma)/2}} \right)
        =\widetilde{\mathcal{O}}\left(r^{\frac{1}{2}}K_2^{\frac{1}{2}}K_3 \epsilon^{-1} \right).
    \end{align*}
    (2)$\beta =2\gamma$
    
    Choose $\epsilon_\ell=\frac{\epsilon}{2(L+1)}$. 
   Then the standard deviation of $\hat{\mu}$ 
  is bounded by
$ \sum\limits_{\ell=0}^L \frac{\epsilon}{2(L+1)}=\frac{\epsilon}{2}$
    and the total cost is
    \begin{align*}
        C=&\widetilde{\mathcal{O}}\left(r^{\frac{1}{2}}K_2^{\frac{1}{2}}K_3\sum_{\ell=0}^L 2^{-(\frac{1}{2}\beta-\gamma) \ell}/\epsilon_\ell \right)\\
        =&\widetilde{\mathcal{O}}\left(r^{\frac{1}{2}}K_2^{\frac{1}{2}}K_3\sum_{\ell=0}^L  2(L+1)/\epsilon \right)\\
        =&\widetilde{\mathcal{O}}\left(r^{\frac{1}{2}}K_2^{\frac{1}{2}}K_3\epsilon^{-1} \right).
    \end{align*}
    (3)$\beta \textless 2\gamma$
    
    Choose $\epsilon_l = \frac{\epsilon}{2}\, 2^{-(\gamma-\frac{1}{2}\beta)L/2}
      \left(2^{(\gamma-\frac{1}{2}\beta)/2}-1\right)\,2^{(\gamma-\frac{1}{2}\beta)\ell/2}$. 
       Then the standard deviation of $\hat{\mu}$ is controlled by
    \begin{align*}
       &\sum\limits_{\ell=0}^L \frac{\epsilon}{2}\, 2^{-(\gamma-\frac{1}{2}\beta)L/2}
        \left(2^{(\gamma-\frac{1}{2}\beta)/2}-1\right)\, 2^{(\gamma-\frac{1}{2}\beta)\ell/2}\\
        =&\frac{\epsilon}{2}\, 2^{-(\gamma-\frac{1}{2}\beta)L/2}
        \, \left(2^{(\gamma-\frac{1}{2}\beta)L/2}-1\right)\\
        =&\frac{\epsilon}{2}\, 
        \left(1-2^{-(\gamma-\frac{1}{2}\beta)L/2}\right)< \frac{\epsilon}{2}.
    \end{align*}
    
    The total cost is
    \begin{align*}
        C=&\widetilde{\mathcal{O}}\left(r^{\frac{1}{2}}K_2^{\frac{1}{2}}K_3\sum_{\ell=0}^L 2^{(\gamma-\frac{1}{2}\beta) \ell}/\epsilon_\ell \right)\\
        =&\widetilde{\mathcal{O}}\left(r^{\frac{1}{2}}K_2^{\frac{1}{2}}K_3\sum_{\ell=0}^L 2^{(\gamma-\frac{1}{2}\beta) \ell} \cdot\epsilon^{-1} 2^{(\gamma-\frac{1}{2}\beta)L/2}
        \cdot 2^{-(\gamma-\frac{1}{2}\beta)\ell/2} \right)\\
        =&\widetilde{\mathcal{O}}\left(r^{\frac{1}{2}}K_2^{\frac{1}{2}}K_3\epsilon^{-1}\sum_{\ell=0}^L 2^{(\gamma-\frac{1}{2}\beta) \ell/2} \cdot 2^{(\gamma-\frac{1}{2}\beta)L/2} \right)\\
       =&\widetilde{\mathcal{O}}\left(r^{\frac{1}{2}}K_2^{\frac{1}{2}}K_3\epsilon^{-1}\cdot2^{(\gamma-\frac{1}{2}\beta)L} \right).
 \end{align*}
 Since $L$ is chosen such that 
$C(L)^2 K_1^2 2^{-2\alpha L} \leq \frac{\epsilon^2}{2}$, we have
\begin{align*}
    2^{(\gamma-\frac{1}{2}\beta)L}=\widetilde{\mathcal{O}}\left(
    K_1^{(\gamma-\frac{1}{2}\beta)/\alpha}
\cdot\epsilon^{-(\gamma-\frac{1}{2}\beta)/\alpha}
    \right).
\end{align*}
The total cost
       \begin{align*}
       C =&\widetilde{\mathcal{O}}\left(r^{\frac{1}{2}}K_1^{(\gamma-\frac{1}{2}\beta)/\alpha}K_2^{\frac{1}{2}}K_3\epsilon^{-1-(\gamma-\frac{1}{2}\beta)/\alpha} \right).
    \end{align*}
   
    By applying Theorem~\ref{thm:unbiased estimator construct}, we can derive the further result.
\end{proof}

Finally, by applying the following lemma, we know that any estimator 
whose variance is bounded by $\epsilon^{2}$ is also an $\epsilon'$-additive-error 
estimator with success probability $0.99$, where $\epsilon$ and $\epsilon'$ 
differ only by a constant factor.

\begin{lemma}\label{lem:variance and epsilon}
    If we obtain an estimator $\hat{\mu}$ of $P$ such that
    \begin{align*}
         \mathbb{E}\big[ (\hat{\mu}-\mathbb{E}[P])^2 \big]\leq \epsilon^2,
    \end{align*}
    then $\hat{\mu}$ is an estimate of $\mathbb{E}[P]$ with additive error 
    $\epsilon' = 10\epsilon$ and success probability at least $0.99$.
\end{lemma}

\begin{proof}
    Since 
   \begin{align*}
        \epsilon^2\geq \E\left[ (\hat{\mu}-\E[P])^2  \right]\geq \epsilon'^2 \mathbb{P}\left[|\hat{\mu}-\E[P]|\geq\epsilon'\right],
    \end{align*}
    setting $\epsilon' = 10\epsilon$ yields 
    $ \mathbb{P}\left[|\hat{\mu}-\E[P]|\geq\epsilon'\right]\leq 0.01$.
    Therefore, $\hat{\mu}$ is an estimator of $\mathbb{E}[P]$ with additive error
    $\epsilon' = 10\epsilon$ and success probability at least $0.99$.
\end{proof}

\subsection{Markov Chain Monte Carlo}
\label{subsec:Markov Chain Monte Carlo}
Markov Chain Monte Carlo (MCMC) is a class of algorithms for 
sampling from complex probability distributions, 
particularly when direct sampling is infeasible. 
It has become a fundamental tool in Bayesian inference, 
statistical physics, and machine learning due to its flexibility and asymptotic correctness.

The key idea of MCMC is to construct a Markov chain whose stationary distribution
coincides with the target. 
By simulating the chain over time, 
we can generate approximate samples 
and use them to estimate expectations, probabilities, or other statistical quantities.
There are various ways to construct such a chain, 
including the Metropolis-Hastings algorithm~\cite{Borkar1953EquationOS,Hastings1970MonteCS} 
and the Gibbs sampler~\cite{Geman1984StochasticRG}.

Markov Chain Monte Carlo  methods can also be accelerated using quantum algorithms.
\cite{Temme2009QuantumMS}, extended the classical Metropolis sampling algorithm to the quantum setting, 
and there has also been extensive research on quantum Gibbs sampling~\cite{Bouland2023QuantumSF,Hwang2024GibbsSP}.

Another important class of sampling methods is based on stochastic differential equations. 
In particular, the overdamped Langevin diffusion, described by the SDE~\eqref{SDE0}, is widely used for sampling from log-concave distributions. 
Algorithm~\ref{alg:Unadjusted Langevin Algorithm} approximates the solution to \eqref{SDE0}.
\begin{algorithm}[H]
    \begin{algorithmic}[1]
      \caption{\algname{Overdamped Langevin Diffusion}}
      \label{alg:Unadjusted Langevin Algorithm}
      \item \textbf{Input}: {Step size $h$, initial sample $x_0 \sim \rho_0$}\\
       \textbf{Output}: {Sequence $x_1,x_2,\cdots$}\\
   \STATE \textbf{For} $n=1,2,\cdots$
   \begin{itemize}
   \item Sample $\xi_n\sim \mathcal{N}(0,I)$
   \item $x_{n+1}=x_n-h\nabla f(x_n)+\sqrt{2h}\xi_{n}$
   \end{itemize}
   \textbf{End for}
  
    \end{algorithmic}
 \end{algorithm}

An alternative method is the underdamped Langevin diffusion, which is defined as follows:
\begin{align*}\label{eqn:underdamped-Langevin}
\mathrm{d} V_t &= - \gamma V_t\mathrm{d} t - u\nabla f(X_t)\mathrm{d} t + \sqrt{2\gamma u}\mathrm{d} W_t, \\
\mathrm{d} X_t &= V_t\mathrm{d} t,
\end{align*}
where $\gamma > 0$ and $u > 0$ are parameters, and $(W_t)_{t \geq 0}$ is a standard Brownian motion.

Moreover, the overdamped Langevin diffusion can be recovered from this system by taking the limit $\gamma \to \infty$ and rescaling time as $t \mapsto t/\gamma$.

%For $(X_t)_{t\geq0}$ which evolves following the Langevin diffusion (\ref{SDE}),
%let $\rho_t$ denote the probability distribution of $X_t$.
%Then $(\rho_t)_{t\geq 0}$ following the Fokker-Planck equation:
%\begin{align*}
%\frac{\partial \rho_t}{\partial t} = \nabla \cdot (\rho_t \nabla f) + \Delta \rho_t = \nabla \cdot \left( \rho_t \nabla \mathrm{log} \frac{\rho_t}{\nu} \right).
%\end{align*}
%where $\nabla\cdot$ is the divergence and $\Delta $ is the Laplacian operator.

\section{Quantum Algorithms for Gibbs Expectation under Dissipativity}
\label{sec:dissipativity}

In this section, we develop quantum algorithms for Gibbs expectation
under assumptions weaker than strong convexity. Previous quantum
algorithms for sampling potential functions and estimating normalizing
constants often require strong convexity of the potential. More
specifically, \cite{Childs2022QuantumAF} proposes three estimation
methods with complexity \(\widetilde{\mathcal O}(\epsilon^{-1})\), all
of which require \(f\) to be \(L\)-smooth and \(\mu\)-strongly convex,
i.e., \(\mu,L>0\) and for any \(x\neq y\in\mathbb R^d\),
\[
\frac{
f(y)-f(x)-\langle \nabla f(x),y-x\rangle
}{
\|x-y\|_2^2/2
}
\in[\mu,L].
\]
This assumption provides a contraction mechanism for the Langevin
dynamics and is useful for controlling both the long-time error and the
distance between coupled trajectories.

Our goal is to replace strong convexity by dissipativity, together with
a weaker one-sided Lipschitz condition. Dissipativity is enough to ensure
long-time stability of the Langevin diffusion, but it does not imply that
two trajectories driven by the same Brownian motion contract toward each
other. Consequently, a standard fine--coarse coupling may fail to provide
a sufficiently small multilevel variance over long time horizons. To
address this issue, we introduce a spring term between the fine and
coarse paths, following the change-of-measure approach of
\cite{Fang2018MultilevelMC}. The spring term keeps the two paths close
under an auxiliary measure, while the Radon--Nikodym derivative converts
the resulting estimator back to the original Langevin dynamics.

There is another, independent issue specific to obtaining the optimal
QA-MLMC precision scaling. In the time-discretized Langevin setting, one
level-\(\ell\) sample has cost proportional to \(h_\ell^{-1}\). Hence the
cost exponent is \(\gamma=1\). The QA-MLMC theorem gives the optimal
\(\widetilde{\mathcal O}(\epsilon^{-1})\) scaling when the level variance
decays as
\[
\operatorname{Var}(\Delta_\ell)=\mathcal O(h_\ell^2),
\]
that is, when \(\beta=2\). However, under only \(L\)-smoothness, the
standard Euler--Maruyama coupling typically gives strong order \(1/2\),
and therefore only
\[
\operatorname{Var}(\Delta_\ell)=\mathcal O(h_\ell).
\]
This is sufficient for classical MLMC to achieve the usual
\(\mathcal O(\epsilon^{-2})\) complexity, but it is not sufficient for
the optimal \(\widetilde{\mathcal O}(\epsilon^{-1})\) QA-MLMC scaling in
our framework.

For this reason, we impose \(L\)-Hessian smoothness, as in
Definition~\ref{def:L-Hessian-smooth}. This assumption is used to prove
a strong order-one estimate for the coupled level differences, yielding
\[
\operatorname{Var}(\Delta_\ell)=\mathcal O(h_\ell^2).
\]
Together with the cost scaling \(C_\ell=\mathcal O(h_\ell^{-1})\), this
places the estimator in the regime where the QA-MLMC theorem yields the
desired \(\widetilde{\mathcal O}(\epsilon^{-1})\) dependence on the
target accuracy. Thus, Hessian smoothness provides a natural sufficient regularity
condition under which the coupled level differences have strong order one,
allowing the QA-MLMC estimator to achieve the optimal
\(\widetilde{\mathcal O}(\epsilon^{-1})\) dependence on the target
accuracy.

To begin with, we consider the one-sided Lipschitz setting and introduce
a QA-MLMC framework for estimating Gibbs expectations. We recall that
the observable $\varphi:\mathbb R^d\to\mathbb R$ is assumed throughout
the paper to be globally Lipschitz.

\begin{theorem}\label{thm:QA-MLMC of one-sided lipschitz}
Assume that $ f $ is $ L $-smooth, $L$-Hessian smooth, and satisfies the one-sided Lipschitz condition with constant $ m $.  
Then, by applying the unbiased QA-MLMC with Theorem~\ref{thm:QMLMC2},
we obtain an unbiased estimator of $\E\left[\varphi(X_T)\right]$ whose variance is bounded by $\epsilon^2$,
where $T=\mathcal{O}(\log \epsilon^{-1})$.  
The total computational complexity is 
$$\widetilde{\mathcal{O}}\left( \frac{L^2\sqrt{d}+Ld}{m^2 }\epsilon^{-1} \right).$$
\end{theorem}
\begin{proof}
Since \(f\) satisfies the one-sided Lipschitz condition with constant \(m\), the Langevin diffusion converges exponentially fast to \(\pi\), so choosing
\[
T=\mathcal O\!\left(m^{-1}\log(\epsilon^{-1})\right)
\]
ensures
\[
\left|\mathbb E_{\pi}[\varphi]-\mathbb E[\varphi(X_T)]\right|
\leq \frac{\epsilon}{2}.
\]

By Corollary~\ref{cor: p moment distance of one-sided}, for $0<h\leq h_{(2)}$, we have
$$
\E\left[\|\widehat{X}_{2N}^f- \widehat{X}_{2N}^c\|^2 \right]
\leq
4 C_{(3)}^2 h^{2}.
$$
Therefore, in Theorem~\ref{thm:QMLMC2}, we take 
$P = \varphi(X_T)$, and let $P_\ell$ denote $\varphi(X_T)$ 
computed using the Euler--Maruyama method with step size $h_{(2)}2^{-\ell}$. Then,
\begin{align*}
  \left|\E\left[P_\ell-P\right]\right|=\widetilde{\mathcal{O}}(C(L,m,d,T)2^{-\ell}),\quad
  \operatorname{Var}(P_\ell-P_{\ell-1})=\widetilde{\mathcal{O}}\left(\frac{(L^2\sqrt{d}+Ld)^2}{m^2}h_{(2)}^22^{-2\ell}\right),\quad
  \mathcal{C}_\ell=\widetilde{\mathcal{O}}\left(\frac{T}{h_{(2)}}2^{\ell}\right).
\end{align*}
Here, $C(L,m,d,T)$ denotes a constant that depends at most polynomially on $L$, $m$, $d$, and $T$. 
Noting that in Theorem~\ref{thm:QMLMC2}, when $\beta \geq 2\gamma$, 
the contribution of $C(L,m,d,T)$ to the overall complexity is at most logarithmic, 
i.e., of order $\log C(L,m,d,T)$, 
we conclude that this dependence is negligible.

Therefore,
we obtain an unbiased estimator of $\E\left[\varphi(X_T)\right]$ 
whose variance is bounded by $\epsilon^2$.
The expected total complexity is 
$$\widetilde{\mathcal{O}}\left( \frac{L^2\sqrt{d}+Ld}{m^2 }\epsilon^{-1} \right).$$
\end{proof}

The one-sided Lipschitz condition 
still excludes a broad class of SDEs that only satisfy the dissipativity condition.  In such cases, even when the fine and coarse trajectories are driven by the same Brownian motion,  
their variance may still grow exponentially with respect to $T$.  

In what follows, we consider algorithms under the weaker one-sided Lipschitz and dissipativity conditions.  
Since the standard one-sided Lipschitz condition can be viewed as a special case of the weaker one-sided Lipschitz assumption,  
we will not discuss the algorithm under the one-sided Lipschitz case in detail.

\subsection{Quantum-Accelerated Multilevel Monte Carlo with Change of Measure}
\label{subsec:Quantum-Accelerated Multilevel Monte Carlo with Change of Measure}
Building upon the change-of-measure approach with a spring coefficient $S>0$
introduced in~\cite{Fang2018MultilevelMC}, we adapt this framework
to the QA-MLMC for Gibbs expectation.

The motivation for adding a spring term stems from the weaker one-sided Lipschitz assumption as defined in Definition~\ref{def:Weaker One-sided Lipschitz}.

Note that if $f$ is $L$-smooth, it can be shown that it also satisfies the weaker one-sided Lipschitz condition, since
\begin{align*}
\langle \nabla f(x) - \nabla f(y), x - y \rangle 
&\geq -\|\nabla f(x) - \nabla f(y)\| \cdot \|x - y\| \\
&\geq -L \|x - y\|^2.
\end{align*}

Consider the coupling under the SDEs with an added "spring" term:
\begin{align*}
\mathrm{d} Y_t^f &= S(Y_t^c-Y_t^f)\mathrm{d} t  -\nabla f(Y_t^f)\mathrm{d} t + \sqrt{2}\mathrm{d} W_t^{{\mathbb{P}}} ,\\
\mathrm{d} Y_t^c &= S(Y_t^f-Y_t^c)\mathrm{d} t -\nabla f(Y_t^c)\mathrm{d} t +  \sqrt{2}\mathrm{d} W_t^{{\mathbb{P}}} .
\end{align*}
where the spring coefficient  satisfies $S>\frac{\lambda}{2}$ 
for both fine path and coarse path for all $\ell>1$. 

Under these modified SDEs, the difference between the two paths evolves according to:
\begin{align*}
\mathrm{d} (Y_t^f- Y_t^c)= 2S(Y_t^c-Y_t^f)\mathrm{d} t - (\nabla f(Y_t^f)-\nabla f(Y_t^c))\mathrm{d} t.
\end{align*}
Provided $S>\frac{\lambda}{2}$, 
applying Ito's formula and the one-sided Lipschitz condition  yields:
\begin{align*}
\mathrm{d} \|Y_t^f- Y_t^c\|^2 \leq 2(\lambda-2S)\|Y_t^f- Y_t^c\|^2\mathrm{d} t,
\end{align*} 
which indicates that the distance between the fine and coarse paths contracts exponentially over time.

On the other hand, define the following Brownian motions:
\begin{align*}
\sqrt{2}\mathrm{d} W_t^{\mathbb{Q}^f} &= S(Y_t^c-Y_t^f)\mathrm{d} t+  \sqrt{2}\mathrm{d} W_t^{\mathbb{P}},\\
\sqrt{2}\mathrm{d} W_t^{\mathbb{Q}^c} &= S(Y_t^f-Y_t^c)\mathrm{d} t + \sqrt{2} \mathrm{d} W_t^{\mathbb{P}},
\end{align*}
 so that $W_t^{\mathbb{Q}^f}$ and $W_t^{\mathbb{Q}^c}$ 
become standard Brownian motions for the following original SDEs:
\begin{align*}
\mathbb{Q}^f: \mathrm{d} X_t^f &=-\nabla f(X_t^f)\mathrm{d} t + \sqrt{2}\mathrm{d} W_t^{\mathbb{Q}^f},\\
\mathbb{Q}^c:\mathrm{d} X_t^c &=-\nabla f(X_t^c)\mathrm{d} t +\sqrt{2} \mathrm{d} W_t^{\mathbb{Q}^c}.
\end{align*}
Note that under measure $P$, the distribution of $X^f_t$ is the same as $Y^f_t$ under measure ${\mathbb{Q}}^f$. 
Combined with Girsanov's theorem, we have:
\begin{align*}
    \E^{\mathbb{P}}\left[ \varphi(X^f_{t}) \right] = \E^{{\mathbb{Q}}^f}\left[\varphi(Y^f_t)\right] = \E^{\mathbb{P}} \left[ \varphi(Y^f_t) \frac{\mathrm{d} {\mathbb{Q}}^f}{\mathrm{d}\mathbb{P}} \right].
\end{align*}
Similarly,
\begin{align*}
    \E^{\mathbb{P}}\left[ \varphi(X^c_{t}) \right] = \E^{{\mathbb{Q}}^c}\left[\varphi(Y^c_t)\right] = \E^{\mathbb{P}} \left[ \varphi(Y^c_t) \frac{\mathrm{d} {\mathbb{Q}}^c}{\mathrm{d}\mathbb{P}} \right].
\end{align*}
where $\frac{\mathrm{d} \mathbb{Q}^f}{\mathrm{d} \mathbb{P}}$ and 
$\frac{\mathrm{d} \mathbb{Q}^c}{\mathrm{d} \mathbb{P}}$ 
are the corresponding Radon-Nikodym derivatives of the measure 
respectively, with respect to the common measure $\mathbb{P}$ 
under which both paths are simulated.

Furthermore, we consider the discretized setting.

For level 0, the numerical estimator is the same as the standard MLMC $\varphi(\widehat{X}_T^0)$.

For level $\ell>1,$ we simulate the SDE with the additional spring terms using 
timestep $h =  2^{-\ell}h_0$ for the fine path and $2h$ for the coarse path.
\begin{itemize}
\item
$\widehat{Y}_{0}^f = \widehat{Y}_{0}^c = X_0$.
\item
At odd timesteps for $n\geq 0,$ we update both paths:
\begin{align*}
 \widehat{Y}_{{2n+1}}^c &=\widehat{Y}_{2n}^c+ S(\widehat{Y}_{2n}^f-\widehat{Y}_{2n}^c)h - \nabla f(\widehat{Y}_{2n}^c) h + \sqrt{2}\Delta W_{2n}^{\mathbb{P}}, \\
\widehat{Y}_{{2n+1}}^f &=\widehat{Y}_{2n}^f+ S(\widehat{Y}_{2n}^c-\widehat{Y}_{2n}^f)h - \nabla f(\widehat{Y}_{2n}^f) h + \sqrt{2}\Delta W_{2n}^{\mathbb{P}}.
\end{align*} 
\item
At even timesteps for $n\geq 0,$ we update 
the spring term and drift term of the fine path, 
but keep both the same for the coarse path:
\begin{align*}
\widehat{Y}_{2n+2}^c &=\widehat{Y}_{{2n+1}}^c+S(\widehat{Y}_{{2n}}^f-\widehat{Y}_{{2n}}^c)h - \nabla f(\widehat{Y}_{{2n}}^c) h +\sqrt{2}\Delta W_{2n+1}^{\mathbb{P}},\\
 \widehat{Y}_{2n+2}^f &=\widehat{Y}_{{2n+1}}^f+ S(\widehat{Y}_{{2n+1}}^c-\widehat{Y}_{{2n+1}}^f)h - \nabla f(\widehat{Y}_{{2n+1}}^f) h +\sqrt{2}\Delta W_{2n+1}^{\mathbb{P}}.
\end{align*} 
\end{itemize}
Equivalently, the coarse path updates can be combined as
\begin{align*}
    \widehat{Y}_{{2n+2}}^c\ =\ \widehat{Y}_{{2n}}^c+  S(\widehat{Y}_{{2n}}^f-\widehat{Y}_{{2n}}^c)2h- \nabla f(\widehat{Y}_{{2n}}^c) 2h+\sqrt{2}\Delta W_{2n}^{\mathbb{P}}+\sqrt{2}\Delta W_{2n+1}^{\mathbb{P}}.
\end{align*}

We now discuss how to perform the change of measure. 
For notational convenience, we temporarily do not distinguish between the fine and coarse paths, and denote the spring term by $\widehat{S}$.
More specifically, consider modifying each timestep of the entire path by introducing a new measure $\widehat{\mathbb{Q}}$. 
Under this measure, the Brownian increments satisfy
\begin{align*}
    \sqrt{2}\Delta W^{\widehat{\mathbb{Q}}}_n = \widehat{S}h + \sqrt{2}\Delta W^{\mathbb{P}}_n
\end{align*}
for $n = 0, 1, \dots, N-1.$

Under measure $\mathbb{P}$,
\begin{align*}
    \widehat{Y}_{{n+1}} &=\widehat{Y}_{n}+ \widehat S h - \nabla f(\widehat{Y}_{n}) h + \sqrt{2}\Delta W_{n}^{\mathbb{P}},
\end{align*}
we obtain
\begin{align*}
    \widehat{Y}_{{n+1}}\sim \mathcal{N}^{\mathbb{P}}(\widehat{Y}_{n}+ \widehat S h - \nabla f(\widehat{Y}_{n}) h,2hI).
\end{align*}
Under measure $\widehat{\mathbb{Q}}_n$,
\begin{align*}
    \widehat{Y}_{{n+1}} &=\widehat{Y}_{n}- \nabla f(\widehat{Y}_{n}) h + \sqrt{2}\Delta W_{n}^{\mathbb{P}},
\end{align*}
and
\begin{align*}
    \widehat{Y}_{{n+1}}\sim \mathcal{N}^{\widehat{\mathbb{Q}}_n}(\widehat{Y}_{n} - \nabla f(\widehat{Y}_{n}) h,2hI).
\end{align*}
The  exact Radon-Nikodym derivative for this single step is
\begin{align*}
    R(\widehat{Y}_{{n+1}},\widehat{Y}_{n}, \widehat{S},h)
    =\frac{\d \widehat{\mathbb{Q}}_n}{\d \mathbb{P}}=\frac{\rho (\widehat{Y}_{{n+1}}|\widehat{Y}_{n} -\nabla f(\widehat{Y}_{n})h, 2h I )}{\rho(\widehat{Y}_{{n+1}}
| \widehat{Y}_{n}+ \widehat{S}h -\nabla f(\widehat{Y}_{n})h , 2h I )},
\end{align*}
where $\rho(x|\mu,\Sigma)$ is the probability density function of $\mathcal N(\mu,\Sigma)$. 
Denote $\mu=\widehat{Y}_{n} -\nabla f(\widehat{Y}_{n})h$,
then
\begin{align*}
    \log R(\widehat{Y}_{{n+1}},\widehat{Y}_{n}, \widehat{S},h)
    =&-\frac{1}{2}\left[ (\widehat{Y}_{n+1}-\mu)^{\top} (2hI)^{-1}(\widehat{Y}_{n+1}-\mu)- (\widehat{Y}_{n+1}-\mu- \widehat{S}h)^{\top} (2hI)^{-1}(\widehat{Y}_{n+1}-\mu-\widehat{S}h)\right]\\
    =&-\frac{1}{2}\langle \widehat{Y}_{n+1}-\mu,  \widehat{S}\rangle+\frac{1}{4h}\| \widehat{S}h\|^2\\
    =&-\frac{1}{2}\langle \widehat{Y}_{n+1}-\widehat{Y}_{n} +\nabla f(\widehat{Y}_{n})h,  \widehat{S}\rangle+\frac{1}{4}\| \widehat{S}\|^2h.
\end{align*}
Hence
\begin{align*}
   R(\widehat{Y}_{{n+1}},\widehat{Y}_{n}, \widehat{S},h)
=& \exp\left(-\frac{1}{2}\left\langle \widehat{Y}_{{n+1}}-\widehat{Y}_{n}+\nabla f(\widehat{Y}_{n})h,\widehat{S}\right\rangle +\frac{1}{4}\|\widehat{S}\|^2 h \right).
\end{align*}
If we express it in terms of $\Delta W_n$, we obtain
\begin{align*}
 R(\widehat{Y}_{{n+1}},\widehat{Y}_{n}, \widehat{S},h)
 =& \exp\left(-\frac{1}{2}\left\langle \widehat{S}h+\sqrt{2}\Delta W_n^\mathbb{P} ,\widehat{S}\right\rangle +\frac{1}{4}\|\widehat{S}\|^2 h \right)\\
 =&\exp\left(-\frac{\sqrt{2}}{2}\langle \Delta W_n^\mathbb{P},\widehat{S}\rangle-\frac{1}{4}\|\widehat{S}\|^2h\right).
\end{align*}

Since $\{\Delta W^{\mathbb{P}}_n\}$ and $\{\Delta W^{\widehat{\mathbb{Q}}}_n\}$ 
are two sets of independent Brownian increments under measures $\mathbb{P}$ and $\widehat{\mathbb{Q}}$ respectively, 
the exact Radon-Nikodym derivative is given by
\begin{align*}
    \frac{\d\widehat{\mathbb{Q}}}{\d\mathbb{P}}
    = \prod_{n=0}^{N-1} R(\widehat{Y}_{n+1}, \widehat{Y}_n, \widehat{S}, h).
\end{align*}

Returning to the coupling between the fine and coarse paths, we introduce two new measures, 
$\widehat{\mathbb{Q}}^f$ and $\widehat{\mathbb{Q}}^c$, corresponding to the fine and coarse paths respectively. 
Their Brownian increments are defined as
\begin{align*}
   \sqrt{2} \Delta W_n^{\widehat{\mathbb{Q}}^f} = \widehat{S}^f_n h +\sqrt{2} \Delta W_n^{\mathbb{P}}, 
    \quad 
   \sqrt{2} \Delta W_n^{\widehat{\mathbb{Q}}^c} = \widehat{S}^c_n h + \sqrt{2}\Delta W_n^{\mathbb{P}},
\end{align*}
where $\widehat{S}^f_n$ and $\widehat{S}^c_n$ denote the spring terms 
at the $n$-th step for the fine and coarse levels. 
The corresponding Radon-Nikodym derivatives can then be computed recursively, 
step by step at the same time as updating the paths:
\begin{itemize}
\item At $t_0,$ we set $R_{0}^f = R_{0}^c = 1$;
\item At odd timesteps for $n\geq 0,$ we only update $R^f$:
\begin{align*}
R_{{2n+1}}^f= R_{{2n}}^f\  R\left(\widehat{Y}_{{2n+1}}^f, \widehat{Y}_{{2n}}^f, S(\widehat{Y}_{{2n}}^c-\widehat{Y}_{{2n}}^f), h\right);
\end{align*}
\item At even timesteps  for $n\geq 0,$ we update both $R^f$ and $R^c$:
\end{itemize}
\begin{align*}
 R_{{2n+2}}^f  =  &  R_{{2n+1}}^f\,  R\left( \widehat{Y}_{{2n+2}}^f ,\widehat{Y}_{{2n+1}}^f,S(\widehat{Y}_{{2n+1}}^c-\widehat{Y}_{{2n+1}}^f),h  \right),\\
R_{{2n+2}}^c   = &  R_{{2n}}^c \,R \left( \widehat{Y}_{{2n+2}}^c ,\widehat{Y}_{{2n}}^c, S(\widehat{Y}_{{2n}}^f-\widehat{Y}_{{2n}}^c), 2h\right).
\end{align*}
Then, after $2N$ steps, we obtain the exact Radon-Nikodym derivatives for the whole path:
\begin{align*}
\frac{\d \widehat{\mathbb{Q}}^f}{\d \mathbb{P}} &= R^f_T = \prod_{n=0}^{2N-1}  R\left(\widehat{Y}_{{n+1}}^f, \widehat{Y}_{{n}}^f,  S(\widehat{Y}_{{n}}^c-\widehat{Y}_{{n}}^f),  h\right),\\
\frac{\d \widehat{\mathbb{Q}}^c}{\d \mathbb{P}}  &= R^c_T = \prod_{n=0}^{N-1}  R \left( \widehat{Y}_{{2n+2}}^c ,\widehat{Y}_{{2n}}^c, S(\widehat{Y}_{{2n}}^f-\widehat{Y}_{{2n}}^c), 2h\right).
\end{align*}
Finally, the multilevel correction estimator becomes 
\begin{align*}
 \varphi(\widehat{Y}_T^f)R^f_T-\varphi(\widehat{Y}_T^c)R^c_T,
\end{align*}
and the identity we use in the new MLMC is 
\begin{align*}
    \E^{{\mathbb{P}}}\left[\varphi(X_T^L)\right] = \E^{{\mathbb{P}}}\left[\varphi(\widehat{X}_T^0)\right] + \sum_{\ell=1}^L \E^{{\mathbb{P}}}\left[\varphi(\widehat{Y}_T^{f,\ell})R^{f,\ell}_T -\varphi(\widehat{Y}_T^{c,\ell})R^{c,\ell}_T\right],  
\end{align*}
where at level 0, the numerical estimator is the same as the standard MLMC.
 The new MLMC estimator becomes
\begin{align*}
\widehat{\varphi}_{new} = & N_0^{-1} \sum_{n=1}^{N_0} \varphi(\widehat{X}_T^{0,(n)})  \\
&+ \sum_{\ell=1}^L N_\ell^{-1} \sum_{n=1}^{N_\ell}\left(  \varphi(\widehat{Y}_T^{f,\ell,(n)})R^{f,\ell,(n)}_T- \varphi(\widehat{Y}_T^{c,\ell,(n)}) R^{c,\ell,(n)}_T\right). 
\end{align*}

As a summary, the level-$\ell$ sampling method under the change of measure
can be written as Algorithm~\ref{alg:classical level-l-mlmc with change of measure}.
\begin{algorithm}[H]
\small{
  \begin{algorithmic}[1]
    \caption{\algname{Level-$\ell$ Sample with Change-of-Measure MLMC}}
    \label{alg:classical level-l-mlmc with change of measure}
    \item \textbf{Input:} Initial state $x_0$,
          step size $h$ and number $N$, spring coefficient $S$

     \textbf{Output:} A level-$\ell$ sample $\Delta_\ell$
 
\STATE Set $y_0^f = y_0^c = x_0$ and $R_0^f = R_0^c = 1$

\STATE \textbf{For} $n=0$ \textbf{to} $N-1$
\begin{tightitemize}
    \item At odd timesteps, update both the coarse and fine paths of the SDE with the spring terms:
        \begin{tightalign*}
        y_{2n+1}^c &= y_{2n}^c + S(y_{2n}^f - y_{2n}^c)h - \nabla f(y_{2n}^c)h + \sqrt{2h} \,\xi_{2n} \\
        y_{2n+1}^f &=y_{2n}^f + S(y_{2n}^c - y_{2n}^f)h - \nabla f(y_{2n}^f)h + \sqrt{2h}\, \xi_{2n}
        \end{tightalign*}
        and update the accumulated Radon–Nikodym derivative at each step:
        \begin{tightalign*}
       R_{2n+1}^f = R_{2n}^f \cdot R(y_{2n+1}^f, y_{2n}^f, S(y_{2n}^c - y_{2n}^f), h)
        \end{tightalign*}
    \item At even timesteps, update the spring term and drift term of the fine
path, but keep both the same for the coarse path:
\begin{tightalign*}
    y_{2n+2}^c &= y_{2n+1}^c + S(y_{2n}^f - y_{2n}^c)h - \nabla f(y_{2n}^c)h + \sqrt{2h} \,\xi_{2n+1} \\
    y_{2n+2}^f &= y_{2n+1}^f + S(y_{2n+1}^c - y_{2n+1}^f)h - \nabla f(y_{2n+1}^f)h + \sqrt{2h} \,\xi_{2n+1}
\end{tightalign*}
and update the Radon–Nikodym derivative accordingly:
\begin{tightalign*}
   R_{2n+2}^f &= R_{2n+1}^f \cdot R(y_{2n+2}^f, y_{2n+1}^f, S(y_{2n+1}^c - y_{2n+1}^f), h) \\
   R_{2n+2}^c &= R_{2n}^c \cdot R(y_{2n+2}^c, y_{2n}^c, S(y_{2n}^f - y_{2n}^c), 2h)
\end{tightalign*}
\end{tightitemize} 
\textbf{End for}

\STATE  \textbf{Return} $\Delta_\ell = \varphi(y_{2N}^f) R_{2N}^f - \varphi(y_{2N}^c) R_{2N}^c$
  \end{algorithmic}
}
\end{algorithm}

\begin{proposition}[Strong Order Error of Discretization]
\label{prop: strong order error of discretization}
Assume that $ f $ is dissipative, $ L $-smooth, $L$-Hessian smooth, and satisfies the weaker one-sided Lipschitz condition with constant $ \lambda $.
Then
as notation above
with $S>\lambda/2$, any
$T>0$ and $p\geq 1$, we have
\begin{align*}
        &\E\left[\left|\varphi(\widehat{Y}_{2N}^f)\frac{\d\widehat{\mathbb{Q}}^f}{\d \mathbb{P}}-\varphi(\widehat{Y}_{2N}^c)\frac{\d\widehat{\mathbb{Q}}^c}{\d \mathbb{P}}\right|^p\right]^{1/p}
    =\mathcal{O}\left(\sqrt{T}\cdot\frac{SLd(L+\sqrt{d})}{2S-\lambda}h\right)
\end{align*}
for $0<h\leq h_{(4)}$, where
   $    h_{(4)}=\mathcal{O}\left(\min\left\{\frac{1}{S},\,\frac{\tilde{\alpha}}{S^2},\,\frac{\tilde{\alpha}}{L^2},\,\frac{\tilde{\alpha}}{2S-\lambda},\,\frac{\min\{1,\tilde\alpha\}}{\|\nabla f(0)\|^2},\,\left(1-\frac{\lambda}{S}\right)^2,\,\frac{(2S-\lambda)^2}{T^2S^2L^2d\,},\,\frac{2S-\lambda}{\sqrt{T}S(L^2+L\sqrt{d})}\,\right\} \right)   .$
\normalsize
\end{proposition}
\begin{proof}
It follows directly from Lemma~\ref{lem:bounds of p moment} - Lemma~\ref{lem:bound of two path with R-D}, which can be found in the Appendix~\hyperref[Appendix A: Strong Error Analysis]{A} .
\end{proof}
%\begin{remark}
%    Note that the weaker one-sided Lipschitz condition can
%    be implied by the $L$-smoothness, so we may omit it in the following.
%\end{remark}

In the following, we aim to accelerate the multilevel Monte Carlo method
in this setting using quantum algorithms and the required oracles have been introduced in Section~\ref{sec:(sub)oracles}.

\begin{algorithm}[H]
\small{
\caption{\algname{Level-$\ell$ Quantum Sample with Change-of-Measure MLMC}}
\label{alg:Level l sample with change-of-measure QA-MLMC}
\begin{algorithmic}[1]

\item \textbf{Inputs}: Oracles $O_I, O_{\nabla f}, O_W, O_R, O_\varphi$,
 step size $h$ and number $N$, spring coefficient $S$

\textbf{Output}: A level-$\ell$ sample state $|\Delta_\ell\rangle$

\STATE Prepare initial states:
$|y_0^f\rangle = |y_0^c\rangle = O_I|0\rangle,\quad |R_0^f\rangle = |R_0^c\rangle = |1\rangle$
\STATE \textbf{For} $n = 0$ \textbf{to} $N-1$

\begin{tightitemize}
    \item Sample seed $s_{2n}$ and query $O_W$:
$|s_{2n}\rangle|0\rangle \mapsto |s_{2n}\rangle |\xi_{s_{2n}}\rangle$
to prepare the stochastic term
\item Query gradients:
$
O_{\nabla f} |y_{2n}^c\rangle|0\rangle = |y_{2n}^c\rangle|\nabla f(y_{2n}^c)\rangle,
 \quad O_{\nabla f} |y_{2n}^f\rangle|0\rangle = |y_{2n}^f\rangle|\nabla f(y_{2n}^f)\rangle
$
\item Apply arithmetic unitary to update the paths of SDEs:
\begin{tightalign*}
|y_{2n+1}^c\rangle &= |y_{2n}^c + S(y_{2n}^f - y_{2n}^c)h - \nabla f(y_{2n}^c)h + \sqrt{2h} \xi_{s_{2n}}\rangle\\
|y_{2n+1}^f\rangle &= |y_{2n}^f + S(y_{2n}^c - y_{2n}^f)h - \nabla f(y_{2n}^f)h + \sqrt{2h} \xi_{s_{2n}}\rangle
\end{tightalign*}
\item Apply $O_R$ to update the Radon–Nikodym derivative of fine path: 
\begin{tightalign*}
    |R_{2n+1}^f\rangle=|R_{2n}^f R(y_{2n+1}^f, y_{2n}^f, S(y_{2n}^c - y_{2n}^f), h)\rangle
\end{tightalign*}
\item Sample seed $s_{2n+1}$ and apply $O_W$:
$|s_{2n+1}\rangle|0\rangle \to |s_{2n+1}\rangle |\xi_{s_{2n+1}}\rangle$
\item Query gradients:
$
O_{\nabla f} |y_{2n}^c\rangle|0\rangle = |y_{2n}^c\rangle|\nabla f(y_{2n}^c)\rangle,
\quad  O_{\nabla f} |y_{2n+1}^f\rangle|0\rangle = |y_{2n+1}^f\rangle |\nabla f(y_{2n+1}^f)\rangle
$
\item Apply arithmetic unitary:
\begin{tightalign*}
|y_{2n+2}^c\rangle &= |y_{2n+1}^c + S(y_{2n}^f - y_{2n}^c)h - \nabla f(y_{2n}^c)h + \sqrt{2h} \xi_{s_{2n+1}}\rangle\\
|y_{2n+2}^f\rangle &= |y_{2n+1}^f + S(y_{2n+1}^c - y_{2n+1}^f)h - \nabla f(y_{2n+1}^f)h + \sqrt{2h} \xi_{s_{2n+1}}\rangle
\end{tightalign*}
\item Apply $O_R$ to get the accumulated Radon–Nikodym derivative of both paths:
\begin{tightalign*}
|R_{2n+1}^f R(y_{2n+2}^f, y_{2n+1}^f, S(y_{2n+1}^c - y_{2n+1}^f), h)\rangle,
\quad |R_{2n}^c R(y_{2n+2}^c, y_{2n}^c, S(y_{2n}^f - y_{2n}^c), 2h)\rangle
\end{tightalign*}
\end{tightitemize}
\textbf{End for}

\STATE Apply $O_\varphi$ to get:
$|y_{2N}^f\rangle|0\rangle \mapsto |y_{2N}^f\rangle|\varphi(y_{2N}^f)\rangle,
\quad |y_{2N}^c\rangle|0\rangle \mapsto |y_{2N}^c\rangle|\varphi(y_{2N}^c)\rangle$

\STATE \textbf{Return}: $ |\Delta_\ell\rangle = |\varphi(y_{2N}^f) \cdot R_{2N}^f - \varphi(y_{2N}^c) \cdot R_{2N}^c\rangle $
\end{algorithmic}
}
\end{algorithm}

\begin{theorem}\label{thm:QA-MLMC under weaker one-sided Lipschitz}
Assume that $ f $ is $ L $-smooth, $L$-Hessian smooth, dissipative,  
and satisfies the weaker one-sided Lipschitz condition with constant $ \lambda $.  

Then, based on the level$-\ell$ samples generated in Algorithm~\ref{alg:Level l sample with change-of-measure QA-MLMC}, 
by applying the unbiased QA-MLMC method with Theorem~\ref{thm:QMLMC2},
we obtain an unbiased estimator of $\E\left[\varphi(X_T)\right]$ whose variance is bounded by $\epsilon^2$,
where $T=\mathcal{O}(\log \epsilon^{-1})$.
The expected total complexity is $$ \widetilde{\mathcal{O}}\left(\frac{SLd(L+\sqrt{d})}{2S-\lambda}\epsilon^{-1} \right),$$
 where $S$ denotes the spring coefficient with $S > \lambda/2$.

\end{theorem}

\begin{proof}
First, we choose $ T = \mathcal{O}(\mathrm{log}(\epsilon^{-1})) $ such that
$$
\left| \mathbb{E}[\varphi(X_T)] - \E_{\pi}[\varphi] \right| \leq \frac{\epsilon}{2}.
$$

Next, let $ h_\ell = h_{(4)} \cdot 2^{-\ell} $.  
By Proposition~\ref{prop: strong order error of discretization}, we have the variance bound
$$
\mathbb{E}[\|\Delta_\ell\|^2] = \mathcal{O}\left( \frac{S^2L^2d^2(L+\sqrt{d})^2}{(2S-\lambda)^2}h_{(4)}^22^{-2\ell}\right),
$$
 so $ \beta = 2 $.
 Combined with the weak error bound from the Euler-Maruyama method, we have $ \alpha = 1 $.  
The cost per level satisfies $$ \mathcal{C}_\ell = \mathcal{O}\left(\frac{T}{h_{(4)}}\,2^{\ell}\right), $$ hence $ \gamma = 1 $.

Applying Theorem~\ref{thm:QMLMC2},  
we conclude that the expected total complexity is 
$$ \widetilde{\mathcal{O}}\left(\frac{SLd(L+\sqrt{d})}{2S-\lambda}\epsilon^{-1}\right).$$
\end{proof}

\begin{remark}
    If the $L$-Hessian smoothness condition is not assumed, then $\beta=1$; in this case, the complexity becomes
    $\widetilde{\mathcal{O}}\left(\epsilon^{-1.5}\right)$.
    Although this does not achieve a quadratic speedup, it still provides an acceleration compared with the classical
    $\mathcal{O}\left(\epsilon^{-2}\right)$ complexity and the direct quantum mean estimation with
    $\widetilde{\mathcal{O}}\left(\epsilon^{-2}\right)$ complexity.
\end{remark}

\subsection{Examples of Non-log-concave Distributions}\label{subsec:example of weaker one-sided}
\subsubsection{Oscillatory Nonconvex Perturbation}
Consider 
$$
f(x) = \frac{1}{2}\|x\|^2 - 2\cos(x_1), 
\quad x=(x_1,\dots,x_d)\in\mathbb{R}^d.
$$
The gradient and Hessian of $f$ are given by
\begin{align*}
  \nabla f(x) =
\begin{pmatrix}
x_1 + 2\sin(x_1)\\
x_2\\
\vdots\\
x_d
\end{pmatrix},
\qquad
\nabla^2 f(x) = I + 2\cos(x_1)\, e_1 e_1^\top,  
\end{align*}
where $e_1=(1,0,\dots,0)^\top$.
Then the eigenvalues of $\nabla^2 f(x)$ are
$$
1+2\cos(x_1), \quad 1, \dots, 1.
$$
In particular, when $x_1=\pi$,
$$
1+2\cos(\pi)= -1 < 0,
$$
so $\nabla^2 f(x)$ is not positive semidefinite. Hence $f$ is not convex, and therefore not strongly convex.

We first verify that $f$ 
satisfies the dissipativity condition and the weaker Lipschitz condition.
For any $x\in\mathbb{R}^d$,
\begin{align*}
   \langle x,\nabla f(x)\rangle
=\|x\|^2 + 2x_1\sin(x_1)\geq \|x\|^2-2|x_1|\geq \frac{1}{2}\|x\|^2-2. 
\end{align*}
Thus $f$ satisfies the dissipativity condition.

Note that the eigenvalues of $\nabla^2 f(x)$ are uniformly bounded below by $-1$. 
Therefore, by the integral form of the mean value theorem, for all $x,y \in \mathbb{R}^d$,
$$
\langle x-y, \nabla f(x) - \nabla f(y)\rangle \geq -\|x-y\|^2.
$$
Hence, $f$ satisfies the weaker one-sided Lipschitz condition.

Moreover, the eigenvalues of $\nabla^2 f(x)$ are uniformly bounded in absolute value by $3$, 
which implies that
$$
\|\nabla f(x) - \nabla f(y)\| \leq 3 \|x-y\|, \quad \forall x,y \in \mathbb{R}^d.
$$

For any $x,y \in \mathbb{R}^d$, we have
$$
\nabla^2 f(x) - \nabla^2 f(y)
=
2\left(\cos(x_1)-\cos(y_1)\right) e_1 e_1^\top.
$$
Since $\cos(\cdot)$ is $1$-Lipschitz, it follows that
$$
|\cos(x_1)-\cos(y_1)| \leq |x_1-y_1| \leq \|x-y\|.
$$
Hence,
$$
\|\nabla^2 f(x) - \nabla^2 f(y)\|
\leq 2 \|x-y\|, \quad \forall x,y \in \mathbb{R}^d.
$$

Therefore, $f$ is $3$-smooth and $2$-Hessian smooth.

The function $f$ satisfies the assumptions of 
Theorem~\ref{thm:QA-MLMC under weaker one-sided Lipschitz}. 
Therefore, its Gibbs expectation can be estimated with mean-squared error at most $\epsilon^2$ 
using a computational complexity of $\widetilde{\mathcal{O}}(\epsilon^{-1})$.

This example provides a multi-dimensional potential $f$ which is not strongly convex (indeed, not even convex), 
yet can still be efficiently sampled and admits a quadratic quantum speedup.

\subsubsection{Radial Nonconvex Potential}

For
$$
f(x)=\frac{1}{2}\|x\|^2 -  a e^{-\|x\|^2}, \quad x\in\mathbb{R}^d,
$$
 where $a> \frac{e}{2}$ is a constant, we have
$$
\nabla f(x) = x + 2 ae^{-\|x\|^2} x = (1+2a e^{-\|x\|^2})x,
$$
and
$$
\nabla^2 f(x)
=
I + 2a e^{-\|x\|^2}(I - 2xx^\top).
$$

For the  Hessian $\nabla^2f$, it has two types of eigenvalues. 
In directions orthogonal to $x$, the eigenvalues are
$$
\lambda_{\mathrm{tan}}(x)=1+2a e^{-\|x\|^2},
$$
while in the radial direction they are
$$
\lambda_{\mathrm{rad}}(x)=1+2a e^{-\|x\|^2}(1-2\|x\|^2).
$$
At $\|x\|=1$, we have
$$
\lambda_{\mathrm{rad}} = 1 - 2a e^{-1}.
$$
Hence, if $a>\frac{e}{2}$, then $\lambda_{\mathrm{rad}}<0$, so $f$ is not convex and therefore not strongly convex.

Firstly,
$$
\langle x, \nabla f(x)\rangle
=
(1+2a e^{-\|x\|^2})\|x\|^2
\geq \|x\|^2,
$$
which shows that $f$ is dissipative.

Since the eigenvalues of $\nabla^2 f(x)$ are uniformly bounded in absolute value by $1+2a$, there exists $L>0$ such that
$$
\|\nabla f(x)-\nabla f(y)\|
\leq L\|x-y\|,
\quad \forall x,y\in\mathbb{R}^d.
$$
Hence,  $f$ is $L$-smooth.

Moreover, the third derivatives of $f$ are combinations of terms of the form
$$
e^{-\|x\|^2}, \quad \|x\| e^{-\|x\|^2}, \quad \|x\|^3 e^{-\|x\|^2},
$$
which are uniformly bounded on $\mathbb{R}^d$. 
Therefore, $\nabla^2 f$ is globally Lipschitz, i.e., there exists $L_H>0$ such that
$$
\|\nabla^2 f(x)-\nabla^2 f(y)\|
\leq L_H \|x-y\|, \quad \forall x,y\in\mathbb{R}^d.
$$
Hence, $f$ is $L_H$-Hessian smooth.

Therefore, $f$ satisfies all the assumptions of Theorem~\ref{thm:QA-MLMC under weaker one-sided Lipschitz}.
Consequently, the Gibbs expectation associated with $f$ can be estimated with 
mean-squared error at most $\epsilon^2$ using a computational complexity of
 $\widetilde{\mathcal{O}}(\epsilon^{-1})$.

\section{Unbiased Quantum Algorithms for Gibbs Expectation }
\label{sec:unbiased}
The results in the previous section establish the multilevel estimates
needed for quantum-accelerated MLMC under weaker dissipativity
assumptions. In particular,
Theorem~\ref{thm:QA-MLMC under weaker one-sided Lipschitz}
extends the analysis from the strongly convex setting to the weaker
one-sided Lipschitz and dissipative setting, and gives biased
QA-MLMC estimators after fixing a finite time horizon and a finest
discretization level.

In this section, we take a further step and consider unbiased estimation
of the invariant Gibbs expectation. The estimates obtained in the
previous section will be used as inputs to the unbiased multilevel
construction: they provide the required bias, variance, and cost rates
for the level differences, while the new ingredient here is the
randomization that removes the remaining limiting bias.

When estimating the Gibbs expectation $\mathbb{E}_\pi[\varphi]$
(or $\mathbb{E}[\varphi(X_\infty)]$) of an SDE,
there are typically two sources of bias.

The first source of bias arises from the finite-time approximation:
one typically chooses a sufficiently large terminal time $T$
such that the distribution of $X_T$ is close to the target distribution $\pi$,
and then uses $\mathbb{E}[\varphi(X_T)]$
as an approximation to $\mathbb{E}[\varphi(X_\infty)]$.
This leads to the time truncation bias
$$
\mathbb{E}[\varphi(X_T)] - \mathbb{E}[\varphi(X_\infty)].
$$

The second source of bias arises from the discretization
of the continuous-time dynamics:
a stepsize $h>0$ is chosen to approximate the SDE.
Let $\widehat X_n$ denote the time-discretized approximation
with $n = T/h$ time steps
(for simplicity assuming that $T/h \in \mathbb{N}$).
One then uses
$\mathbb{E}[\varphi(\widehat X_{T/h})]$
as an approximation to
$\mathbb{E}[\varphi(X_T)]$.
This induces the discretization bias
$$
\mathbb{E}[\varphi(\widehat X_{T/h})]
-
\mathbb{E}[\varphi(X_T)].
$$

Combining the two sources of bias, the total error incurred by using
$\mathbb{E}[\varphi(\widehat X_{T/h})]$
to approximate
$\mathbb{E}[\varphi(X_\infty)]$
can be decomposed as
\begin{align*}
\underbrace{
\mathbb{E}[\varphi(\widehat{X}_{T/h})]
-
\mathbb{E}[\varphi(X_{\infty})]
}_{\text{total bias}}
=
\underbrace{
\mathbb{E}[\varphi(\widehat{X}_{T/h})]
-
\mathbb{E}[\varphi(X_T)]
}_{\text{discretization bias}}
+
\underbrace{
\mathbb{E}[\varphi(X_T)]
-
\mathbb{E}[\varphi(X_{\infty})]
}_{\text{time truncation bias}}.
\end{align*}

Recall that Theorem~\ref{thm:QMLMC2} applies once we can construct a
sequence $\{P_\ell\}_{\ell\ge 0}$ satisfying
\begin{itemize}
\item $|\mathbb{E}[P_\ell - P]|
= \widetilde{\mathcal{O}}\!\left( 2^{-\alpha \ell} \right)$,
\item $\operatorname{Var}(P_\ell - P_{\ell-1})
= \widetilde{\mathcal{O}}\!\left( 2^{-\beta \ell} \right)$,
\item $\mathcal{C}_\ell
= \widetilde{\mathcal{O}}\!\left( 2^{\gamma \ell} \right)$,
\end{itemize}
then, whenever $\beta + 2\alpha > 2\gamma$,
we can construct an unbiased estimator of $\mathbb{E}[P]$.

In Theorem~\ref{thm:QA-MLMC under weaker one-sided Lipschitz},
as in most MLMC algorithms for estimating Gibbs expectations associated
with SDEs,
we choose a sufficiently large time $T$
and construct $P_\ell$ as the level-$\ell$ estimator obtained by
discretizing the SDE with stepsize $h_\ell = \mathcal{O}(2^{-\ell})$.
Consequently,
$$
\left|\mathbb{E}[P_\ell]-
\mathbb{E}[\varphi(X_T)]\right|=\widetilde{\mathcal{O}}\!\left( 2^{-\alpha \ell} \right).
$$
By invoking Theorem~\ref{thm:QMLMC2},
one can construct an unbiased estimator for
$\mathbb{E}[\varphi(X_T)]$.
In other words, the bias arising from time discretization
can be eliminated through the multilevel unbiased construction.

To further remove the time truncation bias
$\mathbb{E}[\varphi(X_T)]-\mathbb{E}[\varphi(X_\infty)],$
a natural idea is to construct a sequence
$\{P_\ell\}_{\ell \geq 0}$
such that
$$
h_\ell \to 0
\quad \text{and} \quad
T_\ell \to \infty
\quad \text{as } \ell \to \infty,
$$
where $P_\ell$ is obtained by discretizing the SDE with stepsize $h_\ell$
over a time horizon $T_\ell$.
That is, we require
$$
\left|\mathbb{E}[P_\ell]-
\mathbb{E}[\varphi(X_\infty)]\right|=\widetilde{\mathcal{O}}\left( 2^{-\alpha \ell} \right).
$$
If this holds, 
although the original MLMC framework (either classical or quantum-accelerated) 
inevitably requires a truncation level $L$, which introduces a bias
$ \mathbb{E}[P_L] - \mathbb{E}[\varphi(X_\infty)] 
= \widetilde{\mathcal{O}}\left(2^{-\alpha L}\right)$,
by employing the unbiased estimator construction 
in Algorithm~\ref{alg:unbiased estimator},
which introduces a random variable $j \sim \mathrm{Geom}(\frac{1}{2})$ together with the bias-correction term
$2^{j}\left(\tilde{\mu}_{j}-\tilde{\mu}_{j-1}\right)$,
Theorem~\ref{thm:QMLMC2} can be applied to construct an unbiased estimator for
$\mathbb{E}[\varphi(X_\infty)]$,
so that both the discretization bias and the time truncation bias are eliminated.

Following the above idea, this section develops two unbiased QA-MLMC constructions for estimating $\mathbb{E}[\varphi(X_\infty)]$.
Based on the framework proposed in~\cite{Giles2016MultilevelMC},
we first introduce a refined unbiased estimator under the one-sided Lipschitz condition.
In addition, we develop a more general algorithm that achieves unbiased estimation without 
this restriction.

\subsection{Unbiased Estimation under One-sided Lipschitz}
\label{subsec:Unbiased Estimation under One-sided Lipschitz}

To construct an unbiased estimator, we begin by expressing $\mathbb{E}[\varphi(X_\infty)]$ as
\begin{align*}
    \mathbb{E}[\varphi(X_\infty)]
    =\mathbb{E}[\varphi(X_{T_0})]+
    \sum_{\ell=1}^\infty \left(\mathbb{E}[\varphi(X_{T_\ell})]-\mathbb{E}[\varphi(X_{T_{\ell-1}})]\right).
\end{align*}
Following the previous idea, we aim to choose suitable increasing sequences 
of time horizons $\{T_\ell\}_{\ell\geq0}$ and step sizes $\{h_\ell\}_{\ell\geq0}$ 
such that $T_\ell \to \infty$ and $h_\ell \to 0$ as $\ell \to \infty$, 
and define $P_\ell$ as the 
time-discretized random variable obtained by evaluating
$\varphi$ at the numerical approximation of 
$X_{T_\ell}$ with step size $h_\ell$.
To apply MLMC, it is essential to construct a strong coupling between the level-$\ell$ 
and level-$(\ell-1)$ approximations so that 
$\operatorname{Var}\left(P_\ell-P_{\ell-1}\right)$ is small.

If we temporarily ignore time discretization and work in continuous time,
then $X_{T_\ell}$ can be viewed as evolving the SDE first over a time interval
of length $T_\ell - T_{\ell-1}$, and then over an additional time interval
of length $T_{\ell-1}$.

The following property ensures that, when driven by the same Brownian motion,
the distance between two paths can exhibit exponential contraction in time,
even if they start from different initial conditions,
which is the key to controlling
$\operatorname{Var}\left(P_\ell - P_{\ell-1}\right)$.

\begin{proposition}\label{prop: different initial but same Brownian}  
        Let $ (W_t)_{t \geq 0} $ be a standard Brownian motion in $ \mathbb{R}^d $. 
    Let $ (X_t)_{t \geq 0} $ and $ (Y_t)_{t \geq 0} $ evolve according to the Langevin diffusion \eqref{SDE0},  
    with initial values $ X_0, Y_0 \in \mathbb{R}^d $.  
    If $ f $ is one-sided Lipschitz with constant $ m $, then we have
    \begin{equation*}
        \mathbb{E}\|X_T - Y_T\|^2 \leq e^{-2mT} \mathbb{E}\|X_0 - Y_0\|^2.
    \end{equation*}
\end{proposition}

\begin{proof}
Applying Itô's lemma, we obtain
\begin{align*}
  \mathrm{d} \left(e^{2mt}\|X_t - Y_t\|^2\right)
  &= 2m e^{2mt} \|X_t - Y_t\|^2  \mathrm{d} t + e^{2mt}  \mathrm{d} \|X_t - Y_t\|^2 \\
  &= 2m e^{2mt} \|X_t - Y_t\|^2  \mathrm{d} t - 2 e^{2mt} \langle X_t - Y_t, \nabla f(X_t) - \nabla f(Y_t) \rangle  \mathrm{d} t.
\end{align*}
Hence,
\begin{align*}
  e^{2mt} \|X_t - Y_t\|^2
  &= \|X_0 - Y_0\|^2 + 2 \int_0^t e^{2ms} \left( m \|X_s - Y_s\|^2 - \langle \nabla f(X_s) - \nabla f(Y_s), X_s - Y_s \rangle \right)  \mathrm{d} s \\
  &\leq \|X_0 - Y_0\|^2.
\end{align*}
Letting $ t = T $ and taking expectations yields
\begin{align*}
  \mathbb{E} \|X_T - Y_T\|^2 \leq e^{-2mT}  \mathbb{E} \|X_0 - Y_0\|^2.
\end{align*}
\end{proof}

To construct a level-$\ell$ unbiased MLMC sample,
the fine path is first evolved over a time interval $\Delta t$.
After this initial stage,
the fine and coarse paths are jointly evolved for a duration
$T_{\ell-1}$ at level $\ell$.
A schematic illustration of this construction is shown in Figure~\ref{fig:2},
and the detailed procedure is presented in
Algorithm~\ref{alg:Sample for MLMC-unbiased}.

\begin{algorithm}[h]
   \small{ \begin{algorithmic}[1]
        \caption{\algname{Level-$\ell$ Sample for Unbiased MLMC}}
        \label{alg:Sample for MLMC-unbiased}
        \item \textbf{Input}: Initial value $x_0$, simulation time $T_{\ell-1}$, $T_\ell$, step size $h$
        
        \textbf{Output}: A level-$\ell$ sample $\Delta_\ell$
        
        \STATE Set $\tilde{x}_0^f = x_0$
        
        \STATE \textbf{For} $k = 1$ \textbf{to} $\frac{T_\ell - T_{\ell-1}}{h}$
        \begin{tightitemize}
            \item Sample $\tilde{\xi}_{k-1} \sim \mathcal{N}(0,I)$
            \item Update the fine path only: $\tilde{x}_k^f = \tilde{x}_{k-1}^f - h \nabla f(\tilde{x}_{k-1}^f) + \sqrt{2h} \,\tilde{\xi}_{k-1}$
        \end{tightitemize}
        \textbf{End for}
        
        \STATE Set $x_0^c = x_0$ as the original initial value, and set
$x_0^f = \tilde{x}_{\frac{T_\ell - T_{\ell-1}}{h}}^f$ as the value obtained in the previous step.

        \STATE \textbf{For} $n = 1$ \textbf{to} $k_\ell = \frac{T_{\ell-1}}{2h}$
        \begin{tightitemize}
            \item Sample $\xi_{n-1}, \xi_{n - \frac{1}{2}} \sim \mathcal{N}(0, I)$
            \item Update the fine and coarse paths simultaneously:
            \begin{tightalign*}
            &x_{n - \frac{1}{2}}^f = x_{n-1}^f - h \nabla f(x_{n-1}^f) + \sqrt{2h} \,\xi_{n-1}\\
            &x_n^f = x_{n - \frac{1}{2}}^f - h \nabla f(x_{n - \frac{1}{2}}^f) + \sqrt{2h} \,\xi_{n - \frac{1}{2}}\\
            & x_n^c = x_{n-1}^c - 2h \nabla f(x_{n-1}^c) + \sqrt{4h} \cdot \frac{1}{\sqrt{2}}(\xi_{n-1} + \xi_{n - \frac{1}{2}})
            \end{tightalign*}
        \end{tightitemize}
        \textbf{End for}
        
        \STATE \textbf{Return} $\Delta_\ell = \varphi(x_{k_\ell}^f) - \varphi(x_{k_\ell}^c)$
    \end{algorithmic}}
\end{algorithm}

After presenting the classical version, 
we now turn to a detailed explanation of its quantum version.
Algorithm~\ref{alg:l-sample for QA-MLMC(unbiased)} provides a detailed quantum algorithm of generating a 
level-$\ell$ sample for the unbiased Gibbs expectation method
and the required oracles are introduced in Section~\ref{sec:(sub)oracles}.

\begin{algorithm}[h]
    \caption{  \algname{Level-$\ell$ Quantum Sample for Unbiased Gibbs Expectation}}
    \label{alg:l-sample for QA-MLMC(unbiased)}
 \small{   \begin{algorithmic}[1]
    
    \item\textbf{Inputs}: Oracle $O_I,O_{\nabla f},O_W,O_\varphi$,
   step sizes $ h_{\ell-1},  h_\ell$, simulation times $T_{\ell-1}, T_\ell$
   
\textbf{Output}: A level-$\ell$ sample state $|\Delta_\ell\rangle$

    \STATE Prepare the initial state $|\tilde{x}_0^{f}\rangle=O_I |0\rangle$
    \STATE \textbf{For} $k = 1$ to $\frac{T_\ell - T_{\ell-1}}{h_\ell}$
   \begin{tightitemize}
    \item Sample $\tilde s_{k-1}$ randomly and query randomness oracle: $O_W |\tilde s_{k-1}\rangle |0\rangle = |\tilde s_{k-1}\rangle |\tilde \xi_{\tilde s_{k-1}}\rangle$
      \item Query gradient oracle: $O_{\nabla f} |\tilde{x}_{k-1}^{f}\rangle |0\rangle = |\tilde{x}_{k-1}^{f}\rangle |\nabla f(\tilde{x}_{k-1}^{f})\rangle$ 
   \item Apply arithmetic unitary to update the fine path:
        $|\tilde{x}_k^{f}\rangle =|\tilde{x}_{k-1}^{f} - h_\ell \nabla f(\tilde{x}_{k-1}^f) + \sqrt{2h_\ell}\,\tilde  \xi_{\tilde s_{k-1}} \rangle$
   \end{tightitemize}     
    
    \textbf{End for}

    \STATE Initialize $|x_0^{c}\rangle=O_I|0\rangle$ and $|x_0^{f}\rangle= |\tilde{x}_{(T_\ell-T_{\ell-1})/h_\ell}^{f}\rangle$

    \STATE \textbf{For} $n = 1$ \textbf{to} $k_\ell= \frac{T_{\ell-1}}{2h_{\ell}}$
\begin{tightitemize}
       \item Sample $s_{n-1},s_{n-1/2}$ randomly and query randomness oracle:
\begin{tightalign*}
             O_W |s_{n-1}\rangle|0\rangle = |s_{n-1}\rangle |\xi_{s_{n-1}}\rangle, \quad 
            O_W |s_{n-1/2}\rangle|0\rangle =|s_{n-1/2}\rangle |\xi_{s_{n-1/2}}\rangle
\end{tightalign*}
        \item  Query gradient oracles: $\quad  O_{\nabla f} |x_{n-1}^{c}\rangle |0\rangle = |x_{n-1}^{c}\rangle |\nabla f(x_{n-1}^{c})\rangle,$
\begin{tightalign*}
            O_{\nabla f} |x_{n-1}^{f}\rangle |0\rangle = |x_{n-1}^{f}\rangle |\nabla f(x_{n-1}^{f})\rangle,\quad
            O_{\nabla f} |x_{n-1/2}^{f}\rangle |0\rangle = |x_{n-1/2}^{f}\rangle |\nabla f(x_{n-1/2}^{f})\rangle  
\end{tightalign*}
        \item Apply arithmetic unitary to update  both paths:
\begin{tightalign*}
            &|x_n^c\rangle =|x_{n-1}^c - h_{\ell-1} \nabla f(x_{n-1}^c) + \sqrt{h_{\ell-1}}\,(\xi_{s_{n-1}} + \xi_{s_{n-1/2}})\rangle\\
            &|x_{n-1/2}^{f}\rangle = |x_{n-1}^{f} - h_\ell \nabla f(x_{n-1}^f) + \sqrt{2h_\ell} \,\xi_{s_{n-1}}\rangle\\
            &|x_n^{f}\rangle = |x_{n-1/2}^{f} - h_\ell \nabla f(x_{n-1/2}^f) + \sqrt{2h_\ell}\, \xi_{s_{n-1/2}}\rangle   
\end{tightalign*}
 \end{tightitemize}
\textbf{End for}
    \STATE Apply query oracle $O_\varphi$:
    $\quad|x_{k_\ell}^{f}\rangle|0\rangle \mapsto |x_{k_\ell}^{f}\rangle|\varphi(x_{k_\ell}^f)\rangle,
        |x_{k_\ell}^{c}\rangle|0\rangle \mapsto |x_{k_\ell}^c\rangle|\varphi(x_{k_\ell}^c)\rangle$
    \STATE \textbf{Return} $|\Delta_\ell\rangle = |\varphi(x_{k_\ell}^f) - \varphi(x_{k_\ell}^c)\rangle$
    \end{algorithmic}}
\end{algorithm}

\begin{theorem}\label{thm:unbiased QA-MLMC}
Assume that $ f $ is  $ L $-smooth, $L$-Hessian smooth , and satisfies the one-sided Lipschitz condition with constant $m$.
Then, based on the level$-\ell$ samples generated in Algorithm~\ref{alg:l-sample for QA-MLMC(unbiased)}, 
by applying the QA-MLMC procedure together with Theorem~\ref{thm:QMLMC2},
we can obtain an unbiased estimator of $ \mathbb{E}_\pi[\varphi] $
with variance at most $\epsilon^2$. 
The expected total complexity is
$$\widetilde{\mathcal{O}}\left(\frac{L^2\sqrt{d}+Ld}{m^2}\epsilon^{-1}\right).$$
\end{theorem}
\begin{proof}
Firstly, by Lemma~\ref{lem:bounds of p moment},  
we know that $\E\left[\|\widehat{X}^f_{0}-\widehat{X}^c_{0}\|^2\right]$  
can be uniformly bounded by $M^2$
with $M=\mathcal{O}(\sqrt{md})$.

Using Corollary~\ref{cor: p moment distance of one-sided} with $\gamma=\frac{m}{2}$.
For $h_\ell=h_{(2)}2^{-\ell}$ and $N=\frac{T_{\ell-1}}{h_{\ell-1}}$,
we have
       \begin{align*}
    \E\left[\|\widehat X^f_{2N}-\widehat X^c_{2N}\|^2\right] 
    \leq C_4^2 e^{-mT_{\ell-1}/2}  
 M^2 +4C_{(3)}^2h_{\ell}^2.
\end{align*}
Choose $T_\ell$ such that 
$C_4^2e^{-mT_{\ell-1}/2} 
 M^2 =C_{(3)}^2h_{\ell}^2,$
 i.e.,
 \begin{align*}
     T_{\ell-1}= \frac{4}{m}\left(\ell \log 2+\log\left(\frac{C_4M}{C_{(3)}h_{(2)}}\right)\right).
 \end{align*}
 Hence
        \begin{align*}
    \E\left[\|\widehat X^f_{T_{\ell-1}/h^{\ell}}-\widehat X^c_{T_{\ell-1}/h^{\ell}}\|^2\right] 
    \leq 5C_{(3)}^2h_{\ell}^2.
\end{align*}

In Proposition~\ref{prop: different initial but same Brownian}, if we take $Y_0 \sim \pi $, then
$$
\left| \mathbb{E}[\varphi(X_{T_\ell})] - \mathbb{E}[\varphi(X_\infty)] \right| = \mathcal{O}(e^{-mT_\ell})=\mathcal{O}(2^{-4\ell}).
$$
Combining this with the weak error of the Euler discretization,
we obtain $\alpha = 1$.

Therefore, in Theorem~\ref{thm:QMLMC2}, $\alpha = 1$ and
\begin{align*}
    V_\ell=\widetilde{\mathcal{O}}\left( \frac{(L^2\sqrt{d}+Ld)^2}{m^2}h_{(2)}^2 2^{-2\ell}\right),
    \quad \mathcal{C}_\ell=\mathcal{O}\left(\frac{T_\ell}{h_{(2)}}2^{\ell}\right)
    =\widetilde{\mathcal{O}}\left(\frac{2^{\ell}}{mh_{(2)}}\right).
\end{align*}
The total complexity is $$\widetilde{\mathcal{O}}\left(\frac{L^2\sqrt{d}+Ld}{m^2}\epsilon^{-1}\right).$$
\end{proof}

\subsection{Unbiased Estimation under Weaker One-sided Lipschitz}

In Theorem~\ref{thm:QA-MLMC under weaker one-sided Lipschitz}, 
after introducing the spring term, 
the distance between the coarse and fine paths remains uniformly small 
under the new measure. As a consequence, 
when transforming back to the original measure, 
the variance remains well controlled.

In contrast, 
if we attempt to use the unbiased construction in Theorem~\ref{thm:unbiased QA-MLMC}, 
the coarse and fine paths under the changed measure start with a 
non-negligible separation. Although this distance may decrease
 during the subsequent evolution, the variance after reverting to
  the original measure can no longer be effectively controlled.

Although Theorem~\ref{thm:unbiased QA-MLMC} cannot be directly generalized to the unbiased case,
by applying Theorem~\ref{thm:unbiased estimator construct},
any biased estimation procedure can be converted into an unbiased one
by introducing a random variable $j \sim \mathrm{Geom}(\frac{1}{2})$
together with the bias-correction term
$2^{j}\left(\tilde{\mu}_{j}-\tilde{\mu}_{j-1}\right)$.
Therefore, in a more general setting, we can obtain an unbiased version of Theorem~\ref{thm:QA-MLMC under weaker one-sided Lipschitz}.

\begin{algorithm}[H]
\caption{\algname{Unbiased  Quantum Algorithm for Gibbs Expectation}}
\label{alg:Unbiased  Quantum Algorithm for Gibbs Expectation}
\begin{algorithmic}[1]
 \item 
\textbf{Output}: An unbiased estimator $\tilde{\mu}$ of $\E[\varphi(X_\infty)]$ with $\operatorname{Var}(\tilde{\mu})\leq \sigma^2$.

\STATE Use Theorem~\ref{thm:QMLMC2} to define a biased estimator $\mathcal{A}(\cdot)$,
where each level difference is generated by
Algorithm~\ref{alg:Level l sample with change-of-measure QA-MLMC}.

\STATE Sample $j \sim \mathrm{Geom}(\frac{1}{2})$.

\STATE Compute $\tilde{\mu}_0=\mathcal{A}(\sigma/M)$,
$\tilde{\mu}_{j}=\mathcal{A}\!\left(2^{-\rho j}\sigma/M\right)$,
and $\tilde{\mu}_{j-1}=\mathcal{A}\!\left(2^{-\rho (j-1)}\sigma/M\right)$.

\STATE Return $\tilde{\mu}=\tilde{\mu}_0+2^{j}(\tilde{\mu}_{j}-\tilde{\mu}_{j-1})$.
\end{algorithmic}
\end{algorithm}

Algorithm~\ref{alg:Unbiased  Quantum Algorithm for Gibbs Expectation}
presents a  procedure for constructing an unbiased quantum estimator
of the Gibbs expectation $\E[\varphi(X_\infty)]$.
The algorithm combines a biased quantum multilevel estimator with a geometric
randomization technique to remove bias while maintaining a controlled variance.

Specifically, the algorithm constructs a biased estimator $\mathcal{A}(\cdot)$ for $\E[\varphi(X_\infty)]$ via QA-MLMC, following Theorem~\ref{thm:QA-MLMC under weaker one-sided Lipschitz}, where each level difference is generated using Algorithm~\ref{alg:Level l sample with change-of-measure QA-MLMC}.
To obtain an unbiased estimate, the algorithm applies geometric randomization: a random level index is sampled from a geometric distribution, and the estimator $\mathcal{A}(\cdot)$ is evaluated at multiple accuracy levels. These estimates are then combined in a multilevel form to produce the final output $\tilde{\mu}$.  Figure~\ref{fig:3} illustrates the construction.

\begin{theorem}\label{thm:unbiased A-MLMC under weaker one-sided Lipschitz}
Assume that $ f $ is $ L $-smooth, $L$-Hessian smooth, dissipative,  
and satisfies the weaker one-sided Lipschitz condition with constant $ \lambda $.  
Then we can obtain an unbiased estimator $\tilde{\mu}$ of $\mathbb{E}_\pi[\varphi]$ 
such that $\operatorname{Var}(\tilde{\mu}) \leq \epsilon^{2}$.  
The expected total complexity is $$ \widetilde{\mathcal{O}}\left(\frac{SLd(L+\sqrt{d})}{2S-\lambda}\epsilon^{-1} \right),$$
 where $S$ denotes the spring coefficient with $S > \lambda/2$.
\end{theorem}

\begin{proof}
    The conclusion  follows directly by applying 
     Theorem~\ref{thm:unbiased estimator construct} to 
    Theorem~\ref{thm:QA-MLMC under weaker one-sided Lipschitz}.
\end{proof}

\section{Quantum Sampling of Heavy-Tailed Distributions}
\label{sec:heavy-tailed}

In recent years, quantum algorithms have shown promising potential for
accelerating Markov chain Monte Carlo methods
\cite{Ozgul2025QuantumSF,Childs2022QuantumAF},
particularly for well-behaved target distributions, such as those
satisfying strong convexity, logarithmic Sobolev inequalities, or
Poincaré inequalities.

However, many practical problems involve heavy-tailed distributions,
for example in Bayesian posterior inference with weak priors, financial
modeling, and robust statistics. In such cases, both classical and
quantum sampling methods may suffer from degraded performance because
the underlying Langevin dynamics can mix slowly and may fail to satisfy
the functional inequalities used in standard complexity analyses.

Therefore, in this section, we study quantum sampling and expectation
estimation for a class of heavy-tailed distributions. More precisely, by
``heavy-tailed'' we refer to the following definition.

\begin{definition}
\label{def:heavy-tailed}
A random variable \(X\) is said to be heavy-tailed if, for every
\(\lambda>0\),
\[
\lim_{x\to\infty} e^{\lambda x}\mathbb{P}(X>x)=\infty.
\]
\end{definition}

One particular instance of heavy-tailed behavior arises when
\(\pi\propto e^{-f}\) and
\[
\lim_{\|x\|\to\infty}\|\nabla f(x)\|=0.
\]
In this case, the drift force \(-\nabla f(x)\) vanishes at infinity, so
the Langevin process is no longer strongly pulled back toward regions of
high probability density. As a result, the Markov chain induced by the
SDE may fail to be exponentially ergodic.

\begin{proposition}[Theorem 2.4 in~\cite{Roberts1996ExponentialCO}]
If
\[
\lim_{\|x\|\to\infty}\|\nabla f(x)\|=0,
\]
then the process \((X_t)_{t\geq 0}\) defined by
\[
\mathrm dX_t=-\nabla f(X_t)\,\mathrm dt+\sqrt{2}\,\mathrm dW_t
\]
is not exponentially ergodic.
\end{proposition}

As a consequence, traditional variance and complexity bounds for
QA-MLMC, which rely on exponential convergence rates or functional
inequalities such as Poincaré or logarithmic Sobolev inequalities, may
no longer apply directly.

To address this difficulty, a natural strategy is to transform the
heavy-tailed target into a distribution with lighter tails. The
transformed Langevin approach of~\cite{He2024AnAO} introduced such a
transformation for heavy-tailed sampling. Building on this idea, we use a
more flexible transformation map by incorporating a polynomial factor
\(r^\alpha\) in addition to the exponential tail factor, and by modifying
the map near the origin so that it coincides with the identity in a
neighborhood of zero. This local modification avoids imposing the
tail-oriented transformation assumptions on the original potential \(f\)
near the origin, while preserving the desired tail behavior of the
transformed potential. The resulting framework does not cover all
heavy-tailed distributions; rather, it applies to those for which the
transformed potential satisfies the dissipativity and smoothness
assumptions needed for the unbiased QA-MLMC theorem.

More concretely, given a smooth invertible map
\(h:\mathbb R^d\to\mathbb R^d\), we write
\[
X=h(Y).
\]
If \(X\sim \pi\propto e^{-f}\), then \(Y\) has density
\[
\pi_h(y)\propto e^{-f_h(y)},
\qquad
f_h(y)=f(h(y))-\log\det(\nabla h(y)).
\]
Thus the original Gibbs expectation can be rewritten as
\[
\mathbb E_\pi[\varphi(X)]
=
\mathbb E_{\pi_h}[\varphi(h(Y))].
\]
In the sequel we denote the transformed observable by
\[
\varphi_h(y):=\varphi(h(y)).
\]
Accordingly, the quantum algorithm will be applied to the transformed
potential \(f_h\) and the transformed observable \(\varphi_h\), and the
unbiased QA-MLMC framework developed in the previous sections can be
used once the required assumptions are verified for \(f_h\) and
\(\varphi_h\).

\subsection{Transformation Map}\label{section Transformation Map}

%\subsection{Construction of Transformation Map}
We now construct the transformation map used in the rest of this
section. 
Consider a smooth invertible transformation map $h:\mathbb{R}^d\rightarrow \mathbb{R}^d$.
If a random vector $X$ has density $\mu$, 
we denote the density of random vector $Y=h^{-1}(X)$ as $\mu_h$, where $\mu_h$ is given by
\begin{align*}
    \mu_h(x)=\mu(h(x))\det(\nabla h(x)).
\end{align*}
In particular, when $X\sim \pi\propto e^{-f}$, we have
\begin{align*}
    \pi_h(y)\propto e^{-f_h(y)} \text{ with }f_h(y)=f(h(y))-\mathrm{log}\det(\nabla h(y)).
\end{align*}

For simplicity, we restrict to an isotropic setting so that the transformation can
be constructed through a one-dimensional radial map.

\begin{assumption}\label{ass:f}
The initial potential function  $f$ is isotropic, 
i.e. $f(x) = f(\|x\|) $, and $f : \mathbb{R} \rightarrow \mathbb{R} $ is three times continuously differentiable.
\end{assumption}

Under this assumption, we define the radial stretching function \(g\) as follows:
\begin{align}\label{g}
    g(r)=
        \begin{cases}
           r,& r\in [0, R_1),\\
           \chi(r)r+(1-\chi(r))r^\alpha e^{br^\beta},&r\in[R_1,R_2], \\
           r^\alpha e^{br^\beta},&r\in(R_2,\infty).
        \end{cases}
\end{align}
where $\alpha,b\geq0$, $\alpha\geq 1$ when $b=0$, $\beta\in (1,2]$, and
$\chi:[0,\infty)\rightarrow[0,1]$ satisfies 
\begin{itemize}
    \item $\chi(r)=1$ for $r\in[0,R_1]$;
    \item $\chi(r)=0$ for $r\in[R_2,\infty)$;
 \item $\chi^{(k)}(R_1)=\chi^{(k)}(R_2)=0$ for $k=1,2,3$ and $\chi'(r)\leq 0$ on $(R_1,R_2)$.
\end{itemize}
\begin{remark}
     Such a $\chi(r)$ is easy to construct, and a concrete construction can be found in Appendix~\hyperref[Appendix B: Technical Details for heavy tailed]{B}, 
Lemma~\ref{lem:construct of chi}.
\end{remark}

The following lemma shows that, with a suitable choice of the transition
region, this radial stretching function defines a valid one-dimensional
change of variables.

\begin{lemma}
Assume that $\alpha$ and $b$ are not both zero (i.e., $g$ is nontrivial in the tail).
By choosing suitable $0<R_1<R_2$, the function $g$ can be made to satisfy 
$$
g \in C^3\left((0,\infty)\right),
$$
and it is onto and strictly monotonically increasing. Hence, $g$ is invertible.
\end{lemma}

\begin{proof}
First, we show that $g\in C^3\left((0,\infty)\right)$.

It is clear that $g$ is three times continuously differentiable on 
$(0,R_1)$, $(R_1,R_2)$, and $(R_2,\infty)$. 
Moreover, in a small neighborhood of $R_1$ we have 
$$
g(r)=\chi(r)r+(1-\chi(r))r^\alpha e^{br^\beta},
$$
and since $\chi^{(k)}(R_1)=0$ for $k=1,2,3$, 
all derivatives of $g$ up to order $3$ match on both sides of $R_1$.
The situation at $R_2$ is analogous, and hence $g$ is three times continuously differentiable.

Next, we prove that $g$ is strictly increasing.
For $0<r<R_1$ we have $g'(r)=1>0$, and for $r>R_2$,
$$
g'(r)= \alpha r^{\alpha-1}e^{br^\beta}+b\beta r^{\beta-1}r^\alpha e^{br^\beta}
=r^{\alpha-1}e^{b r^\beta}\left(\alpha + b\beta r^\beta\right).
$$
Since $\alpha \geq 0$, $b \geq 0$, $\beta>0$, $r>0$, and $\alpha$ and $b$ are not both zero, 
we obtain $\alpha + b\beta r^\beta > 0$, and hence $g'(r) > 0$ on $(R_2,\infty)$.

 For $r\in(R_1,R_2)$ we have
$$
g(r)=\chi(r) r + (1-\chi(r)) r^\alpha e^{br^\beta}.
$$
Denote $\psi(r)=r^\alpha e^{b r^\beta}$ for brevity, then 
$$g'(r) = \chi(r)\cdot 1 + (1-\chi(r))\psi'(r) + \chi'(r)\left(r-\psi(r)\right),$$
where  $\psi'(r)=r^{\alpha-1} e^{b r^\beta}\big(\alpha + b\beta r^\beta\big).$

Since $\psi(r)/r = r^{\alpha-1} e^{b r^\beta}$ and $\alpha\geq 1$ when $b=0$,
 there exists $R_1>0$ such that $\psi(r)\geq r$ for all $r\geq R_1$. 
Choose any $R_2>R_1$. On $(R_1,R_2)$, all three terms of 
$g'(r)$ are nonnegative, and the first two terms are not simultaneously zero. 
Therefore, $g'(r)>0$, which implies that $g$ is strictly increasing.

Finally,
since $g(0)=0$, $\lim\limits_{r\to\infty} g(r)=\infty$ and $g$ is continuous and strictly increasing on $[0,\infty)$, 
it follows that the range of $g$ is $(0,\infty)$ .
Therefore, by continuity, strict monotonicity, and surjectivity, 
we conclude that $g$ is invertible on $(0,\infty)$.
\end{proof}

The invertibility of \(g\) allows us to lift this one-dimensional radial
stretching to an invertible transformation on \(\mathbb R^d\).
 Define the isotropic transformation $h:\mathbb{R}^d\rightarrow\mathbb{R}^d$ as
  \begin{align}
    h(x)=\begin{cases}
        g(\|x\|)\frac{x}{\|x\|},&x\neq 0,\\
        0,&x=0.
    \end{cases}
  \end{align}
It's straightforward to see that
\begin{align*}
    h^{-1}(x)=\begin{cases}
        g^{-1}(\|x\|)\frac{x}{\|x\|},&x\neq 0,\\
        0,&x=0.
    \end{cases}
\end{align*}

Let us now introduce the definition of the transformed function $f_h$.
\begin{definition}\label{def:transformed density}
Let $ X $ be a random vector with density $ \mu $.  
We define the density of the transformed random vector $ Y = h^{-1}(X) $ as $ \mu_h $,  
which we refer to as the transformed density of $ \mu $.

In particular, if $ X $ admits a density $ \pi \propto e^{-f} $,  
then the random variable $ Y = h^{-1}(X) $ has density
$$
\pi_h(y) \propto e^{-f_h(y)} \quad \text{with} \quad f_h(y) = f(h(y)) - \mathrm{log} \det(\nabla h(y)).
$$
\end{definition}

Furthermore, the transformed potential $f_h$ can be represented as
\begin{align}
    \label{eq:formula of f_h}
    f_h(x) = f(g(\|x\|)) - \mathrm{log}  \,\det(\nabla h(x)) 
= f(g(\|x\|)) - \mathrm{log} \,g'(\|x\|) - (d-1)\mathrm{log}  \,g(\|x\|) + (d-1)\mathrm{log}  \|x\|.
\end{align}

The following lemma ensures that $\pi_h$ and $f_h$ can be well-defined.

\begin{lemma}
For the mapping $h$ defined above, the Jacobian determinant satisfies
$$
\det(\nabla h(x)) > 0 \quad \text{for all } x \in \mathbb{R}^d.
$$
Moreover, 
$$
\log\det(\nabla h(x))
$$
is well-defined and continuous on $\mathbb{R}^d$.
\end{lemma}

\begin{proof}
Let $r=\|x\|$. For $x\neq0$, denote $\varphi(r)=g(r)/r$. Then $h(x)=\varphi(r)x$ and
$$
\nabla h(x)=\varphi(r)I + \frac{\varphi'(r)}{r}x x^\top.
$$
The eigenvalues of $\nabla h(x)$ are $\varphi(r)$ with multiplicity $n-1$ 
on the orthogonal complement of $x$ and $\varphi(r)+\varphi'(r)r$ 
in the radial direction. Hence for $x\neq0$,
$$
\det(\nabla h(x))=\varphi(r)^{d-1}\left(\varphi(r)+\varphi'(r)r\right)
=\left(\frac{g(r)}{r}\right)^{d-1} g'(r).
$$
Since $g(r)>0$ and $g'(r)>0$ for all $r>0$,
each factor  is strictly positive,
and therefore $\det(\nabla h(x))>0$ for every $x\neq0$. 

Finally we check the origin. 
Note that $g\in C^1$ at $0$ with $g(0)=0$ and $g'(0)>0$, we have
$$
\lim_{r\to 0+}\frac{g(r)}{r}=g'(0),
\quad
\lim_{r\to 0+}g'(r)=g'(0).
$$
Hence
$$
\det(\nabla h(0))=\lim_{r\to 0+}\left(\frac{g(r)}{r}\right)^{d-1} g'(r)=g'(0)^d>0.
$$
Thus $\det(\nabla h(x))>0$ also holds at $x=0$.

Moreover, since $g \in C^3((0,\infty))$, 
we know that $\nabla h(x)$ is continuous on $\mathbb{R}^d$ 
and that the determinant is a continuous positive function. 
Therefore, $\log\det(\nabla h(x))$ is well-defined and continuous on $\mathbb{R}^d$.
\end{proof}

We have now defined the transformation $ h $ and the transformed potential $ f_h $.  
For simplicity, in the remainder of this section, unless otherwise specified,  
we assume that the functions $ g $, $ h $, and $ f_h $ are as defined above,  
and that the function $ f $ satisfies Assumption~\ref{ass:f}.

The following two lemmas provide an intuitive description of the properties of the gradient and the Hessian matrix of $f_h$.

\begin{lemma}[Gradient of Transformed Potential]\label{lemma:gradient}
The gradient of $f_h$ can be explicitly written as
$$
 \nabla f_h(x) = \left[g'(r) f'(g(r)) -\frac{g''(r)}{g'(r)}- (d-1)\frac{g'(r)}{g(r)}
   + (d-1)\frac{1}{r}\right]\frac{x}{r},\quad x\neq0,
$$
where $r = \|x\|$.
\end{lemma}

\begin{proof}

Denote $r = \|x\|$, \eqref{eq:formula of f_h} gives that
\begin{align*}
   \nabla f_h(x) = \left[g'(r) f'(g(r)) -\frac{g''(r)}{g'(r)}- (d-1)\frac{g'(r)}{g(r)}
   + (d-1)\frac{1}{r}\right]\frac{x}{r}
\end{align*}
for $x\neq 0$.

\end{proof}

\begin{lemma}[Eigenvalues of Hessian Matrix]\label{lemma:eigenvalues}
The Hessian matrix $\nabla^2 f_h(x)$ has  two distinct eigenvalues:
\begin{align*}
    \lambda_1(r)=&(g'(r))^2 f''(g(r))+g''(r) f'(g(r)) 
   -\frac{g'''(r)g'(r)-(g''(r))^2}{(g'(r))^2}- (d-1)\frac{g''(r)g(r)-(g'(r))^2}{(g(r))^2}
   - (d-1)\frac{1}{r^2},\\
   \lambda_2(r)=& \frac{1}{r}\left[g'(r) f'(g(r)) -\frac{g''(r)}{g'(r)}- (d-1)\frac{g'(r)}{g(r)}
   + (d-1)\frac{1}{r}\right].
\end{align*}
\end{lemma}

\begin{proof}

 A classical fact is that 
for any radial function $U(x)=u(r)$ with 
$r=\|x\|$, the Hessian matrix $\nabla^2 U(x)$ has two distinct eigenvalues,
$$
\lambda_{\mathrm{rad}}(r)=u''(r), \quad 
\lambda_{\mathrm{tan}}(r)=\frac{u'(r)}{r},
$$
where the radial eigenvalue $\lambda_{\mathrm{rad}}$ has multiplicity $1$, 
and the tangential eigenvalue $\lambda_{\mathrm{tan}}$ has multiplicity $d-1$.

Therefore, 
\begin{align*}
    \lambda_1(r)=&(g'(r))^2 f''(g(r))+g''(r) f'(g(r)) 
   -\frac{g'''(r)g'(r)-(g''(r))^2}{(g'(r))^2}- (d-1)\frac{g''(r)g(r)-(g'(r))^2}{(g(r))^2}
   - (d-1)\frac{1}{r^2},\\
   \lambda_2(r)=& \frac{1}{r}\left[g'(r) f'(g(r)) -\frac{g''(r)}{g'(r)}- (d-1)\frac{g'(r)}{g(r)}
   + (d-1)\frac{1}{r}\right].
\end{align*}
\end{proof}

\begin{lemma}
   $f_h$ is three times continuously differentiable on $\mathbb{R}^d$.
\end{lemma}
\begin{proof}
    Because $g\in C^3$ away from the origin and $f\in C^3$,
    the Jacobian logarithm is $C^3$ and hence
    $f_h \in C^3(\mathbb{R}^d \setminus \{0\})$. 
   Therefore, it remains to check the behavior as $x \to 0$. 
   
    Since $g(r)=r$ when $r<R_1$, we have $h(x) = x$ and $\det(\nabla h(x)) = 1$ in a neighborhood of $0$. 
Hence, $f_h(x) = f(x)$ for $\|x\| < R_1$.
    The three times differentiability of $f$ directly implies that $f_h$
is three times differentiable at $0$.
\end{proof}

\subsection{Transformation Assumptions}

We now state the assumptions under which the transformed potential
\(f_h\) satisfies the analytic conditions required by the QA-MLMC
framework. These assumptions are imposed on the original potential
\(f\), and we verify the resulting properties of \(f_h\) in three steps:
smoothness, dissipativity, and Hessian smoothness.

\begin{assumption}[Smooth Under Transformation]\label{assumption:smooth}
There exist constants $N, L > 0$  such that for all  $r > \max\{N,R_2\}$, it holds that
\begin{align*}
&\left|(\psi'(r))^2 f''(\psi(r))+\psi''(r) f'(\psi(r)) 
   -\frac{\psi'''(r)\psi'(r)-(\psi''(r))^2}{(\psi'(r))^2}- (d-1)\frac{\psi''(r)\psi(r)-(\psi'(r))^2}{(\psi(r))^2}
   - (d-1)\frac{1}{r^2}\right| <L, \\
 &\left| \frac{1}{r}\left[\psi'(r) f'(\psi(r)) -\frac{\psi''(r)}{\psi'(r)}- (d-1)\frac{\psi'(r)}{\psi(r)}
   + (d-1)\frac{1}{r}\right]\right|<L,
\end{align*}
where $\psi(r) = r^\alpha e^{br^\beta} $ for all $r \geq R_2$.
\end{assumption}

This assumption controls the two eigenvalues of the Hessian of
\(f_h\) in the tail, which yields global smoothness after combining it
with compactness on the interior region.
\begin{lemma}\label{lemma of Lipschitz Gradient}
If the initial function $f$ satisfies Assumption~\ref{assumption:smooth}, 
then there exists a constant $L_h > 0$ such that the transformed function $f_h(x)$ is $L_h$-smooth, 
i.e., for all $x, y \in \mathbb{R}^d$, we have
$$
\|\nabla f_h(x) - \nabla f_h(y)\| \leq L_h \| x - y \|.
$$
\end{lemma} 

\begin{proof}
By combining Lemma~\ref{lemma:eigenvalues} and Assumption~\ref{assumption:smooth}, 
we know that 
when $\|x\| \geq \max\{N,R_2\}$,
$$
\left|\lambda_i(\|x\|) \right|\leq L, \quad \text{for } i = 1, 2.
$$

Denote $\widetilde{N}=\mathrm{max}\{N, R_2\}$.
On the other hand, when $\|x\| \leq \widetilde{N}$, 
since $f_h \in C^2(\mathbb{R}^d)$,
the Hessian $\nabla^2 f_h(x)$ is continuous on the compact set 
$\{x \in \mathbb{R}^d : \|x\| \leq \widetilde{N}\}$, and hence bounded. Therefore,
$$
\max_{\|x\| \leq \widetilde{N}} \|\nabla^2 f_h(x)\| < \infty.
$$

Define
$$
L_h= \mathrm{max}\left\{L, \mathrm{max}_{\|x\| \leq \widetilde{N}} \|\nabla^2 f_h(x)\| \right\}.
$$
Then for all $x \in \mathbb{R}^d$, we have $\|\nabla^2 f_h(x)\| \leq L_h$, 
and hence $f_h$ is $L_h$-smooth.
\end{proof}

Next, we impose a tail coercivity condition on the radial derivative of
\(f_h\). This condition is designed to ensure dissipativity of the
transformed potential.

\begin{assumption}[Dissipativity under Transformation]\label{assumption:dissipativity}
There exist constants $A, B, N > 0$ such that for all $r > \max\{N,R_2\}$:
\begin{align*}
   \left[ \psi'(r)f'(\psi(r))- \frac{\psi''(r)}{\psi'(r)}
- (d-1)\frac{\psi'(r)}{\psi(r)} +\frac{d-1}{r}\right]r > Ar^2 - B,
\end{align*}
where $\psi(r) = r^\alpha e^{br^\beta}$ for all $r \geq R_2$.
\end{assumption}

The following lemma extends this tail dissipativity bound to the whole
space by using continuity on the compact interior region.
\begin{lemma}\label{lemma of dissipativity}
If the initial function $f$ satisfies Assumption~\ref{assumption:dissipativity}, 
then the transformed potential $f_h$ is dissipative, i.e., 
there exist constants $A_h, B_h > 0$ such that
\begin{align*}
   \langle \nabla f_h(x), x \rangle\geq A_h \|x\|^2 - B_h.
\end{align*}
\end{lemma}

\begin{proof}
For $r=\|x\|>R_2$,
$$
\nabla f_h(x) = \left[ \psi'(r)f'(\psi(r))- \frac{\psi''(r)}{\psi'(r)}
- (d-1)\frac{\psi'(r)}{\psi(r)} +\frac{d-1}{r}\right]\frac{x}{r},
$$
we have
$$
 \langle \nabla f_h(x), x \rangle  =
 \left[ \psi'(r)f'(\psi(r))- \frac{\psi''(r)}{\psi'(r)}
- (d-1)\frac{\psi'(r)}{\psi(r)} +\frac{d-1}{r}\right]r.
$$
By Assumption~\ref{assumption:dissipativity}, the right-hand side is lower bounded as
$$
\langle \nabla f_h(x), x \rangle \geq A \|x\|^2 - B.
$$

Since $f_h \in C^2(\mathbb{R}^d)$, the function $\langle \nabla f_h(x), x \rangle$ is continuous. Therefore, the minimum over the compact set 
$\{x \in \mathbb{R}^d : \|x\| \leq R_2 \}$
exists and is finite.
Define
$$
B_h = \mathrm{max}\left\{ 0, B, -\mathop{\mathrm{min}}\limits_{\|x\| \leq R_2} \left\{\langle \nabla f_h(x), x \rangle \right\}\right\}.
$$
Then for all $x \in \mathbb{R}^d$, we have
$$
\langle \nabla f_h(x), x \rangle \geq A \|x\|^2 - B_h.
$$
Hence, $f_h$ satisfies a dissipativity condition with parameters $A_h = A$ and $B_h$ as defined above.
\end{proof}

Finally, we impose a tail bound on the third-order radial derivatives,
which will imply Hessian smoothness of \(f_h\).

\begin{assumption}[Hessian Smoothness Under Transformation]\label{assumption:hessian-smooth}
There exist constants $N, L > 0$ such that for all $r > \max\{N, R_2\}$, it holds that
\begin{align*}
|u'''(r)| \leq L,
\quad
\left|\frac{u''(r)}{r}-\frac{u'(r)}{r^{2}}\right|\leq L,
\end{align*}
where $u(r)=f(\psi(r)) - \mathrm{log} \psi'(r) - (d-1)\mathrm{log} \psi(r) + (d-1)\mathrm{log} r$
and 
$\psi(r) = r^\alpha e^{br^\beta} $ for all $r \geq R_2$.
\end{assumption}

The next lemma converts this radial third-order bound into a global
Lipschitz bound for the Hessian.

\begin{lemma}\label{lem: Hessian smooth under transform}
If the initial potential $f$ satisfies Assumption~\ref{assumption:hessian-smooth}, 
then there exists a constant $L_h > 0$ such that the transformed potential 
$f_h(x) = f(h(x)) - \log \det(\nabla h(x))$ is $L$-Hessian smooth, 
that is, for all $x, y \in \mathbb{R}^d$, 
\begin{align*}
  \|\nabla^2 f_h(x) - \nabla^2 f_h(y)\| \leq L_h \|x - y\|.  
\end{align*}
\end{lemma}
\begin{proof}
%    By directly differentiating and computing the two eigenvalues of 
%    $\nabla^2 f_h$ in Lemma~\ref{lemma:eigenvalues}, the desired conclusion of the lemma follows immediately.

Note that for radial function $f_h(x)=u(r)$, with 
$r=\|x\|$, $e_r=\frac{x}{r}$, we have
\begin{align*}
    \nabla f_h(x)=u'(r)e_r,\quad \nabla^2f_h(x)=u''(r)e_r\otimes e_r+\frac{u'(r)}{r}(I-e_r\otimes e_r),
\end{align*}
or we can write
\begin{align*}
    (\nabla^2f_h(x))_{ij}=u''(r)\frac{x_ix_j}{r^2}+\frac{u'(r)}{r}\left(\delta_{ij}-\frac{x_ix_j}{r^2}\right).
\end{align*}
So
\begin{align*}
    \frac{\partial }{\partial x_k} (\nabla^2f_h(x))_{ij}
    &=\left(u'''(r)-\frac{3u''(r)}{r}+\frac{3u'(r)}{r^2}\right)\frac{x_ix_jx_k}{r^3}
    +\left(\frac{u''(r)}{r}-\frac{u'(r)}{r^2}\right)
    \left(\delta_{ik}\frac{x_j}{r}+\delta_{jk}\frac{x_i}{r}+\delta_{ij}\frac{x_k}{r}\right)\\
    &=u'''(r) \frac{x_ix_jx_k}{r^3}+\left(\frac{u''(r)}{r}-\frac{u'(r)}{r^2}\right)
    \left[\frac{x_i}{r}\left(\delta_{jk}-\frac{x_jx_k}{r^2}\right)
    +\frac{x_j}{r}\left(\delta_{ki}-\frac{x_kx_i}{r^2}\right)
    +\frac{x_k}{r}\left(\delta_{ij}-\frac{x_ix_j}{r^2}\right)\right].
\end{align*}
Since
\begin{align*}
    \sum_{i,j,k}(x_ix_jx_k)^2=r^6,\quad
    \sum_k \left(\delta_{jk}-\frac{x_jx_k}{r^2}\right)x_k
    = \sum_i \left(\delta_{ki}-\frac{x_kx_i}{r^2}\right)x_i
    = \sum_j \left(\delta_{ij}-\frac{x_ix_j}{r^2}\right)x_j=0,
\end{align*}
for any $v,w,s\in e_r^{\bot }$ with $\|v\|=\|w\|=\|s\|=1$, we have
\begin{align*}
    \nabla^3 f_h(x)[e_r,e_r,e_r]=&\sum_{i,j,k} \frac{\partial }{\partial x_k} (\nabla^2f_h(x))_{ij}\frac{x_ix_jx_k}{r^3} =u'''(r),\\
    \nabla^3 f_h(x)[e_r,e_r,v]=&\sum_{i,j,k} \frac{\partial }{\partial x_k} (\nabla^2f_h(x))_{ij}\frac{x_ix_j}{r^2}v_k=0,\\
    \nabla^3 f_h(x)[e_r,v,w]=&\sum_{i,j,k} \frac{\partial }{\partial x_k} (\nabla^2f_h(x))_{ij}\frac{x_i}{r}v_jw_k
    =\left(\frac{u''(r)}{r}-\frac{u'(r)}{r^2}\right)\langle v,w\rangle,\\
    \nabla^3 f_h(x)[v,w,s]=&\sum_{i,j,k} \frac{\partial }{\partial x_k} (\nabla^2f_h(x))_{ij}v_iw_js_k=0.
\end{align*}
Therefore, in order to guarantee that 
$$
\|\nabla^{2} f_h(x) - \nabla^{2} f_h(y)\| \leq L_h\|x-y\|, \quad \forall x,y\in\mathbb{R}^d,
$$
it suffices to ensure that the third–order derivative tensor is uniformly 
bounded on the only two nonzero families of directions, i.e.,
$$
|u'''(r)|\leq L_h,
\quad
\left|\frac{u''(r)}{r}-\frac{u'(r)}{r^{2}}\right|\leq L_h,
$$
for $r> \max\{N, R_2\}$.

\end{proof}

\subsection{Quantum Algorithm for Gibbs Expectation of Heavy-Tailed distributions}\label{sec:(subsec)quantum sampling-heavy}

We now apply the transformed dynamics to estimate Gibbs expectations under heavy-tailed target distributions.
Recall that, for a transformation $h$, the transformed potential and
distribution are given by
\[
f_h(y)=f(h(y))-\log|\det \nabla h(y)|,
\qquad
\pi_h(y)\propto e^{-f_h(y)}.
\]
The transformed observable is
\[
\varphi_h=\varphi\circ h.
\]
Therefore, by the change of variables $x=h(y)$,
\[
\mathbb E_{\pi}[\varphi(X)]
=
\mathbb E_{\pi_h}[\varphi_h(Y)].
\]
Hence it suffices to estimate the transformed Gibbs expectation under
$\pi_h$.

We first explain how to implement the gradient oracle for the transformed potential \(f_h\). By the chain rule, \[ \nabla f_h(y) = J_h(y)^\top \nabla f(h(y)) - \nabla \log |\det J_h(y)|, \] where \(J_h(y)=\nabla h(y)\).
By querying the gradient oracle $\mathcal{O}_{\nabla f}$,
the query oracle $\mathcal{O}_h$ for the transformation function,
the correction oracle $\mathcal{O}_a$,
and the Jacobian--vector product oracle $\mathcal{O}_m$
introduced in Section~\ref{sec:(sub)oracles}, each exactly once,
we can construct the oracle $\mathcal{O}_{\nabla f_h}$ for $\nabla f_h$ as follows:
\begin{align*}
    |y\rangle|0\rangle|0\rangle
    \xrightarrow{O_h}&|y\rangle|h(y)\rangle|0\rangle \\
    \xrightarrow{O_{\nabla f}}&|y\rangle|h(y)\rangle|\nabla f(h(y))\rangle \\
    \xrightarrow{O_m}&|y\rangle|h(y)\rangle|J_h(y)^\top \nabla f(h(y))\rangle \\
    \xrightarrow{O_a}&|y\rangle|h(y)\rangle|J_h(y)^\top \nabla f(h(y)) - \nabla \mathrm{log}|\det J_h(y)|\rangle \\
    =&|y\rangle|h(y)\rangle|\nabla f_h(y)\rangle.
\end{align*}

Using this oracle, the transformed Langevin recursion is given by \[ Y_{n+1} = Y_n-\eta \nabla f_h(Y_n)+\sqrt{2\eta}\,\xi_n, \] where \((\xi_n)_{n\ge0}\) are i.i.d. standard Gaussian random vectors in \(\mathbb R^d\), and \(\eta>0\) is the time stepsize. The corresponding quantum sampling procedure is summarized below.

\begin{algorithm}[H]
\caption{\algname{Quantum Transformed Langevin Algorithm}}
\label{alg:Quantum Transformed Langevin Algorithm}
\begin{algorithmic}[1]
\item  \textbf{Input}: Oracle $O_{\nabla f_h}$, initial state $ \ket{y_0} $, step size $ \eta $\\
 \textbf{Output}: Quantum sampling states $ \ket{x_1},\ket{x_2},\cdots$
\STATE \textbf{For} $n = 1,2,\cdots$
\begin{tightitemize}
        \item Sample seed $s_{n}$ and query $O_W$:
$|s_{n}\rangle|0\rangle \mapsto |s_{n}\rangle |\xi_{s_{n}}\rangle$
    \item Query $O_{\nabla f_h}$: $\ket{y_n}\ket{0} \mapsto \ket{y_n}\ket{\nabla f_h(y_n)}$
    \item Apply arithmetic unitary:
$       |y_{n+1}\rangle = |y_n - \eta \nabla f_h(y_n) + \sqrt{2\eta} \xi_{s_{n}}\rangle$
    \item Query $O_h$ to obtain final state:  $|x_{n+1}\rangle=|h(y_{n+1})\rangle$
\end{tightitemize}
    \textbf{End for}
\STATE \textbf{Return}: $ \ket{x_1},\ket{x_2},\cdots$
\end{algorithmic}
\end{algorithm}

We now combine the transformed dynamics with the unbiased QA-MLMC framework. By applying Theorem~\ref{thm:unbiased A-MLMC under weaker one-sided Lipschitz},
we can obtain an unbiased and accurate estimate of $\mathbb{E}_\pi[\varphi]$.

\begin{theorem}[Unbiased QA-MLMC for transformed heavy-tailed Gibbs expectations]
\label{thm:QA-MLMC of heavy-tailed}
Suppose that the potential function \(f\) satisfies
Assumptions~\ref{assumption:smooth},~\ref{assumption:dissipativity},
and~\ref{assumption:hessian-smooth}, so that the transformed potential
\[
f_h(y)=f(h(y))-\log\det(\nabla h(y))
\]
is dissipative, \(L_h\)-smooth, \(L_{h}\)-Hessian smooth, and satisfies
the weaker one-sided Lipschitz condition with constant \(\lambda_h\).
Assume that \(\varphi_h\) is globally Lipschitz.
 Then, by
Theorem~\ref{thm:unbiased A-MLMC under weaker one-sided Lipschitz},
\[
\mathbb E_\pi[\varphi(X)]
=
\mathbb E_{\pi_h}[\varphi_h(Y)]
\]
can be estimated unbiasedly with variance at most $\epsilon^2$. The expected total query complexity is
\[
\widetilde{\mathcal O}\left(
\frac{S L_h d(L_h+\sqrt d)}{2S-\lambda_h}\epsilon^{-1}
\right),
\]
where \(S>\lambda_h/2\).
\end{theorem}

\subsection{Examples of Heavy-tailed Distributions}\label{subsec:example of heavy}
First, we consider two common examples of heavy-tailed distributions: the symmetric stable distribution and the Student’s $t$-distribution. 
In these two cases, we set $\alpha=0$ and $\beta=2$, under which Assumptions~\ref{assumption:smooth}, \ref{assumption:dissipativity} and \ref{assumption:hessian-smooth}  
can be simplified as follows.
\begin{assumption}[Special Case of Assumption~\ref{assumption:smooth}]
    \label{ass:special smooth}
There exist constants $N, L > 0$  such that for all  $r > \max\{N,R_2\}$, it holds that
    \begin{align*}
    &\left|2b\psi(r) f'(\psi(r)) - 2bd + (d-2)r^{-2} \right| < L,\\
    &\left|4b^2r^2\psi^2(r)f''(\psi(r)) + (2b+4b^2r^2)\psi(r) f'(\psi(r))-2 bd- (d-2)r^{-2}  \right|< L,
\end{align*}
where $\psi(r) = e^{br^2} $ for all $r \geq R_2$.
\end{assumption}
\begin{assumption}[Special Case of Assumption~\ref{assumption:dissipativity}]
     \label{ass:special dissipative}
There exist constants $A, B, N > 0$ such that for all $r > \max\{N,R_2\}$:
 \begin{align*}
   2br^2\psi(r) f'(\psi(r)) - 2bdr^2 + (d-2) > Ar^2 - B,
\end{align*}
where $\psi(r) =  e^{br^2}$ for all $r \geq R_2$.
\end{assumption}
\begin{assumption}[Special Case of Assumption~\ref{assumption:hessian-smooth}]
     \label{ass:special hessian-smooth}
There exist constants $N, L > 0$  such that for all  $r > \max\{N,R_2\}$, it holds that
    \begin{align*}
    &\left|(12b^2r+8b^3r^3)\psi(r) f'(\psi(r))+(12b^2r+24b^3r^3)\psi^2(r)f''(\psi(r))
    +8b^3r^3\psi^3(r)f'''(\psi(r))+\frac{2(d-2)}{r^3}\right| < L,\\
    &\left|4b^2r\psi(r) f'(\psi(r))+4b^2r\psi^2(r)f''(\psi(r))
    -\frac{2(d-2)}{r^3}\right| < L,
\end{align*}
where $\psi(r) = e^{br^2} $ for all $r \geq R_2$.
\end{assumption}

\subsubsection{Symmetric Stable Distribution}

\begin{definition}
A (non-degenerate) distribution $F$ is stable if, for all $n \in \mathbb{N}$, 
there exist
 constants $c_n > 0$ and dn such that, whenever $X_1,X_2,\cdots,X_n$, 
 and $X$ are independently and
 identically distributed with distribution function $F$, 
 the sum $X_1 +\cdots+X_n$ is distributed as
$ c_nX +d_n$.
\end{definition}

 The coefficient $c_n$ must take the form $n^{1/a}$ for some $a\in (0,2]$.
  The exponent $a$ is called the index or
   characteristic exponent of the distribution.

    The characteristic function $\phi(t)$ of a symmetric stable distribution $F$ 
    has the form
    $$\phi(t)=e^{-c|t|^a}.$$

     In general the characteristic function of a stable distribution takes the form
     \begin{align*}
      \phi(t)=\begin{cases}
        \text{exp}(-\sigma^a|t|^{a}[1+i\beta(\tan \frac{\pi}{2})(\text{sign}t)(|\sigma t|^{1-a}-1)]+i\mu t)&a\neq 1,\\
               \text{exp}(-\sigma|t|[1+i\beta\frac{2}{\pi}(\text{sign}t)\log(\sigma|t|)]+i\mu t)&a= 1.
              \end{cases}
     \end{align*}

Let $X$ be symmetric $a$-stable ($0<a\le2$) with characteristic function
$\phi(t)=e^{-c|t|^a}$, $c>0$. By the inverse Fourier representation, we have
$$
\pi(x)=\frac{1}{\pi}\int_0^\infty e^{-c t^a}\cos(tx) \d t,\qquad
\pi'(x)=-\frac{1}{\pi}\int_0^\infty te^{-c t^a}\sin(tx) \d t.
$$

\begin{proposition}
    As $|x|\to\infty$, we have
    \begin{align*}
        \pi(x)&\sim \frac{c\Gamma(1+a)\sin(\pi a/2)}{\pi}|x|^{-(1+a)},\\
\pi'(x)&\sim -\frac{c\Gamma(2+a)\sin(\pi a/2)}{\pi}
\frac{\operatorname{sign}(x)}{|x|^{2+ a}},\\
\pi''(x)&\sim \frac{c\Gamma(3+a)\sin\frac{\pi a}{2}}{\pi}|x|^{-(3+a)},\\
\pi'''(x)&\sim -\frac{c\Gamma(4+a)\sin\bigl(\tfrac{\pi a}{2}\bigr)}{\pi}
\operatorname{sign}(x)|x|^{-(4+a)}.
    \end{align*}
\end{proposition}

The proof of this proposition is provided in Appendix~\hyperref[Appendix B: Technical Details for heavy tailed]{B}; see Lemmas~\ref{lem:pi}--\ref{lem:pi'''} for details.

Denote $f(x)=-\log \pi(x)$, then $\pi(x)=e^{-f}$.
\begin{proposition}
  As $x\rightarrow\infty$, 
$$
f(x)= (1+a)\log x+O(1).
$$
\end{proposition}

\begin{proof}
For convenience, denote $C_1=\frac{c\Gamma(1+a)\sin(\pi a/2)}{\pi}$, which is a constant.
We can write 
$$
\pi(x)=C_1 x^{-(a+1)}+o(x^{-(a+1)}).
$$
Therefore,
\begin{align*}
f(x)=-\log \pi(x)
  &= -\log\Big(C_1 x^{-(a+1)}+o(x^{-(a+1)})\Big)\\
  &= -\log\Big(C_1 x^{-(a+1)}\Big)
     -\log\left(1+\frac{o(x^{-(a+1)})}{C_1 x^{-(a+1)}}\right).
\end{align*}
As $x\to\infty$, we have 
$$
\frac{o(x^{-(a+1)})}{C x^{-(a+1)}}\to 0,
$$
so that 
$$
\log\left(1+\frac{o(x^{-(a+1)})}{C x^{-(a+1)}}\right)\to 0.
$$
It follows that
$$
f(x)=-\log\left(C x^{-(a+1)}\right)+o(1)=(1+a)\log x-\log C+o(1)=(1+a)\log x+O(1).
$$
\end{proof}

\begin{proposition}
As $|x|\to\infty$,
$$
f'(x)\sim\frac{1+a}{x},\quad f''(x)\sim-\frac{1+a}{x^2},
\quad f'''(x)\sim\frac{2(1+a)}{x^3}.
$$
\end{proposition}

\begin{proof}
For $f'$, since $f'(x)=-\frac{\pi'(x)}{\pi(x)}$,
\begin{align*}
   \frac{f'(x)}{x^{-1}}
    =&(1+a)\left(-\pi(x)\cdot \frac{\pi}{c\Gamma(2+a)\sin(\pi a/2)}
\frac{|x|^{2+ a}}{\operatorname{sign}(x)}\right)\cdot\left(\frac{1}{\pi(x)} \cdot\frac{c\Gamma(1+a)\sin(\pi a/2)}{\pi}|x|^{-(1+a)}\right)\\
\to & (1+a)  \text{ as } |x|\to \infty.
\end{align*}
So $$f'(x)\sim\frac{1+a}{x}$$
as $|x|\to \infty $.

For $f''$,  since
  \begin{align*}
    f''(x)=-\frac{\pi''(x)}{\pi(x)}+\left(\frac{\pi'(x)}{\pi(x)}\right)^2,
  \end{align*}
similarly,  we have
  \begin{align*}
    f''(x)\sim-\frac{(1+a)(2+a)}{x^2}+\frac{(1+a)^2}{x^2}=-\frac{1+a}{x^2}.
  \end{align*}

  For $f'''$,  since
  \begin{align*}
    f'''(x)=-\frac{\pi'''(x)}{\pi(x)}
  + 3\frac{\pi'(x)\pi''(x)}{\pi(x)^2}
  - 2\frac{\pi'(x)^3}{\pi(x)^3},
  \end{align*}
  we have
  \begin{align*}
    f'''(x)\sim\frac{(1+a)(2+a)(3+a)}{x^3}-3\frac{(1+a)^2(2+a)}{x^3}+2\frac{(1+a)^3}{x^3}
    =\frac{2(1+a)}{x^3}.
  \end{align*}
\end{proof}

Next, we verify that the symmetric stable distribution fulfills the assumptions required by the algorithm.
\begin{proposition}
There exist constants $N, L > 0$  such that for all  $r > \max\{N,R_2\}$, it holds that
    \begin{align*}
    &2b\psi(r) f'(\psi(r)) - 2b -r^{-2} < L,\\
    &4b^2r^2\psi^2(r)f''(\psi(r)) + (2b+4b^2r^2)\psi(r) f'(\psi(r))-2 b+r^{-2} < L,
\end{align*}
where $\psi(r) = e^{br^2} $ for all $r \geq R_2$, i.e.,
Assumption~\ref{ass:special smooth} holds.
\end{proposition}
 
\begin{proof}
    First,
\begin{align*}
    2b\psi(r) f'(\psi(r)) - 2b -r^{-2}
    =2b(1+a)-2b-r^{-2}+o(1).
\end{align*}
This can be bounded by a constant.

For the second term, since
\begin{align*}
        4b^2r^2\psi^2(r)f''(\psi(r)) + (2b+4b^2r^2)\psi(r) f'(\psi(r))
        \to -4b^2r^2(1+a)+(2b+4b^2r^2)(1+a)=2b(1+a)
\end{align*}
as $r\to \infty$,
 it can also be controlled by a constant. 
\end{proof}

\begin{proposition}
There exist constants $A, B, N > 0$ such that for all $r > \max\{N,R_2\}$:
 \begin{align*}
   2br^2\psi(r) f'(\psi(r)) - 2br^2 -1 > Ar^2 - B,
\end{align*}
where $\psi(r) =  e^{br^2}$ for all $r \geq R_2$,
i.e.,
Assumption~\ref{ass:special dissipative} holds.
\end{proposition}

\begin{proof}
  Since 
  \begin{align*}
    2br^2\psi(r) f'(\psi(r)) - 2br^2 -1 \to 2bar^2-1,\quad r\to\infty,
  \end{align*}
  the conclusion holds.
\end{proof}

\begin{proposition}
There exist constants $N, L > 0$  such that for all  $r > \max\{N,R_2\}$, it holds that
    \begin{align*}
    &(12b^2r+8b^3r^3)\psi(r) f'(\psi(r))+(12b^2r+24b^3r^3)\psi^2(r)f''(\psi(r))
    +8b^3r^3\psi^3(r)f'''(\psi(r))-\frac{2}{r^3} < L,\\
    &\left|4b^2r\psi(r) f'(\psi(r))+4b^2r\psi^2(r)f''(\psi(r))
    +\frac{2}{r^3}\right| < L,
\end{align*}
where $\psi(r) = e^{br^2} $ for all $r \geq R_2$,
i.e.,
Assumption~\ref{ass:special hessian-smooth} holds.
\end{proposition}
\begin{proof}
    As $r\to \infty$,
    \begin{align*}
      &  (12b^2r+8b^3r^3)\psi(r) f'(\psi(r))+(12b^2r+24b^3r^3)\psi^2(r)f''(\psi(r))
    +8b^3r^3\psi^3(r)f'''(\psi(r))-\frac{2}{r^3}\\
\to&  (12b^2r+8b^3r^3)(1+a)-(12b^2r+24b^3r^3)(1+a) + 16b^3r^3(1+a)-\frac{2}{r^3}=\frac{2(d-2)}{r^3},\\
&4b^2r\psi(r) f'(\psi(r))+4b^2r\psi^2(r)f''(\psi(r))
    +\frac{2}{r^3}
    \to4b^2r(1+a)-4b^2r(1+a)+\frac{2}{r^3}=\frac{2}{r^3}.
    \end{align*}
    Therefore, the conclusion holds.
\end{proof}

\begin{corollary}\label{cor:QA-MLMC of heavy-tailed with symmetric stable} Assume that $\pi$ is a symmetric $a$-stable distribution with $0<a<2$. Let $h$ be the transformation map defined above and set $\varphi_h=\varphi\circ h$. If $\varphi_h$ is globally Lipschitz, then we can estimate $\mathbb{E}_{\pi}[\varphi]$ unbiasedly with variance at most $\epsilon^2$. The expected total query complexity is $\widetilde{\mathcal{O}}(\epsilon^{-1})$. \end{corollary}

\subsubsection{Student’s $t$-distribution}\label{subsub:student-t}

Consider the Student’s $t$-distribution with the density and potential function
\begin{align}\label{eq:student-t}
  \pi(x)\propto(1+|x|^2)^{-\frac{d+\kappa}{2}},\quad f(x)=\frac{d+\kappa}{2}\log(1+|x|^2).
\end{align}

\begin{proposition}
There exist constants $N, L > 0$  such that for all  $r > \max\{N,R_2\}$, it holds that
    \begin{align*}
    &2b\psi(r) f'(\psi(r)) - 2bd + (d-2)r^{-2} < L,\\
    &4b^2r^2\psi^2(r)f''(\psi(r)) + (2b+4b^2r^2)\psi(r) f'(\psi(r))-2 bd- (d-2)r^{-2} < L,
\end{align*}
where $\psi(r) = e^{br^2} $ for all $r \geq R_2$,
i.e.,
Assumption~\ref{ass:special smooth} holds.
\end{proposition}
\begin{proof}
    By direct computation, we have
    \begin{align*}
        &2b\psi(r) f'(\psi(r)) - 2bd + (d-2)r^{-2}
        = 2b(d+\kappa)\frac{\psi^2(r)}{1+\psi^2(r)} - 2bd + (d-2)r^{-2},\\
        &4b^2r^2\psi^2(r)f''(\psi(r)) + (2b+4b^2r^2)\psi(r) f'(\psi(r))-2 bd- (d-2)r^{-2}\\
        =&4b^2r^2(d+\kappa)\frac{(1-\psi^2(r))\psi^2(r)}{(1+\psi^2(r))^2}+(2b+4b^2r^2)(d+\kappa)\frac{\psi^2(r)}{1+\psi^2(r)}
        -2bd-(d-2)r^{-2}\\
        =&8b^2(d+\kappa)\frac{\psi^2(r)r^2}{(1+\psi^2(r))^2}+2b(d+\kappa)\frac{\psi^2(r)}{1+\psi^2(r)}-2bd-(d-2)r^{-2}.
    \end{align*}
    Both formulas are bounded as $r\to\infty$.
\end{proof}

\begin{proposition}
There exist constants $A, B, N > 0$ such that for all $r > \max\{N,R_2\}$:
 \begin{align*}
   2br^2\psi(r) f'(\psi(r)) - 2bdr^2 + (d-2) > Ar^2 - B,
\end{align*}
where $\psi(r) =  e^{br^2}$ for all $r \geq R_2$,
i.e.,
Assumption~\ref{ass:special dissipative} holds.
\end{proposition}
\begin{proof}
  Since
  \begin{align*}
     2br^2\psi(r) f'(\psi(r)) - 2bdr^2 + (d-2)
        = 2b(d+\kappa)r^2\frac{\psi^2(r)}{1+\psi^2(r)} - 2bdr^2 + (d-2)
        \to 2b\kappa r^2 +(d-2)
  \end{align*}
as $r\to\infty$, the conclusion holds.
\end{proof}

\begin{proposition}
    There exist constants $N, L > 0$  such that for all  $r > \max\{N,R_2\}$, it holds that
    \begin{align*}
    &(12b^2r+8b^3r^3)\psi(r) f'(\psi(r))+(12b^2r+24b^3r^3)\psi^2(r)f''(\psi(r))
    +8b^3r^3\psi^3(r)f'''(\psi(r))+\frac{2(d-2)}{r^3} < L,\\
    &\left|4b^2r\psi(r) f'(\psi(r))+4b^2r\psi^2(r)f''(\psi(r))
    -\frac{2(d-2)}{r^3}\right| < L,
\end{align*}
where $\psi(r) = e^{br^2} $ for all $r \geq R_2$,
i.e.,
Assumption~\ref{ass:special hessian-smooth} holds.
\end{proposition}
\begin{proof}
    First, for the first expression,
    \begin{align*}
       & (12b^2r+8b^3r^3)\psi(r) f'(\psi(r))+(12b^2r+24b^3r^3)\psi^2(r)f''(\psi(r))
    +8b^3r^3\psi^3(r)f'''(\psi(r))+\frac{2(d-2)}{r^3}\\
    =&(12b^2r+8b^3r^3)(d+\kappa)\frac{\psi^2(r)}{1+\psi^2(r)}
   +(12b^2r+24b^3r^3)(d+\kappa) \frac{\psi^2(r)(1-\psi^2(r))}{(1+\psi^2(r))^2}\\
   &\quad  +16b^3r^3(d+\kappa)\frac{\psi^4(r)(\psi^2(r)-3)}{(1+\psi^2(r))^3}+\frac{2(d-2)}{r^3}.
    \end{align*}
    As $r\to\infty$, this expression tends to
    \begin{align*}
     (12b^2r+8b^3r^3)(d+\kappa)-(12b^2r+24b^3r^3)(d+\kappa) +16b^3r^3(d+\kappa)=0.
    \end{align*}

    For the second expression,
    \begin{align*}
      & 4b^2r\psi(r) f'(\psi(r))+4b^2r\psi^2(r)f''(\psi(r))
    -\frac{2(d-2)}{r^3}\\
    =& 4b^2r (d+\kappa)\frac{\psi^2(r)}{1+\psi^2(r)}
    +4b^2r(d+\kappa) \frac{\psi^2(r)(1-\psi^2(r))}{(1+\psi^2(r))^2}
-\frac{2(d-2)}{r^3}\to 0
    \end{align*}
    as $r\to \infty$.

    Therefore, the conclusion holds.
\end{proof}

\begin{corollary}\label{cor:QA-MLMC of heavy-tailed student-t} Assume that $\pi$ is the Student’s $t$-distribution defined above. Let $h$ be the transformation map defined above and set $\varphi_h=\varphi\circ h$. If $\varphi_h$ is globally Lipschitz, then we can estimate $\mathbb{E}_{\pi}[\varphi]$ unbiasedly with variance at most $\epsilon^2$. The expected total query complexity is $\widetilde{\mathcal{O}}(\epsilon^{-1})$. \end{corollary}

To complement the theoretical result for the Student’s $t$-distribution, 
we provide a numerical illustration based on the spring--coupled Langevin sampler. 
We generate samples using the fine-path dynamics with the spring term and 
apply the corresponding change-of-measure correction to approximate the target distribution.
\begin{figure}[htbp]
    \centering
    \includegraphics[width=0.5\textwidth]{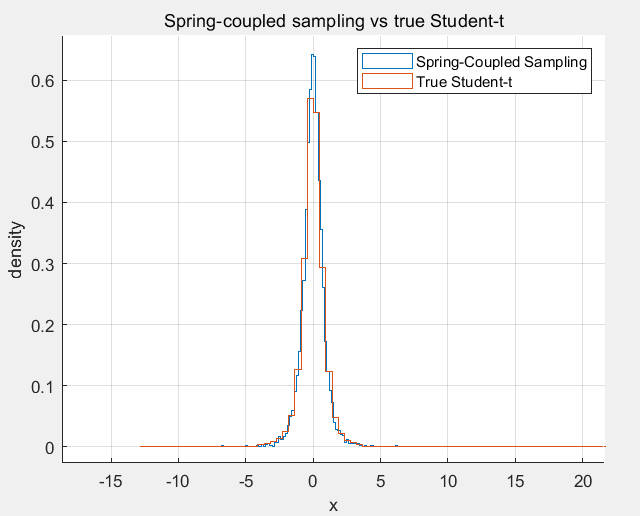}
    \caption{Histogram of samples generated by the spring–coupled Langevin sampler}
    \label{fig:spring-studentt}
\end{figure}

Figure~\ref{fig:spring-studentt} shows that the empirical distribution 
is in good agreement with the true Student’s $t$-distribution, 
supporting the effectiveness of our approach in the heavy-tailed setting.

\section{Applications}\label{sec:application}
In Subsections~\ref{subsec:example of weaker one-sided} and 
\ref{subsec:example of heavy}, we presented several examples to illustrate 
which distributions satisfy the assumptions of the corresponding theorems 
under the dissipative setting 
(Theorem~\ref{thm:QA-MLMC under weaker one-sided Lipschitz}) 
and the heavy-tailed setting 
(Theorem~\ref{thm:QA-MLMC of heavy-tailed}), respectively.

In this section, we further provide concrete applications. 
Rather than only identifying admissible distributions, 
we also specify the corresponding test functions $\varphi$ 
in practical scenarios, such as statistics, machine learning, and finance.

In particular, the first four examples focus on applications of 
Theorem~\ref{thm:unbiased A-MLMC under weaker one-sided Lipschitz} 
under the dissipative condition, while the fifth example 
illustrates the application of 
Theorem~\ref{thm:QA-MLMC of heavy-tailed} 
in the heavy-tailed setting.
\subsection{Bayesian Logistic Regression}

We consider a Bayesian logistic regression model, 
which extends the classical logistic regression by incorporating 
the principle of Bayesian inference and treating the model parameters as random variables.

Given a dataset $\mathcal{D} = \{(x_i, y_i)\}_{i=1}^n$,  with $x_i \in \mathbb{R}^d$ and 
$y_i \in \{-1, +1\}$.
In a binary classification problem, 
the logistic regression model assumes that
\begin{align*}
    \mathbb{P}[y_i = 1 | x_i, \beta] = \sigma(x_i^\top \beta),
\end{align*}
where $\sigma(z) = \frac{1}{1 + e^{-z}}$ is the sigmoid function, and 
$\beta \in \mathbb{R}^d$ is the parameter vector to be inferred.

In the Bayesian framework, $\beta$ is a random variable and usually
with  prior distribution $\beta \sim \mathcal{N}(0, \Sigma)$.
The likelihood function is defined as
$$p(\mathcal{D} | \beta) = 
\prod_{i=1}^n \sigma(y_i x_i^\top \beta),$$
and the posterior distribution of $\beta$ is given by
$$p(\beta | \mathcal{D}) \propto p(\mathcal{D} | \beta)  p(\beta).$$

Consider a Gaussian prior $W \sim \mathcal{N}(0, \Sigma)$
with precision matrix $\Lambda = \Sigma^{-1} \succ 0$, and define the potential function
\begin{align*}
    f(w) = \frac{1}{2} w^\top \Lambda w 
    + \sum_{i=1}^n \ell(y_i x_i^\top w),
    \quad \ell(z) = \log(1 + e^{-z}).
\end{align*}
We aim to estimate the Bayesian predictive probability 
$$\mu = \mathbb{E}_\pi[\varphi(W)],$$
where $\pi \propto e^{-f(w)}$, 
$\varphi(w) = \sigma(x_{\mathrm{new}}^\top w)$ with $\sigma(t)= \frac{1}{1 + e^{-t}}$,
and $x_{\mathrm{new}}$ denotes a new data point.
The overdamped Langevin diffusion is given by
\begin{align*}
    \d W_t = -\nabla f(W_t) \d t + \sqrt{2} \d B_t,
\end{align*}
whose invariant distribution is $\pi$.

Compute the gradient and Hessian of $f$, we have
\begin{align*}
\nabla f(w) 
&= \Lambda w - \sum_{i=1}^n \sigma(-y_i x_i^\top w) y_i x_i^\top, \\
\nabla^2 f(w) 
&= \Lambda + \sum_{i=1}^n \sigma(y_i x_i^\top w)(1 - \sigma(y_i x_i^\top w)) y_i x_i^\top.
\end{align*}

Denote $m = \lambda_{\min}(\Lambda)$ 
and $M = \lambda_{\max}(\Lambda)$ 
as the smallest and largest eigenvalues of $\Lambda$, 
and assume that the data vectors $x_i$ are bounded, i.e., $\|x_i\| \leq R$ for all $i$.

\begin{proposition}[Dissipativity]
There exist constants $a>0$ and $b\geq0$ such that
$$
\langle w, \nabla f(w)\rangle \geq a \|w\|^2 - b.
$$
More precisely, $a = \frac{m}{2}, 
b = \frac{1}{2m}\left(\sum_{i=1}^n \|x_i\|\right)^2$.
\end{proposition}

\begin{proof}
Since
\begin{align*}
    \langle w, \nabla f(w)\rangle 
    = w^\top \Lambda w - \sum_i \sigma(-y_i x_i^\top w) y_i x_i^\top w
    \geq m\|w\|^2 - \sum_i|y_i x_i^\top w|,
\end{align*}
using Young's inequality,
\begin{align*}
    \langle w, \nabla f(w)\rangle  
    \geq \frac{m}{2}\|w\|^2 - \frac{1}{2m}\left(\sum_{i=1}^n \|x_i\|\right)^2.
\end{align*}
\end{proof}

\begin{proposition}[Smoothness]
The gradient $\nabla f$ and Hessian $\nabla^2f$ are Lipschitz continuous, i.e.,
\begin{align*}
    \|\nabla f(w) - \nabla f(v)\| \leq \left(M + \frac{nR^2}{4}\right) \|w - v\|,\\
\|\nabla^2 f(w) - \nabla^2 f(v)\| \leq \frac{nR^3}{6\sqrt{3}} \|w - v\|.
\end{align*}
\end{proposition}

\begin{proof}
Since $\sigma(a)\sigma(-a) = \sigma(a)(1 - \sigma(a)) \in (0, 1/4]$,
\begin{align*}
    \|\nabla^2 f(w)\|
    \leq \|\Lambda\| + \frac{1}{4}\sum_{i=1}^n \|x_i\|^2
    = M + \frac{1}{4}\sum_{i=1}^n \|x_i\|^2.
\end{align*}
By the mean value theorem,
\begin{align*}
    \|\nabla f(w) - \nabla f(v)\|
    \leq \sup_{\theta} \|\nabla^2 f(\theta)\| \|w - v\|\leq \left(M + \frac{nR^2}{4}\right) \|w - v\|.
\end{align*}

On the other hand,
\begin{align*}
    \nabla^2 f(w) - \nabla^2 f(v)
    = \sum_{i=1}^n \big(g(y_i x_i^\top w) - g(y_i x_i^\top v)\big)x_i x_i^\top,
\end{align*}
where $g(z) = \sigma(z)\sigma(-z)$.
Since
\begin{align*}
    g'(z) = \sigma(z)\sigma(-z)(1 - 2\sigma(z))
\end{align*}
and $\max_z |g'(z)| = \frac{1}{6\sqrt{3}}$,
the mean value theorem yields
\begin{align*}
    |g(y_i x_i^\top w) - g(y_i x_i^\top v)|
    \leq \frac{1}{6\sqrt{3}}|y_i x_i^\top (w - v)|
    \leq \frac{R}{6\sqrt{3}}\|w - v\|.
\end{align*}
Thus
\begin{align*}
    \|\nabla^2 f(w) - \nabla^2 f(v)\|
\leq \frac{R}{6\sqrt{3}}
\sum_{i=1}^n \|x_i\|^2 \|w - v\|\leq \frac{nR^3}{6\sqrt{3}}
  \|w - v\|.
\end{align*}
\end{proof}

\begin{proposition}[Properties of $\varphi$]
The function $\varphi(w) = \sigma(x_{\mathrm{new}}^\top w)$ is globally Lipschitz, i.e.,
\begin{align*}
    |\varphi(w) - \varphi(v)| \leq \frac{\|x_{\mathrm{new}}\|}{4} \|w - v\|.
\end{align*}
\end{proposition}
\begin{proof}
Note that $\sigma'(t)=\sigma(t)\left(1-\sigma(t)\right)\leq \frac{1}{4}$, 
by the chain rule,
\begin{align*}
    \nabla \varphi(w) =\sigma'\left(x_{\mathrm{new}}^\top w\right) x_{\mathrm{new}}
    \Rightarrow 
    \|\nabla \varphi(w)\| \leq \frac{1}{4}\|x_{\mathrm{new}}\|.
\end{align*}
Applying the mean value theorem,
there exists $\theta$ on the segment $[v,w]$ such that
\begin{align*}
    |\varphi(w)-\varphi(v)|
= |\langle \nabla \varphi(\theta), w-v\rangle|
\leq \|\nabla \varphi(\theta)\|\|w-v\|
\leq \frac{\|x_{\mathrm{new}}\|}{4}\|w-v\|.
\end{align*}
\end{proof}

Therefore, by applying Theorem~\ref{thm:unbiased A-MLMC under weaker one-sided Lipschitz},
 we can estimate unbiasedly with an expected complexity of $\widetilde{\mathcal{O}}(\epsilon^{-1})$.

\subsection{Bayesian Mixture Modeling }
In fact, in the previous example, the function $f$ is $m$-strongly convex.
We now consider a Bayesian model with a mixture-of-Gaussians prior and a logistic likelihood~\cite{Neal1992},
whose potential function $f$ is not strongly convex
but still satisfies the dissipativity condition and a weaker one-sided Lipschitz property.
This model provides a natural testbed for analyzing convergence and variance
under non-log-concave settings.

Let $\{(x_i,y_i)\}_{i=1}^n$ with $x_i\in\mathbb{R}^d$, $y_i\in\{-1,+1\}$.
Fix a positive definite covariance $\Sigma\succ0$ and denote $\Lambda=\Sigma^{-1}$.
Consider a finite Gaussian mixture prior with weights $(p_k)_{k=1}^K$, $p_k>0$, $\sum_k p_k=1$,
and component means $\{m_k\}_{k=1}^K\subset\mathbb{R}^d$.
Define the potential function
\begin{align*}
f(w)
= -\log\left(\sum_{k=1}^K p_k \exp\left(-\frac{1}{2} (w-m_k)^\top \Lambda (w-m_k)\right)\right)
+ \sum_{i=1}^n \ell(y_i x_i^\top w),
\end{align*}
where $\ell(z)=\log(1+e^{-z})$.
We aim  to estimate the expectation $\mu=\mathbb{E}_\pi[\varphi(W)]$ with
$\varphi(w)=\sigma(x_{\mathrm{new}}^\top w)$ and $\sigma(t)=\frac{1}{1+e^{-t}}$.

Assume that the component means and the feature vectors are all bounded:
$$\max_k \|m_k\| \leq M, \text{ and } \|x_i\| \leq R \text{ for all } i.$$
Denote 
\begin{align*}
    f_{\mathrm{mix}}(w)=-\log\left(\sum_{k=1}^K p_k \exp\left(-\frac{1}{2} (w-m_k)^\top \Lambda (w-m_k)\right)\right),
    \quad
g(w)=\sum_{i=1}^n \ell(y_i x_i^\top w),
\end{align*}
and define
\begin{align*}
    a_k(w)=p_k \exp\left(-\tfrac12 (w-m_k)^\top\Lambda(w-m_k)\right),\quad
\pi_k(w)=\frac{a_k(w)}{\sum_{k=1}^K a_k(w)}.
\end{align*}
Then the  gradient and Hessian can be computed as
\begin{align*}
\nabla f_{\mathrm{mix}}(w)
&= \Lambda\left(w-\bar m(w)\right),\\
\nabla^2 f_{\mathrm{mix}}(w)
&= \Lambda - \Lambda\sum_k \pi_k(w)(m_k-\bar m(w))(m_k-\bar m(w))^\top\Lambda,
\end{align*}
with $\bar m(w)=\sum_k \pi_k(w)m_k$.

\begin{proposition}[Smoothness]
    The gradient $\nabla f$ and Hessian $\nabla^2f$ are Lipschitz continuous, i.e.,
\begin{align*}
     \|\nabla f(w)-\nabla f(v)\|\leq \left(\|\Lambda\|+\|\Lambda\|^2M^2+\frac14\sum_{i=1}^n\|x_i\|^2\right)\|w-v\|,\\
   \|\nabla^2 f(w)-\nabla^2 f(v)\|\leq \left(8\|\Lambda\|^3M^3+\frac{nR^3}{6\sqrt3}\right)\|w-v\|,
 \end{align*}
The first bound includes the covariance term of the Gaussian mixture.
Indeed, the mixture Hessian is $\Lambda-\operatorname{Cov}_{\pi(w)}(\Lambda m_k)$,
and the covariance norm is at most $\|\Lambda\|^2M^2$.
The third derivative of the mixture log-normalizer is a third centered
moment of $\Lambda m_k$, bounded by $8\|\Lambda\|^3M^3$;
the logistic term contributes at most $nR^3/(6\sqrt3)$.
\end{proposition}

\begin{proposition}[Dissipativity]
There exist $a>0$ and $b\geq0$ such that
\begin{align*}
    \langle w,\nabla f(w)\rangle \geq a\|w\|^2 - b,\quad \forall w\in\mathbb{R}^d.
\end{align*}
\end{proposition}

\begin{proof}
Since $\|\bar m(w)\|\leq \max_k\|m_k\|\leq M$, we have
\begin{align*}
    \langle w,\nabla f_{\mathrm{mix}}(w)\rangle
    = w^\top \Lambda (w-\bar m(w)) 
    \geq &\lambda_{\min}(\Lambda)\|w\|^2 - \|\Lambda\|\|w\|\|\bar m(w)\|\\
    \geq& \lambda_{\min}(\Lambda)/2\|w\|^2-\|\Lambda\|^2 M^2/(2\lambda_{\min}(\Lambda)).
\end{align*}
Combining with $\langle w, \nabla g(w)\rangle \geq -nR\|w\|$ 
and using Young's inequality to absorb the linear term into a quadratic term yields
\begin{align*}
   \langle w,\nabla f(w)\rangle \geq a\|w\|^2 - b.
\end{align*}
\end{proof}

Consequently, by applying Theorem~\ref{thm:unbiased A-MLMC under weaker one-sided Lipschitz}, 
the desired quantity can be estimated unbiasedly with an  expected complexity of $\widetilde{\mathcal{O}}(\epsilon^{-1})$

\subsection{Robust Regression with Correntropy Loss}
In many real-world applications such as computer vision, wireless localization, and sensor calibration, the data may contain outliers or exhibit heavy-tailed noise, making classical quadratic losses unstable.
To address this, the Welsch (or correntropy) loss
\begin{align*}
     \phi(t) = 1 - \exp\left(-\frac{t^2}{2\sigma^2}\right)
\end{align*}
is often employed in place of the least-squares penalty.
It behaves quadratically for small residuals but saturates for large ones,
thereby effectively suppressing the influence of outliers.

Given data $(x_i,y_i)_{i=1}^n$ with $x_i\in\mathbb{R}^d$, $y_i\in\mathbb{R}$
Consider the regularized empirical risk 
\begin{align*}
 f(w) = \frac{1}{n}\sum_{i=1}^n \phi\left(y_i - x_i^\top w\right) 
 + \frac{\lambda_0}{2}\|w\|^2,
\end{align*}
and the  corresponding Langevin dynamics is
 \begin{align*}
    \d X_t = -\nabla f(X_t)\d t + \sqrt{2}\d W_t.
 \end{align*}

In this example, 
$\E_{\pi}[\varphi]$ 
represents the robust average of the parameters 
or the steady-state expectation or variance 
of model predictions. 
For instance, 
$\varphi(w)=x_{\text{test}}^\top w$ 
corresponds to the posterior mean prediction at a new test point. 
Thus, the Langevin-based estimator provides a robust alternative to pointwise optimization, 
yielding stable and uncertainty-aware regression predictions.

We assume the data are all bounded, i.e. there exists $R>0$ such that
\begin{align*}
\max_{1\leq i\leq n}\|x_i\| \leq R,\quad \max_{1\leq i\leq n}|y_i| \leq R.
\end{align*}

First, we can compute directly that
\begin{align*}
  \phi'(t) &= \frac{t}{\sigma^2}e^{-t^2/(2\sigma^2)},\\
  \phi''(t) &= \frac{1}{\sigma^2}e^{-t^2/(2\sigma^2)}\left(1-\frac{t^2}{\sigma^2}\right),\\
  \phi'''(t) &= -\frac{t}{\sigma^4}e^{-t^2/(2\sigma^2)}\left(3-\frac{t^2}{\sigma^2}\right).
\end{align*}
Note that 
\begin{align*}
    \lim_{s\to \infty}|s e^{-s^2/2}(3-s^2)|=0,
\end{align*}
there exists a constant $C>0$ such that $|s e^{-s^2/2}(3-s^2)|\leq C$ for all $s$.
Hence $|\phi'''(t)|\leq\frac{C}{\sigma^3}$.

\begin{proposition}[Dissipativity]
    For all $w\in\mathbb{R}^d$, we have
\begin{align*}
    \langle w,\nabla f(w)\rangle\geq \tfrac{\lambda_0}{2}\|w\|^2 - \tfrac{R^4}{2\lambda_0\sigma^4}.
\end{align*}
\end{proposition}
\begin{proof}
Denote $t_i(w)=y_i-x_i^\top w$ and $g_i(w)=\nabla_w\phi\left(t_i(w)\right)=-\phi'(t_i)x_i$. 
Then
$$
  \langle w,\nabla f(w)\rangle=\frac{1}{n}\sum_{i=1}^n \langle w,g_i(w)\rangle
  +\lambda_0\|w\|^2.
$$
Using $|\phi'(t)|\leq |t|/\sigma^2$ and $|t_i(w)|\leq |y_i|+\|x_i\|\|w\|\leq R(1+\|w\|)$,
we have
\begin{align*}
      |\langle w,g_i(w)\rangle|\leq \|w\|\|x_i\||\phi'(t_i)|
  \leq \frac{R}{\sigma^2}\|w\||t_i(w)|
  \leq \frac{R^2}{\sigma^2}\|w\|+\frac{R^2}{\sigma^2}\|w\|^2.
\end{align*}

Hence
\begin{align*}
   \langle w,\nabla f(w)\rangle\geq 
   \left(\lambda_0-\frac{R^2}{\sigma^2}\right)\|w\|^2 - \frac{R^2}{\sigma^2}\|w\|.
\end{align*}

Using the inequality
$\frac{R^2}{\sigma^2}\|w\|\leq \frac{\lambda_0}{2}\|w\|^2+\frac{R^4}{2\lambda_0\sigma^4}$,
we have
\begin{align*}
    \langle w,\nabla f(w)\rangle\geq \tfrac{\lambda_0}{2}\|w\|^2 - \tfrac{R^4}{2\lambda_0\sigma^4}.
\end{align*}
\end{proof}

\begin{proposition}[$L$-Smoothness]
    For any $w,v\in\mathbb{R}^d$, we have
\begin{align*}
      \|\nabla f(w)-\nabla f(v)\|\leq \left(\lambda_0+\frac{R^2}{\sigma^2}\right)\|w-v\|.
\end{align*}
\end{proposition}
\begin{proof}
 Note that
$$\nabla^2 f(w)=\frac{1}{n}\sum_i \phi''(t_i(w))x_i x_i^\top + \lambda_0 I,$$
 Since $|\phi''|\leq \sigma^{-2}$ and $\|x_i x_i^\top\|\leq \|x_i\|^2\leq R^2$, 
 we obtain the uniform bound $$\|\nabla^2 f(w)\|\leq \lambda_0+\frac{R^2}{\sigma^2}$$
  for all $w$. 
 The mean value theorem gives the claim.
\end{proof}

\begin{proposition}[$L$-Hessian-Smoothness]
For any $w,v\in\mathbb{R}^d$, we have
\begin{align*}
    \|\nabla^2 f(w)-\nabla^2 f(v)\|\leq \frac{CR^3}{\sigma^3}\|w-v\|.
\end{align*}
\end{proposition}
\begin{proof}
Differentiate once more gives
 $$\nabla^3 f(w)=\frac{1}{n}\sum_i \phi'''(t_i(w))(-x_i)\otimes x_i x_i^\top.$$
So
$$\|\nabla^3 f(w)\|\leq \frac{1}{n}\sum_i |\phi'''(t_i(w))|\|x_i\|^3 \leq \frac{CR^3}{\sigma^3}.$$ 
Hence, for any $w,v\in\mathbb{R}^d$, we have
\begin{align*}
    \|\nabla^2 f(w)-\nabla^2 f(v)\|\leq \frac{CR^3}{\sigma^3}\|w-v\|.
\end{align*}
\end{proof}

Similarly, applying Theorem~\ref{thm:unbiased A-MLMC under weaker one-sided Lipschitz}, 
the desired quantity can be estimated unbiasedly with an   expected complexity 
of $\widetilde{\mathcal{O}}(\epsilon^{-1})$.

\begin{remark}[Failure of One-sided Lipschitz]
Pick any index $i$ and let $|t|=|y_i-x_i^\top w|>\sigma$,
then $\phi''(t)<0$. 
So for suitable data and $w$, 
the matrix $\frac{1}{n}\sum_i \phi''(t_i)x_i x_i^\top$ 
becomes negative in some direction.
 If $\lambda_0$ is small enough, 
 $\lambda_{\min}\left(\nabla^2 f(w)\right)<0$ at such $w$. 
 Therefore $\inf_{w}\lambda_{\min}\left(\nabla^2 f(w)\right)\leq 0$, 
there is no $m>0$ such that 
$\langle w-v,\nabla f(w)-\nabla f(v)\rangle\geq m\|w-v\|^2$ 
holds for all $w,v\in\mathbb{R}^d$.
This implies that in this example, 
$f$ does not satisfy the one-sided Lipschitz condition, 
and is far from being strongly convex.
\end{remark}

\subsection{Cosine-Modulated Quadratic Well}
Consider the potential
\begin{align*}
   f(x)=\frac{1}{2}\|x\|^2+\lambda_0\sum_{i=1}^d 
\cos\left(\frac{x_i}{\sqrt{d}}\right), 
\end{align*}
where $\lambda_0>0$ controls the strength of the periodic perturbation.
This type of potential function arises in various contexts, such as
Bayesian posterior landscapes with oscillatory likelihood terms,
periodic physical systems like spin models or molecular potentials,
and synthetic benchmarks used to study the mixing behavior of
Langevin dynamics in multi-well energy landscapes~\cite{Wu2025,jones2025differentiableneuralnetworkrepresentation}.

For the invariant measure $\pi(\mathrm{d}x)\propto e^{-f(x)}\mathrm{d}x,$
we are interested in estimating $\E_{\pi}[\varphi]$ for several $\varphi$.
For example, taking $\varphi(x)=x_j$ allows us to examine the mean position of the $j$-th coordinate, 
which reflects possible asymmetry among the wells of the potential. 
Moreover, $\varphi(x)=\mathbf{1}_{\{\|x\|\leq R\}}$ enables us to quantify the probability that the 
diffusion remains within a neighborhood of the origin, 
providing insight into its concentration and confinement properties.

\begin{proposition}[Dissipativity]
    For all $x\in\mathbb{R}^d$, we have
\begin{align*}
    \langle x,\nabla f(x)\rangle
         \geq \frac{1}{2}\|x\|^2 - \frac{1}{2}\lambda_0^2d.
\end{align*}
\end{proposition}
\begin{proof}
    We can first compute that
\begin{align*}
     \nabla f(x)
    &= x - \frac{\lambda_0}{\sqrt{d}}
       \left(\sin\left(\frac{x_1}{\sqrt{d}}\right),\ldots,\sin\left(\frac{x_d}{\sqrt{d}}\right)\right).
\end{align*}
Hence
    \begin{align*}
          \langle x,\nabla f(x)\rangle
      = \|x\|^2 - \frac{\lambda_0}{\sqrt{d}}
        \sum_{i=1}^d x_i\sin\left(\frac{x_i}{\sqrt{d}}\right)
      \geq \|x\|^2 - \lambda_0\sqrt{d}\|x\|.
    \end{align*}
    Using  Young's inequality 
    $\lambda_0{\sqrt{d}}\|x\|\leq \frac{1}{2}\|x\|^2 + \frac{1}{2}\lambda_0^2d$, we have
    \begin{align*}
        \langle x,\nabla f(x)\rangle
             \geq \frac{1}{2}\|x\|^2 - \frac{1}{2}\lambda_0^2d.
    \end{align*}
    So $f$ is dissipative with constants $a=\frac{1}{2}$ and $b=\frac{1}{2}\lambda_0^2d$.
\end{proof}

\begin{proposition}[$L$-Smoothness]
    For any $x,y\in\mathbb{R}^d$, we have
\begin{align*}
      \|\nabla f(x)-\nabla f(y)\|\leq (1+\lambda_0)\|x-y\|.
\end{align*}
\end{proposition}
\begin{proof}
 Since
\begin{align*}
     \nabla^2 f(x)
    &= I_d - \frac{\lambda_0}{d}\sum_{i=1}^d\cos\left(\frac{x_i}{\sqrt{d}}\right)e_ie_i^\top,
\end{align*}
we have $\|\nabla^2 f(x)\|\leq 1+\lambda_0$.
So using  mean value theorem, we know that $f$ is $L$-smooth.
\end{proof}

\begin{proposition}[$L$-Hessian-Smoothness]
For any $x,y\in\mathbb{R}^d$, we have
\begin{align*}
    \|\nabla^2 f(x)-\nabla^2 f(y)\|\leq \frac{ \lambda_0}{\sqrt{d}}\|x-y\|.
\end{align*}
\end{proposition}
\begin{proof}
Differentiating once more gives
\begin{align*}
    \nabla^3 f(x)
  = \frac{\lambda_0}{d^{\frac{3}{2}}}\sum_{i=1}^d
  \sin\left(\frac{x_i}{\sqrt{d}}\right)
  e_i\otimes e_i\otimes e_i,
\end{align*}
and thus $\|\nabla^3 f(x)\|\leq\frac{ \lambda_0}{\sqrt{d}}$,
so the Hessian is also smooth with constant $L=\frac{ \lambda_0}{\sqrt{d}}$.
\end{proof}

Similarly, applying Theorem~\ref{thm:unbiased A-MLMC under weaker one-sided Lipschitz}, 
the desired quantity can be estimated unbiasedly with an  expected complexity 
of $\widetilde{\mathcal{O}}(\epsilon^{-1})$.

\begin{remark}
    The eigenvalues of $\nabla^2 f(x)$ are
    \begin{align*}
        \lambda_i(x)=1-\frac{\lambda_0}{d}\cos\left(\frac{x_i}{\sqrt{d}}\right),
    \end{align*}
which oscillate between $1-\frac{\lambda_0}{d}$ and $1+\frac{\lambda_0}{d}$.
In particular, if $\lambda_0>d$, some $\lambda_i(x)$ become negative,
so $\nabla^2 f(x)$ is indefinite,
 $f$ fails the one-sided Lipschitz inequality
\begin{align*}
      \langle x-y,\nabla f(x)-\nabla f(y)\rangle
  \geq m\|x-y\|^2
\end{align*}
for some $m>0$.
\end{remark}

\subsection{Truncated Loss and Capped Payoffs in Finance}

Consider the heavy-tailed distribution $\pi$ given by the Student’s $t$ model defined in \eqref{eq:student-t}. 
As discussed in Subsection~\ref{subsub:student-t}, in particular Corollary~\ref{cor:QA-MLMC of heavy-tailed student-t}, 
for any globally Lipschitz function $\varphi$, the expectation $\mathbb{E}_{\pi}[\varphi]$ can be estimated by an
unbiased estimator with variance at most $\epsilon^2$ in
$\widetilde{\mathcal{O}}(\epsilon^{-1})$ complexity.

A representative example of such a function $\varphi$ in finance is the capped call option payoff, given by
$$
\varphi(x)=\min\{(\ell(x))^+,M\},
$$
where $\ell:\mathbb{R}^d\to\mathbb{R}$ is a Lipschitz loss function 
and $M>0$ is a fixed cap.

Such functionals arise naturally in several areas of finance. 
In risk management and insurance, $\ell(x)$ represents a loss 
and $\min\{(\ell(x))^+,M\}$ models truncated losses, 
reflecting capped exposure to extreme events. 
In derivative pricing, the same structure corresponds to the payoff of a capped call option,
 $\min\{(S_T-K)^+,L\}$, where $\ell(x)=S_T-K$ and $M=L$. 
Similar forms also appear in credit risk when modeling exposure-at-default,
 where losses are effectively bounded by contractual limits \cite{McNeil2005QuantitativeRM, Glasserman2003MonteCM}.

\section{Summary and Discussion}\label{sec:summary}

 In this paper, we refine existing quantum-accelerated multilevel Monte Carlo (QA-MLMC) methods 
by establishing explicit bounds on their mean-squared error (MSE) 
and reducing their dependence on the dimension in high-dimensional settings. 
In particular, our improved approach enables 
unbiased estimation in a broad range of scenarios.

The main contribution of this work is the development of several quantum algorithms 
for estimating Gibbs expectations under weakened assumptions, requiring only dissipativity. 
Compared with classical methods, these algorithms achieve a quadratic speedup, reducing the 
complexity from $\widetilde{\mathcal{O}}(\epsilon^{-2})$ to 
$\widetilde{\mathcal{O}}(\epsilon^{-1})$. Moreover, our framework provides unbiased estimators 
for $\E_\pi[\varphi]$, in contrast to existing quantum algorithms that typically yield biased 
estimators under stronger assumptions such as  one-sided Lipschitz conditions.

Furthermore, we extend our framework to heavy-tailed distributions. Through a careful design of 
transformation functions, we show that heavy-tailed distributions
 can be estimated as efficiently as light-tailed ones.
 We also present a range of examples arising from statistics, machine learning, 
financial modeling, and related areas, demonstrating the broad 
practical applicability of our methods.

This work shows how multilevel Monte Carlo methods can be systematically 
incorporated into quantum algorithms,  extending the range of 
techniques available for Gibbs expectation estimation. By combining 
multilevel Monte Carlo, unbiased estimation, and quantum mean estimation, 
we also provide further insight into the potential of quantum algorithms 
for stochastic processes and expectation estimation problems.

\medskip
\paragraph{\textbf{\textsf{Dependence on the dimension}}}
This work provides a detailed analysis of the dependence of the complexity on the dimension $d$ of the SDE. 
Under the assumption that $\|\widehat{Y}^c_0\|$, $\|\widehat{Y}^f_0\|$, and $\|\nabla f(0)\|$ are $\mathcal{O}(1)$, 
the resulting complexity is $\mathcal{O}(d)$ in the one-sided Lipschitz setting. 
In contrast, in the dissipative setting, the change-of-measure argument introduces exponential terms, 
yielding an overall complexity of $\mathcal{O}(d^{3/2})$.

Although this assumption is relatively strong, we emphasize that all our results remain valid when these quantities scale as $\mathcal{O}(d)$ or at other rates;  
in such cases, the dependence on $d$ increases accordingly. 
Further details can be found in Appendix~\hyperref[Appendix A: Strong Error Analysis]{A} and
how to further reduce the dimension dependence remains an interesting direction for future research.

On the other hand, it is worth noting that all arguments  in Appendix~\hyperref[Appendix A: Strong Error Analysis]{A}
apply equally to classical multilevel methods for SDEs. 
Compared with classical methods, our approach reduces the dependence on $d$ in two main aspects. 
First, the cost of simulating the SDE decreases from $\mathcal{O}(d)$ per step to $\mathcal{O}(1)$ oracle queries, 
since each classical update involves all $d$ coordinates. 
Second, if the multilevel variance scales as $\mathrm{Var}=\mathcal{O}(C(d))$, 
then its contribution to the overall complexity is $\mathcal{O}(C(d))$ in the classical setting, 
whereas in the QA-MLMC framework it becomes $\mathcal{O}(\sqrt{C(d)})$, 
yielding a quadratic improvement in the variance-dependent term.

Finally, we emphasize that it is important to distinguish between the two notions of dimension. 
In Theorem~\ref{thm:QMLMC2}, the dimension $r$ refers to that of the target random variable, 
which is different from the dimension $d$ of the underlying SDE considered here. 
In the Gibbs expectation problem studied in this work, 
the observable $\varphi$ is assumed to be scalar-valued, i.e., $r=1$.

\medskip
\paragraph{\textbf{\textsf{Limitations and future directions}}}
However, several limitations remain.
First, our analysis relies on idealized quantum oracles and does not explicitly account for 
implementation costs on quantum hardware. 
Second, the current framework still relies on certain structural assumptions, 
such as dissipativity and, in some cases, $L$-Hessian smoothness, 
which may not be satisfied in more general or highly irregular models.
Third, while our framework can handle certain heavy-tailed distributions via suitable transformations, 
the current construction does not cover all heavy-tailed settings.

These limitations suggest several directions for future work. 
It would be of interest to develop more implementation-aware quantum algorithms, 
further relax the structural assumptions, and extend the framework to a broader class of heavy-tailed distributions. 
Moreover, it would also be of interest to generalize the framework to more complex stochastic systems, 
such as solving stochastic partial differential equations (SPDEs) and stochastic optimization problems.

\section*{Acknowledgments}
JPL acknowledges support from Quantum Science and Technology--National Science and Technology Major Project under Grant No.~2024ZD0300500, Excellent Young Scientists Fund Program, start-up funding from Tsinghua University and Beijing Institute of Mathematical Sciences and Applications.

\bibliographystyle{abbrv}
\bibliography{refs} \label{bib}

\section*{Appendix A: Strong Error Analysis}
\label{Appendix A: Strong Error Analysis}
\addcontentsline{toc}{section}{Appendix A: Strong Error Analysis}
In the following, we assume that 
$f$ is $L$-smooth, $L$-Hessian smooth, 
dissipative with
$$
\langle x,\nabla f(x)\rangle \geq \tilde\alpha \|x\|^2-\tilde\beta,
$$
and satisfies a weaker one-sided Lipschitz condition 
with constant $\lambda$. 
In the following, we assume that $\|\widehat{Y}^c_0\|$, $\|\widehat{Y}^f_0\|$, and $\|\nabla f(0)\|$ are $\mathcal{O}(1)$.
Under these assumptions, we prove several properties of the modified SDE 
introduced in Subsection~\ref{subsec:Quantum-Accelerated Multilevel Monte Carlo with Change of Measure}.

The proof below is largely based on \cite{Fang2018MultilevelMC}, 
while we further supplement it with a detailed complexity analysis regarding
the parameters of interest, such as $L$ and $d$.

For simplicity, write
$$a(h)\lesssim b(h)$$
if there exists $h_0>0$ such that $a(h)\leq b(h)$ for all $0<h\leq h_0$, 
where $h_0$ may depend on deterministic constants such as 
$S,L,\tilde\alpha,\tilde\beta,\|\nabla f(0)\|$, but not on the Brownian samples.

On the other hand, 
the proofs of the following lemmas will involve many constants 
that are independent of the step size $h$ and the Brownian motion. 
To distinguish between different types of constants,
 we adopt the following notation.

We use $C_i$ to denote constants that are independent of all model parameters, 
including $L, S, \lambda, \tilde{\alpha}, m$, 
as well as the dimension $d$ and the evolution time $T$. 
In contrast, we use $C_{(i)}$ and $h_{(i)}$ to denote quantities that may depend on these parameters. 
We will also explicitly specify how $C_{(i)}$ and $h_{(i)}$ depend on the above quantities.
\begin{lemma}\label{lem:bounds of p moment}
For $p\geq 1$, 
there exist constants $C_1, C_2>0$  such that
\begin{align*}
\sup_{0\leq n\leq 2N}
\mathbb E\left[
\left(\|\widehat Y^f_{n}\|^2+\|\widehat Y^c_{n}\|^2\right)^{p/2}\right]
\leq C_2^p \tilde{\alpha}^{-p/2} d^{p/2}p^{p/2}.
\end{align*}
provided that $h$ is sufficiently small and satisfies
$$
h \leq C_1 \min\left\{S^{-1},\tilde{\alpha} L^{-2},\tilde{\alpha} \|\nabla f(0)\|^{-2},\tilde{\alpha}^{-1}\right\}.
$$
\end{lemma}

\begin{proof}
We prove the result for $p\geq 4$; the case $1\leq p<4$ then follows from H\"older's inequality.

When $t=0$, the two numerical paths start from the same initial value
$\widehat Y^f_{0}=\widehat Y^c_{0}=X_0.$

For odd time points,
\begin{align*}
  \widehat Y^c_{{2n+1}}=&
\widehat Y^c_{{2n}}+S\left(\widehat Y^f_{{2n}}-\widehat Y^c_{{2n}}\right)h
-\nabla f(\widehat Y^c_{{2n}})h+\sqrt2\,\Delta W_{2n},\\
\widehat Y^f_{{2n+1}}
=&\widehat Y^f_{{2n}}+S\left(\widehat Y^c_{{2n}}-\widehat Y^f_{{2n}}\right)h
-\nabla f(\widehat Y^f_{{2n}})h+\sqrt2\,\Delta W_{2n}.
\end{align*}
Squaring the coarse update,
\begin{align*}
\|\widehat Y^c_{{2n+1}}\|^2
=\left\|Sh\,\widehat Y^f_{{2n}}+(1-Sh)\left(\widehat Y^c_{{2n}}
+\frac{-\nabla f(\widehat Y^c_{{2n}})h+\sqrt2\,\Delta W_{2n}}{1-Sh}\right)
\right\|^2.
\end{align*}
Since $\|\cdot\|^2$ is convex and $0\leq Sh< 1$ for sufficiently small $h=\mathcal{O}(S^{-1})$,
\begin{align*}
\|\widehat Y^c_{{2n+1}}\|^2
\leq & Sh\|\widehat Y^f_{{2n}}\|^2+(1-Sh)\left\|\widehat Y^c_{{2n}}
+\frac{-\nabla f(\widehat Y^c_{{2n}})h+\sqrt2\,\Delta W_{2n}}{1-Sh}\right\|^2\\
=&Sh\|\widehat Y^f_{{2n}}\|^2
+(1-Sh)\|\widehat Y^c_{{2n}}\|^2
+\frac{1}{1-Sh}\|-\nabla f(\widehat Y^c_{{2n}})h+\sqrt2\Delta W_{2n}\|^2
+2\langle Y^c_{{2n}},-\nabla f(\widehat Y^c_{{2n}})h+\sqrt2\Delta W_{2n}\rangle\\
\lesssim&
Sh\|\widehat Y^f_{{2n}}\|^2
+(1-Sh)\|\widehat Y^c_{{2n}}\|^2
+4\|\nabla f(\widehat Y^c_{{2n}})\|^2h^2
+8\|\Delta W_{2n}\|^2
-2\langle \widehat Y^c_{{2n}},\nabla f(\widehat Y^c_{{2n}})\rangle h
+2\sqrt2\langle \widehat Y^c_{{2n}},\Delta W_{2n}\rangle.
\end{align*}
By $L$-smoothness, we have $\|\nabla f(x)\|\leq L\|x\|+\|\nabla f(0)\|$.
Hence for any fixed $\eta>0$,
\begin{align*}
\|\nabla f(x)\|^2 h^2\leq2L^2 h^2\|x\|^2+2\|\nabla f(0)\|^2h^2
\lesssim\eta h\|x\|^2+\eta h
\end{align*}
for $h=\mathcal{O}(\min\{\eta L^{-2},\eta \|\nabla f(0)\|^{-2}\})$.
On the other hand, by dissipativity, we have
\begin{align*}
    4\|\nabla f(x)\|^2 h^2-2\langle x,\nabla f(x)\rangle h
    \leq -2(\tilde{\alpha}-2\eta)\|x\|^2h+(2\tilde{\beta}+4\eta)h.
\end{align*}
Combining the above estimates, fix constants
$\alpha\in(0,\tilde\alpha)$ and $\beta\in(\tilde\beta,\infty)$
with $\alpha,\eta=\mathcal{O}(\tilde\alpha)$,
we have
\begin{align*}
\|\widehat Y^c_{{2n+1}}\|^2
\lesssim&
Sh\|\widehat Y^f_{{2n}}\|^2
+(1-Sh-2\alpha h)\|\widehat Y^c_{{2n}}\|^2 +8\|\Delta W_{2n}\|^2
+2\beta h
+2\sqrt2\,\langle \widehat Y^c_{{2n}},\Delta W_{2n}\rangle
\end{align*}
when $h=\mathcal{O}\left(\min\left\{S^{-1},\tilde{\alpha} L^{-2},\tilde{\alpha} \|\nabla f(0)\|^{-2}\right\}\right)$.
Similarly,
\begin{align*}
\|\widehat Y^f_{{2n+1}}\|^2\lesssim&
Sh\|\widehat Y^c_{{2n}}\|^2+(1-Sh-2\alpha h)\|\widehat Y^f_{{2n}}\|^2 +8\|\Delta W_{2n}\|^2
+2\beta h+2\sqrt2\,\langle \widehat Y^f_{{2n}},\Delta W_{2n}\rangle.
\end{align*}

For even time points,
\begin{align*}
\widehat Y^c_{{2n+2}}=&\widehat Y^c_{{2n}}
+2S\left(\widehat Y^f_{{2n}}-\widehat Y^c_{{2n}}\right)h
-2\nabla f(\widehat Y^c_{{2n}})h+\sqrt2(\Delta W_{2n}+\Delta W_{2n+1}),\\
\widehat Y^f_{{2n+2}}=&\widehat Y^f_{{2n+1}}
+S\left(\widehat Y^c_{{2n+1}}-\widehat Y^f_{{2n+1}}\right)h
-\nabla f(\widehat Y^f_{{2n+1}})h+\sqrt2\,\Delta W_{2n+1}.
\end{align*}
Similarly, we have
\begin{align*}
\|\widehat Y^c_{{2n+2}}\|^2\lesssim&2Sh\|\widehat Y^f_{{2n}}\|^2
+(1-2Sh-4\alpha h)\|\widehat Y^c_{{2n}}\|^2 +8\|\Delta W_{2n}+\Delta W_{2n+1}\|^2
+4\beta h +2\sqrt2\,\langle \widehat Y^c_{{2n}},\Delta W_{2n}+\Delta W_{2n+1}\rangle,\\
\|\widehat Y^f_{{2n+2}}\|^2\lesssim&Sh\|\widehat Y^c_{{2n+1}}\|^2
+(1-Sh-2\alpha h)\|\widehat Y^f_{{2n+1}}\|^2 +8\|\Delta W_{2n+1}\|^2
+2\beta h+2\sqrt2\,\langle \widehat Y^f_{{2n+1}},\Delta W_{2n+1}\rangle.
\end{align*}
It implies that
\begin{align*}
\|\widehat Y^f_{{2n+2}}\|^2
\lesssim{}\,&
(1-2Sh-4\alpha h +\mathcal{O}(h^2))\|\widehat Y^f_{{2n}}\|^2
+(2Sh+\mathcal{O}(h^2))\|\widehat Y^c_{{2n}}\|^2
+8(1+Sh)\|\Delta W_{2n}\|^2+8\|\Delta W_{2n+1}\|^2\\
&\quad +4\beta h
+2\sqrt2 (1-Sh-2\alpha h)\langle \widehat Y^f_{{2n}},\Delta W_{2n}\rangle
+2\sqrt2 Sh \langle \widehat Y^c_{{2n}},\Delta W_{2n}\rangle
+2\sqrt2 \langle \widehat Y^f_{{2n+1}},\Delta W_{2n+1}\rangle+\mathcal{O}(h^2).
\end{align*}
So for any fixed $\gamma\in(0,\alpha)$ with $\gamma=\mathcal{O}({\alpha})=\mathcal{O}(\tilde{\alpha})$,
\begin{align*}
\|\widehat Y^c_{{2n+2}}\|^2+\|\widehat Y^f_{{2n+2}}\|^2
\lesssim&(1-4\gamma h)\left(\|\widehat Y^c_{{2n}}\|^2+\|\widehat Y^f_{{2n}}\|^2\right) 
+24\left(\|\Delta W_{2n}\|^2+\|\Delta W_{2n+1}\|^2\right)
+8\beta h \\
&\quad+2\sqrt{2}\langle (1+Sh)\widehat Y^c_{{2n}}
+(1-Sh-2\alpha h)\widehat Y^f_{{2n}},\Delta W_{2n}\rangle
+2\sqrt{2}\langle \widehat Y^c_{{2n}}+\widehat Y^f_{{2n+1}},\Delta W_{2n+1}\rangle.
\end{align*}

Let $e^{-4\gamma h}\phi_{{2n}} = (1+Sh)\widehat Y^c_{{2n}} + (1-Sh-2\alpha h)\widehat Y^f_{{2n}}$, 
$e^{-2\gamma h}\phi_{{2n+1}} = \widehat Y^c_{{2n}}+\widehat Y^f_{{2n+1}}.$
Since $1-4\gamma h\leq e^{-4\gamma h}$ and $e^{4\gamma h}\lesssim \sqrt{2}$ for small $h=\mathcal{O}(\gamma^{-1})=\mathcal{O}(\tilde{\alpha}^{-1})$, 
multiplying both sides by $e^{4 (n+1)\gamma h}$ gives
\begin{align*}
e^{4 (n+1)\gamma h}\left(\|\widehat Y^f_{{2n+2}}\|^2+\|\widehat Y^c_{{2n+2}}\|^2\right)
\lesssim&e^{4n\gamma h}\left(\|\widehat Y^f_{{2n}}\|^2+\|\widehat Y^c_{{2n}}\|^2\right) 
+48e^{4n\gamma h}\left(\|\Delta W_{2n}\|^2+\|\Delta W_{2n+1}\|^2\right) \\
&\quad+16\beta e^{4n\gamma h}h +4e^{4n\gamma h}\langle \phi_{{2n}},\Delta W_{2n}\rangle
+4 e^{2(2n+1)\gamma h}\langle \phi_{{2n+1}},\Delta W_{2n+1}\rangle.
\end{align*}
Summing over even time points yields
\begin{align}\label{eq:suming Y^f+Y^c}
e^{4n\gamma h}
\left(\|\widehat Y^f_{{2n}}\|^2+\|\widehat Y^c_{{2n}}\|^2\right)
\lesssim&
\|\widehat Y^f_{0}\|^2+\|\widehat Y^c_{0}\|^2 
+48\sum_{k=0}^{2n-1}e^{2k\gamma h}\|\Delta W_k\|^2
+16\beta\sum_{k=0}^{n-1}e^{4k\gamma h}h 
+4\sum_{k=0}^{2n-1}e^{2k\gamma h}\langle \phi_{k},\Delta W_k\rangle.
\end{align}

For odd time points, we have
\begin{align*}
    \|\widehat Y^f_{{2n+1}}\|^2+\|\widehat Y^c_{{2n+1}}\|^2
\lesssim&
(1-2\alpha h)\left(\|\widehat Y^f_{{2n}}\|^2+\|\widehat Y^c_{{2n}}\|^2\right)
+16\|\Delta W_{2n}\|^2
+2\sqrt2\left(\langle \widehat Y^c_{{2n}},\Delta W_{2n}\rangle
+\langle \widehat Y^f_{{2n}},\Delta W_{2n}\rangle\right)+4\beta h.
\end{align*}
Using Young's inequality, there exist constants
$\alpha_1>1$ and $\beta_1>\max(1,\alpha_1\beta)$
with $\alpha_1,\beta_1=\mathcal{O}(1)$,
such that
\begin{align*}
\|\widehat Y^f_{{2n+1}}\|^2+\|\widehat Y^c_{{2n+1}}\|^2
\lesssim&(1-2\gamma h)
\left(\alpha_1\left(\|\widehat Y^f_{{2n}}\|^2+\|\widehat Y^c_{{2n}}\|^2
+24\|\Delta W_{2n}\|^2\right)+4\beta_1 h\right).
\end{align*}
Multiplying by $e^{2(2n+1)\gamma h}$ and using~\eqref{eq:suming Y^f+Y^c}, we obtain
\begin{align}\label{eq:odd upper bound of Yf+Yc}
e^{2(2n+1)\gamma h}
\left(\|\widehat Y^f_{{2n+1}}\|^2+\|\widehat Y^c_{{2n+1}}\|^2\right)
\lesssim &\alpha_1\left(\|\widehat Y^f_{0}\|^2+\|\widehat Y^c_{0}\|^2\right) 
+48\alpha_1\sum_{k=0}^{2n}e^{2k\gamma h}\|\Delta W_k\|^2 \nonumber\\ 
&\quad+16\beta_1\sum_{k=0}^{n}e^{4k\gamma h}h 
+4\alpha_1\sum_{k=0}^{2n-1}e^{2k\gamma h}\langle \phi_{k},\Delta W_k\rangle.
\end{align}

Fixing $n$, we combine~\eqref{eq:suming Y^f+Y^c} and~\eqref{eq:odd upper bound of Yf+Yc},
raise both sides to the power $p/2$, and take expectations to obtain
\begin{align*}
\mathbb E\left[\sup_{0\leq n\leq 2N} e^{ p n\gamma h}
\left(\|\widehat Y^f_{n}\|^2+\|\widehat Y^c_{n}\|^2\right)^{p/2}\right]
\lesssim4^{p/2-1}(48\alpha_1\beta_1)^{p/2}(I_1+I_2+I_3+I_4),
\end{align*}
where
\begin{align*}
I_1=&\mathbb E\left[\left(\|\widehat Y^f_{0}\|^2+\|\widehat Y^c_{0}\|^2\right)^{p/2}
\right],\\
%=2^{p/2}\|x_0\|^p,\\
I_2=&\left(\sum_{k=0}^{n-1}e^{2k\gamma h}h\right)^{p/2},\\
I_3=&\mathbb E\left[\left|\sum_{k=0}^{n-1}e^{2k\gamma h}\|\Delta W_k\|^2
\right|^{p/2}\right],\\
I_4=&\mathbb E\left[\sup_{0\leq n\leq N}\left|
\sum_{k=0}^{n-1}e^{2k\gamma h}\langle \phi_{k},\Delta W_k\rangle\right|^{p/2}
\right].
\end{align*}

We now estimate these terms one by one.

For $I_2$,
\begin{align*}
I_2\leq\left(\int_0^{nh} e^{2\gamma t}\d t\right)^{p/2}
\leq(2\gamma)^{-p/2}e^{pn\gamma h }.
\end{align*}

For $I_3$, we use the weighted Jensen inequality:
for nonnegative $a_k$ and any $r>1$,
\begin{align*}
\left|\sum_k a_k b_k\right|^r\leq\left(\sum_k a_k\right)^{r-1}
\sum_k a_k |b_k|^r.
\end{align*}
Applying this with
$a_k=e^{2k\gamma h}h$,$b_k=\frac{\|\Delta W_k\|^2}{h}$ and $r=\frac p2,$
we obtain
\begin{align*}
I_3=\mathbb E\left[\left|\sum_{k=0}^{n-1}e^{2k\gamma h}h\,\frac{\|\Delta W_k\|^2}{h}
\right|^{p/2}\right]
\leq\left(\sum_{k=0}^{n-1}e^{2k\gamma h}h\right)^{p/2-1}\sum_{k=0}^{n-1}e^{2k\gamma  h}h\,
\mathbb E\left[\frac{\|\Delta W_k\|^p}{h^{p/2}}\right].
\end{align*}
Since
\begin{align*}
\mathbb E\left[\frac{\|\Delta W_k\|^p}{h^{p/2}}\right]
\leq d^{p/2} p!!\leq d^{p/2}p^{p/2},
\end{align*}
we have
\begin{align*}
I_3\leq d^{p/2}p^{p/2}
\left(\sum_{k=0}^{n-1}e^{2k\gamma  h}h\right)^{p/2}
\leq d^{p/2}p^{p/2}(2\gamma)^{-p/2}e^{pn\gamma h}.
\end{align*}

For $I_4$, write the discrete martingale as an It\^o integral:
\begin{align*}
\sum_{k=0}^{n-1}e^{2k\gamma  h}\langle \phi_{k},\Delta W_k\rangle
=\int_0^{{n}}e^{2\gamma \lfloor s/h\rfloor h}
\langle \phi_{\lfloor s/h\rfloor h}\d W_s\rangle.
\end{align*}
Then by the Burkholder--Davis--Gundy inequality in~\cite{BARLOW1982198}, there exists a universal constant $C_{\mathrm{BDG}}$ such that
\begin{align*}
I_4&\leq\mathbb E\left[\sup_{0\leq t\leq nh}
\left|\int_0^t e^{2\gamma \lfloor s/h\rfloor h}
\langle \phi_{\lfloor s/h\rfloor h},\d W_s\rangle\right|^{p/2}\right] \\
&\leq(C_{\mathrm{BDG}}p)^{p/4}\,\mathbb E\left[\left(
\sum_{k=0}^{n-1}e^{4k\gamma  h}\|\phi_{k}\|^2 h\right)^{p/4}\right].
\end{align*}
Recall that
$e^{-4\gamma h}\phi_{{2k}} = (1+Sh)\widehat Y^c_{{2k}} + (1-Sh-2\alpha h)\widehat Y^f_{{2k}}$, 
$e^{-2\gamma h}\phi_{{2k+1}} = \widehat Y^c_{{2k}}+\widehat Y^f_{{2k+1}}.$
Using Young's inequality, for $h=\mathcal{O}(\min\{S^{-1},\tilde\alpha^{-1}\})$,
\begin{align*}
\|\phi_{{2k}}\|^2\lesssim 8\left(\|\widehat Y^c_{{2k}}\|^2+\|\widehat Y^f_{{2k}}\|^2\right),
\quad \|\phi_{{2k+1}}\|^2\lesssim
4\left(\|\widehat Y^c_{{2k}}\|^2+\|\widehat Y^f_{{2k+1}}\|^2\right).
\end{align*}
Therefore,
\begin{align*}
I_4\lesssim&(12C_{\mathrm{BDG}}p)^{p/4}
\left(\sum_{k=0}^{n-1}e^{2k\gamma  h}h\right)^{p/4-1} 
\mathbb E\left[\sum_{k=0}^{n-1}e^{2k\gamma  h}h\,\sup_{0\leq n\leq 2N}
e^{pn\gamma h/2}
\left(\|\widehat Y^c_{k}\|^2+\|\widehat Y^f_{k}\|^2\right)^{p/4}\right].
\end{align*}
For $\zeta>0$, applying Young's inequality,
\begin{align*}
I_4\leq&\frac1{4\zeta}
\mathbb E\left[\sup_{0\leq n\leq 2N}e^{pn\gamma  h}
\left(\|\widehat Y^f_{n}\|^2+\|\widehat Y^c_{n}\|^2\right)^{p/2}\right] 
+\zeta\left(\frac{6C_{\mathrm{BDG}}}{\gamma}\right)^{p/2}p^{p/2}e^{pn\gamma h}.
\end{align*}
Note carefully that this step does not create any new explicit power of $d$
and all explicit $d$-dependence enters through $I_3$ above.

Finally, combining the estimates for $I_1$--$I_4$,
\begin{align*}
\mathbb E\left[\sup_{0\leq n\leq 2N}e^{ p n\gamma h}
\left(\|\widehat Y^f_{n}\|^2+\|\widehat Y^c_{n}\|^2\right)^{p/2}\right]
\lesssim4^{p/2-1}(48\alpha_1\beta_1)^{p/2}
\Bigg\{\mathbb E\left[\left(\|\widehat Y^f_{0}\|^2+\|\widehat Y^c_{0}\|^2\right)^{p/2}\right]
+(2\gamma)^{-p/2}e^{pn\gamma h}\\
+d^{p/2}p^{p/2}(2\gamma)^{-p/2}e^{pn\gamma h}+\frac1{4\zeta}
\mathbb E\left[\sup_{0\leq n\leq 2N}e^{p\gamma  nh}
\left(\|\widehat Y^f_{n}\|^2+\|\widehat Y^c_{n}\|^2\right)^{p/2}\right] 
+\zeta\left(\frac{6C_{\mathrm{BDG}}}{\gamma}\right)^{p/2}p^{p/2}e^{pn\gamma h}
\Bigg\}.
\end{align*}
Choosing $\zeta=4^{p/2-1}(48\alpha_1\beta_1)^{p/2}$,
we obtain
\begin{align*}
\mathbb E\left[\sup_{0\leq n\leq 2N}
e^{p \gamma nh}
\left(\|\widehat Y^f_{n}\|^2+\|\widehat Y^c_{n}\|^2\right)^{p/2}\right]
\leq C_2^p \tilde{\alpha}^{-p/2} d^{p/2}p^{p/2}e^{pn\gamma h},
\end{align*}
provided that $h$ is sufficiently small and satisfies
$$
h \leq C_1 \min\left\{S^{-1},\tilde{\alpha} L^{-2},\tilde{\alpha} \|\nabla f(0)\|^{-2},\tilde{\alpha}^{-1}\right\}
$$
for some constants $C_1, C_2>0$.
Hence,
\begin{align*}
\sup_{0\leq n\leq 2N}\mathbb E\left[
\left(\|\widehat Y^f_{n}\|^2+\|\widehat Y^c_{n}\|^2\right)^{p/2}\right]
\leq C_2^p \tilde{\alpha}^{-p/2} d^{p/2}p^{p/2}.
\end{align*}
\end{proof}

\begin{corollary}\label{cor: bound of nabla f}
   For any $p\geq 1$ and $h \leq C_1 \min\left\{S^{-1},\tilde{\alpha} L^{-2},\tilde{\alpha} \|\nabla f(0)\|^{-2},\tilde{\alpha}^{-1}\right\}$,  
    \begin{align*}
  &\sup_{0\leq n\leq 2N}\mathbb E\left[
\|\nabla f(\widehat Y^f_{n})\|^p\right]
\leq 2^{p-1}L^p C_2^p d^{p/2}p^{p/2}+2^{p-1}\|\nabla f(0)\|^p,\\
 & \sup_{0\leq n\leq 2N}\mathbb E\left[
\|\nabla f(\widehat Y^c_{n})\|^p\right]
\leq 2^{p-1}L^p C_2^p d^{p/2}p^{p/2}+2^{p-1}\|\nabla f(0)\|^p.
    \end{align*}
\end{corollary}
\begin{proof}
    This follows directly from Lemma~\ref{lem:bounds of p moment} and the smoothness of $f$.
   \begin{align*}
\sup_{0\leq n\leq 2N}\mathbb E\left[
\|\nabla f(\widehat Y^f_{n})\|^p\right]
\leq&\sup_{0\leq n\leq 2N}\mathbb E\left[
\left(L\|\widehat Y^f_{n}\|+\|\nabla f(0)\|\right)^{p}\right]\\
\leq & 2^{p-1}\sup_{0\leq n\leq 2N}\mathbb E\left[
\left(L^p\|\widehat Y^f_{n}\|^p+\|\nabla f(0)\|^p\right)\right]\\
\leq& 2^{p-1}L^p C_2^p d^{p/2}p^{p/2}+2^{p-1}\|\nabla f(0)\|^p.
\end{align*} 
The proof for the coarse path is analogous.
\end{proof}

\begin{lemma}\label{lem:p-bounds of distance of two paths}
For $p\geq 1$, there exist constants $C_3, C_4, C_{(1)}, C_{(2)} > 0$,
with
$$
C_{(1)}=\mathcal{O}\left(\frac{L\sqrt{d} }{2S-\lambda}\right),\quad C_{(2)} = \mathcal{O}\left(\frac{L^2\sqrt{d}+Ld}{2S-\lambda}\right),
$$
and for all
 $$0<h\leq h_{(1)}:=C_3\min\left\{S^{-1},\tilde{\alpha} L^{-2},\tilde{\alpha} \|\nabla f(0)\|^{-2},\tilde{\alpha}^{-1}
(2S-\lambda)^{-1},\tilde{\alpha} S^{-2}, \|\nabla f(0)\|^{-2},(1-\lambda/S)^2
\right\},$$
the following holds.

\medskip

\noindent
\textbf{(1) when $ \widehat{Y}_{0}^f= \widehat{Y}_{0}^c$:} 
\begin{align*}
    \sup_{0\leq n\leq N}\E\left[
\|\widehat{Y}_{{2n}}^f- \widehat{Y}_{{2n}}^c\|^p \right]
\leq C_{(1)}^p p^{p/2}h^{p/2},
\quad 
\sup_{0\leq n\leq N} \E\left[ 
\|\widehat{Y}_{{2n}}^f- \widehat{Y}_{{2n}}^c\|^p \right]
\leq
C_{(2)}^p p^{p}h^{p}.
\end{align*}

\medskip

\noindent
\textbf{(2) when $ \widehat{Y}_{0}^f\neq \widehat{Y}_{0}^c$:} 
\begin{align*}
 \E\left[\|\widehat{Y}_{{2N}}^f- \widehat{Y}_{{2N}}^c\|^p \right]
 \leq
&C_4^p e^{-p \gamma T/2}\mathbb{E}\left[\|\widehat{Y}_{0}^f- \widehat{Y}_{0}^c\|^p\right]
+C_{(2)}^p p^{p}h^{p},
\end{align*}
where $\gamma\in(0,2S-\lambda)$ and $\gamma=\Theta(2S-\lambda)$.
\end{lemma}
\begin{proof}

The proof is given for $p\geq 4$; the result for 
$1\leq p <  4$ follows from H{\"o}lder's inequality.

The different updates on odd and even time points give
\begin{align*}
\widehat{Y}_{{2n+1}}^f-\widehat{Y}_{{2n+1}}^c 
=& (1-2Sh)(\widehat{Y}_{{2n}}^f- \widehat{Y}_{{2n}}^c)- 
(\nabla f(\widehat{Y}_{{2n}}^f) - \nabla f(\widehat{Y}_{{2n}}^c))h,\\
\widehat{Y}_{{2n+2}}^f- \widehat{Y}_{{2n+2}}^c
=& (1-Sh)(\widehat{Y}_{{2n+1}}^f -\widehat{Y}_{{2n+1}}^c)
-Sh(\widehat{Y}_{{2n}}^f-\widehat{Y}_{{2n}}^c)-(\nabla f(\widehat{Y}_{{2n+1}}^f)
- \nabla f(\widehat{Y}_{{2n}}^c))h,
\end{align*}
and then
\begin{align*}
    \widehat{Y}_{{2n+2}}^f -  \widehat{Y}_{{2n+2}}^c
 =&  (1-4Sh+2S^2h^2)(\widehat{Y}_{{2n}}^f- \widehat{Y}_{{2n}}^c) 
 - (2-Sh)(\nabla f(\widehat{Y}_{{2n}}^f) - \nabla f(\widehat{Y}_{{2n}}^c))h 
 - (\nabla f(\widehat{Y}_{{2n+1}}^f)- \nabla f(\widehat{Y}_{{2n}}^f))h.
\end{align*}
After squaring both sides and applying the $L$-smoothness and the weaker one-sided Lipschitz condition,
provided that $0\leq Sh<1$, we have
\begin{align*}
\|\widehat{Y}_{{2n+2}}^f- \widehat{Y}_{{2n+2}}^c\|^2  =&  (1-4Sh+2S^2h^2)^2 \|\widehat{Y}_{{2n}}^f- \widehat{Y}_{{2n}}^c\|^2 \\
& +(2-Sh)^2h^2\|\nabla f(\widehat{Y}_{{2n}}^f) - \nabla f(\widehat{Y}_{{2n}}^c)\|^2 +h^2\|\nabla f(\widehat{Y}_{{2n+1}}^f)- \nabla f(\widehat{Y}_{{2n}}^f)\|^2 \\
& - 2(1-4Sh+2S^2h^2)(2-Sh)h\langle \widehat{Y}_{{2n}}^f- \widehat{Y}_{{2n}}^c, \nabla f(\widehat{Y}_{{2n}}^f) - \nabla f(\widehat{Y}_{{2n}}^c) \rangle \\
& -2(1-4Sh+2S^2h^2)h \langle \widehat{Y}_{{2n}}^f- \widehat{Y}_{{2n}}^c,\nabla f(\widehat{Y}_{{2n+1}}^f)- \nabla f(\widehat{Y}_{{2n}}^f) \rangle 
\\
&+2(2-Sh)h^2 \langle \nabla f(\widehat{Y}_{{2n}}^f) - \nabla f(\widehat{Y}_{{2n}}^c), \nabla f(\widehat{Y}_{{2n+1}}^f)- \nabla f(\widehat{Y}_{{2n}}^f) \rangle \\
 \leq&  (1-4Sh+2S^2h^2)^2 \|\widehat{Y}_{{2n}}^f- \widehat{Y}_{{2n}}^c\|^2 \\
& +4L^2h^2\|\widehat{Y}_{{2n}}^f - \widehat{Y}_{{2n}}^c\|^2 +L^2h^2\|\widehat{Y}_{{2n+1}}^f- \widehat{Y}_{{2n}}^f\|^2 \\
& + 2(1-4Sh+2S^2h^2)(2-Sh) h\lambda\| \widehat{Y}_{{2n}}^f- \widehat{Y}_{{2n}}^c\|^2 \\
& -2(1-4Sh+2S^2h^2) h\langle \widehat{Y}_{{2n}}^f- \widehat{Y}_{{2n}}^c,\nabla f(\widehat{Y}_{{2n+1}}^f)- \nabla f(\widehat{Y}_{{2n}}^f) \rangle 
\\
& +2 L^2h^2\| \widehat{Y}_{{2n}}^f -\widehat{Y}_{{2n}}^c\|^2+2L^2h^2\|\widehat{Y}_{{2n+1}}^f- \widehat{Y}_{{2n}}^f\|^2\\
=&\left[1 - 4(2S-\lambda)h + (20S^2 - 18S\lambda + 6L^2)h^2 + 16S^2(\lambda - S)h^3 + 4S^3(S - \lambda)h^4\right]\|\widehat{Y}_{{2n}}^f- \widehat{Y}_{{2n}}^c\|^2\\
&+3L^2h^2\|\widehat{Y}_{{2n+1}}^f- \widehat{Y}_{{2n}}^f\|^2 
 -2(1-4Sh+2S^2h^2) h\langle \widehat{Y}_{{2n}}^f- \widehat{Y}_{{2n}}^c,\nabla f(\widehat{Y}_{{2n+1}}^f)- \nabla f(\widehat{Y}_{{2n}}^f) \rangle .
\end{align*}
Since $S>\lambda/2$, 
fix $\gamma\in(0,2S-\lambda)$ with $\gamma=\Theta(2S-\lambda)$,  we have
\begin{equation}
    \label{error:2}
    \begin{aligned}
        \|\widehat{Y}_{{2n+2}}^f- \widehat{Y}_{{2n+2}}^c\|^2 
        \lesssim& (1-4\gamma h) \|\widehat{Y}_{{2n}}^f- \widehat{Y}_{{2n}}^c\|^2 +3L^2h^2\|\widehat{Y}_{{2n+1}}^f- \widehat{Y}_{{2n}}^f\|^2 \\
        &-2(1-4Sh+2S^2h^2) h\langle \widehat{Y}_{{2n}}^f- \widehat{Y}_{{2n}}^c,\nabla f(\widehat{Y}_{{2n+1}}^f)- \nabla f(\widehat{Y}_{{2n}}^f) \rangle
    \end{aligned}  
\end{equation}
provided that $h$ is sufficiently small and satisfies
 $h=\mathcal{O}(\min\{S^{-1},(2S-\lambda)L^{-2}\})$.
Following this estimate, we use two different approaches to get different upper bounds.

Before proceeding, we collect two elementary bounds. First,
\begin{align}\label{eq:bound of sum e 2k gamma h}
\sum_{k=0}^{n-1} e^{2\gamma kh} h
\leq \frac{e^{2\gamma nh}-1}{e^{2\gamma h}-1}h
\leq (2\gamma)^{-1} e^{2\gamma nh},
\qquad
\sum_{k=0}^{n-1} e^{4\gamma kh} h
\leq (4\gamma)^{-1} e^{4\gamma nh}.
\end{align}
Second, since $p\geq 4$, the function $x\mapsto |x|^{p/2}$ is convex, and hence for any $w_k>0$,
\begin{align*}
\left|\sum_{k=0}^{n} w_k a_k\right|^{p/2}
\leq
\left(\sum_{k=0}^{n} w_k\right)^{p/2-1}
\sum_{k=0}^{n} w_k |a_k|^{p/2}.
\end{align*}
Applying  with 
$w_k=e^{2\gamma kh}h$ 
yields
\begin{align}\label{eq:jensen}
\left|\sum_{k=0}^{n} e^{2\gamma kh}h\cdot a_k\right|^{p/2}
\leq
\left(\sum_{k=0}^{n} e^{2\gamma kh}h \right)^{p/2-1}
\left(\sum_{k=0}^{n} e^{2\gamma kh}h\cdot |a_k|^{p/2}\right).
\end{align}
Similarly,
\begin{align}\label{eq:jensen_4}
\left|\sum_{k=0}^{n} e^{4\gamma kh}h\cdot a_k\right|^{p/4}
\leq
\left(\sum_{k=0}^{n} e^{4\gamma kh}h \right)^{p/4-1}
\left(\sum_{k=0}^{n} e^{4\gamma kh}h\cdot |a_k|^{p/4}\right).
\end{align}
These estimates will be used repeatedly throughout the remainder of the analysis.

We first establish a weaker result. By Cauchy inequality,
\begin{align*}
   2\left|\langle \widehat{Y}_{{2n}}^f- \widehat{Y}_{{2n}}^c,\nabla f(\widehat{Y}_{{2n+1}}^f)- \nabla f(\widehat{Y}_{{2n}}^f) \rangle\right| 
   \leq 2\gamma\|\widehat{Y}_{{2n}}^f- \widehat{Y}_{{2n}}^c\|^2
   +(2\gamma)^{-1}L^2\|\widehat{Y}_{{2n+1}}^f- \widehat{Y}_{{2n}}^f\|^2,
\end{align*}
we have
\begin{align*}
\|\widehat{Y}_{{2n+2}}^f- \widehat{Y}_{{2n+2}}^c\|^2 
\lesssim& (1-2\gamma h) \|\widehat{Y}_{{2n}}^f- \widehat{Y}_{{2n}}^c\|^2 
+\left((2\gamma)^{-1}+3h\right) L^2h\|\widehat{Y}_{{2n+1}}^f- \widehat{Y}_{{2n}}^f\|^2\\
\lesssim& (1-2\gamma h) \|\widehat{Y}_{{2n}}^f- \widehat{Y}_{{2n}}^c\|^2 
+\gamma^{-1}L^2h\|\widehat{Y}_{{2n+1}}^f- \widehat{Y}_{{2n}}^f\|^2
\end{align*}
for $h=\mathcal{O}(\gamma^{-1})$.
Then  we multiply by $e^{ ({2n+2})\gamma h}$ on both sides and $e^{2\gamma h}\lesssim 2$ gives
\begin{align*}
e^{ ({2n+2})\gamma h}\|\widehat{Y}_{{2n+2}}^f- \widehat{Y}_{{2n+2}}^c\|^2 \lesssim&  
e^{ {2n}\gamma h} \|\widehat{Y}_{{2n}}^f - \widehat{Y}_{{2n}}^c\|^2 +2\gamma^{-1}L^2e^{2n\gamma h }h
\|\widehat{Y}_{{2n+1}}^f- \widehat{Y}_{{2n}}^f\|^2 .
\end{align*}
Summing over multiple timesteps  gives
\begin{align*}
e^{2n\gamma  h}\|\widehat{Y}_{{2n}}^f- \widehat{Y}_{{2n}}^c\|^2 \lesssim &  
\|\widehat{Y}_{0}^f- \widehat{Y}_{0}^c\|^2
+2\gamma^{-1}L^2h \sum_{k=0}^{n-1} e^{ {2k}\gamma  h}\|\widehat{Y}_{{2k+1}}^f- \widehat{Y}_{{2k}}^f\|^2.
\end{align*}
Then, raising both sides to the power $p/2$, fixing $n\in[0,N]$, and taking expectations, we obtain
\begin{align}\label{eq:summing of weak}
&\E\left[ e^{\gamma p {n} h}\|\widehat{Y}_{{2n}}^f- \widehat{Y}_{{2n}}^c\|^p \right] \notag\\
\lesssim&
2^{p/2-1}\mathbb{E}\left[\|\widehat{Y}_{0}^f- \widehat{Y}_{0}^c\|^p\right]
+2^{p/2-1}(2\gamma^{-1})^{p/2} L^p \E\left[ \left( \sum_{k=0}^{n-1} e^{2k\gamma  h}h\|\widehat{Y}_{{2k+1}}^f- \widehat{Y}_{{2k}}^f\|^2  \right)^{p/2}
\right]\notag\\
\lesssim &
2^{p/2-1}\mathbb{E}\left[\|\widehat{Y}_{0}^f- \widehat{Y}_{0}^c\|^p\right]+
(4\gamma^{-1})^{p/2} L^p \left(\sum_{k=0}^{n-1} e^{2k\gamma  h} h\right)^{p/2-1}  
\E\left[  \sum_{k=0}^{n-1} e^{2k\gamma h }h\cdot \|\widehat{Y}_{{2k+1}}^f- \widehat{Y}_{{2k}}^f\|^p  
\right].
\end{align}
Note that $\E\left[ \|\Delta W_{2k}\|^p\right]\leq d^{p/2}h^{p/2}p^{p/2},$ we have
\begin{align}
\label{small increment}
\E\left[\|\widehat{Y}_{{2k+1}}^f- \widehat{Y}_{{2k}}^f\|^p \right] =& \E\left[\|-\nabla f(\widehat{Y}_{{2k}}^f)h+ S(\widehat{Y}_{{2k}}^c-\widehat{Y}_{{2k}}^f)h+\sqrt{2}\Delta W_{2k}\|^{p}\right] \nonumber\\
\leq & 2^{p-1}\E\left[ \|-\nabla f(\widehat{Y}_{{2k}}^f)+ S(\widehat{Y}_{{2k}}^c-\widehat{Y}_{{2k}}^f)\|^p\right]h^p + 2^{3p/2-1}\E\left[ \|\Delta W_{2k}\|^p\right]\nonumber\\
\lesssim&   2^{3p/2} p^{p/2} d^{p/2}h^{p/2}.
\end{align}
There, we use Corollary~\ref{cor: bound of nabla f} and Lemma~\ref{lem:bounds of p moment} to obtain
\begin{align}\label{eq: bound of nabla +S}
    \E\left[ \|-\nabla f(\widehat{Y}_{{2k}}^f)+ S(\widehat{Y}_{{2k}}^c-\widehat{Y}_{{2k}}^f)\|^p\right]
    \leq& 3^{p-1}\E\left[ \|\nabla f(\widehat{Y}_{{2k}}^f)\|^p+ S^p\|\widehat{Y}_{{2k}}^c\|^p+S^p\|\widehat{Y}_{{2k}}^f\|^p\right]\notag\\
\leq& 3^{p-1}\left(2^{p-1}L^p C^p d^{p/2}p^{p/2}+2^{p-1}\|\nabla f(0)\|^p+
S^pC^p \tilde{\alpha}^{-p/2} d^{p/2}p^{p/2}\right).
\end{align}
Therefore, \eqref{small increment} holds provided that $h$ is sufficiently small, namely
$h=\mathcal{O}\big(\min\{L^{-2},\,\tilde{\alpha}S^{-2},\,\|\nabla f(0)\|^{-2}\}\big),$
and that the conditions in Lemma~\ref{lem:bounds of p moment} are satisfied.

Combining \eqref{eq:summing of weak} with \eqref{eq:bound of sum e 2k gamma h} implies that
\begin{align}\label{eq:weaker bound of distance}
&\E\left[ e^{\gamma p {n} h}
\|\widehat{Y}_{{2n}}^f- \widehat{Y}_{{2n}}^c\|^p \right]\notag\\ 
\lesssim &2^{p/2-1}\mathbb{E}\left[\|\widehat{Y}_{0}^f- \widehat{Y}_{0}^c\|^p\right]+
(4\gamma^{-1})^{p/2}L^p\cdot (2\gamma)^{-p/2}e^{p\gamma nh }\cdot  2^{3p/2} p^{p/2} d^{p/2}h^{p/2}\notag\\
\lesssim& 
2^{p/2-1}\mathbb{E}\left[\|\widehat{Y}_{0}^f- \widehat{Y}_{0}^c\|^p\right]
+2^{2p} L^p\gamma^{-p}    p^{p/2}d^{p/2}h^{p/2} e^{\gamma p nh}.
\end{align}
This holds for all $0 \leq n \leq N$. 
In particular, when $\widehat{Y}_{0}^f= \widehat{Y}_{0}^c$, we have
\begin{align*}
   \sup_{0\leq n\leq N} \E\left[
\|\widehat{Y}_{{2n}}^f- \widehat{Y}_{{2n}}^c\|^p \right]\leq 2^{2p} L^p\gamma^{-p}    p^{p/2}d^{p/2}h^{p/2}.
\end{align*}

On the other hand, to derive a stronger estimate, we directly multiply by $e^{2(2n+2)\gamma h}$ on both sides of \eqref{error:2} and $e^{4\gamma h}\lesssim 2$ gives
\begin{align*}
e^{2(2n+2)\gamma h}\|\widehat{Y}_{{2n+2}}^f- \widehat{Y}_{{2n+2}}^c\|^2 
\lesssim&  e^{4n\gamma h} \|\widehat{Y}_{{2n}}^f- \widehat{Y}_{{2n}}^c\|^2 
+6e^{4n\gamma h}L^2h^2\|\widehat{Y}_{{2n+1}}^f- \widehat{Y}_{{2n}}^f\|^2\\
&\quad-2(1-4Sh+2S^2h^2)\,e^{2(2n+2)\gamma h} \langle \widehat{Y}_{{2n}}^f- \widehat{Y}_{{2n}}^c,\nabla f(\widehat{Y}_{{2n+1}}^f)- \nabla f(\widehat{Y}_{{2n}}^f) \rangle h.
\end{align*}
Summing over multiple timesteps  gives
\begin{align*}
e^{4n\gamma h}\|\widehat{Y}_{{2n}}^f- \widehat{Y}_{{2n}}^c\|^2
 \lesssim &\|\widehat{Y}_{0}^f- \widehat{Y}_{0}^c\|^2
 +  6h^2\sum_{k=0}^{n-1} e^{4k\gamma h}L^2\|\widehat{Y}_{{2k+1}}^f- \widehat{Y}_{{2k}}^f\|^2 \\
& -2(1-4Sh+2S^2h^2) \sum_{k=0}^{n-1} e^{2 ({2k+2})\gamma h}\langle \widehat{Y}_{{2k}}^f- \widehat{Y}_{{2k}}^c,\nabla f(\widehat{Y}_{{2k+1}}^f)- \nabla f(\widehat{Y}_{{2k}}^f) \rangle h.
\end{align*}
Then raising both sides to the power $p/2,$ fixing $n\in[0,N],$ taking expectation and by Jensen's inequality, we obtain
\begin{align*}
\E\left[ e^{2\gamma pn h }
\|\widehat{Y}_{{2n}}^f- \widehat{Y}_{{2n}}^c\|^p \right] \lesssim& 
 3^{p/2-1}6^{p/2}\left(\mathbb{E}\left[\|\widehat{Y}_{0}^f- \widehat{Y}_{0}^c\|^p\right]
 +I_1+I_2\right), 
\end{align*}
where
\begin{align*}
I_1 =& \E\left[\ \left|\sum_{k=0}^{n-1} e^{4k\gamma h}L^2\|\widehat{Y}_{{2k+1}}^f- \widehat{Y}_{{2k}}^f\|^2 h^2\right|^{p/2}\ \right],
\\
I_2=& \E\left[\left|  \sum_{k=0}^{n-1} e^{4k\gamma h}\langle \widehat{Y}_{{2k}}^f- \widehat{Y}_{{2k}}^c,\nabla f(\widehat{Y}_{{2k+1}}^f)- \nabla f(\widehat{Y}_{{2k}}^f) \rangle h \right|^{p/2}\ \right].
\end{align*}

For $I_1,$ Jensen's inequality~\eqref{eq:jensen} and the estimate~\eqref{small increment} give
\begin{align*}
I_1 \leq&  
\left(\sum_{k=0}^{n-1} e^{4k\gamma h}h\right)^{p/2-1}
\sum_{k=0}^{n-1} e^{4k\gamma h}h\cdot L^p\E\left[\|\widehat{Y}_{{2k+1}}^f- \widehat{Y}_{{2k}}^f\|^p\right]  h^{p/2}  \\
\lesssim& (4\gamma)^{-p/2}e^{2p\gamma  nh}\cdot L^p \cdot 2^p p^{p/2}d^{p/2}h^{p/2}\cdot h^{p/2}
\lesssim\gamma^{-p/2} L^p p^{p/2}d^{p/2} h^p e^{2p\gamma nh}.
\end{align*}

For $I_2,$ by the mean value theorem, there exists $\theta\in[0,1]$ such that
\begin{align*}
    \nabla f(\widehat{Y}_{{2k+1}}^f)- \nabla f(\widehat{Y}_{{2k}}^f)
    =\nabla^2f(\theta \widehat{Y}_{{2k}}^f+(1-\theta)\widehat{Y}_{{2k+1}}^f)(\widehat{Y}_{{2k+1}}^f-\widehat{Y}_{{2k}}^f).
\end{align*}
Hence
\begin{align*}
&\left|\langle \widehat{Y}_{{2k}}^f- \widehat{Y}_{{2k}}^c,\nabla f(\widehat{Y}_{{2k+1}}^f)- \nabla f(\widehat{Y}_{{2k}}^f) \rangle 
- \langle \widehat{Y}_{{2k}}^f- \widehat{Y}_{{2k}}^c,\nabla^2 f(\widehat{Y}_{{2k}}^f)(\widehat{Y}_{{2k+1}}^f- \widehat{Y}_{{2k}}^f) \rangle\right|\\
=& \left|\langle \widehat{Y}_{{2k}}^f- \widehat{Y}_{{2k}}^c,
\left(\nabla^2f(\theta \widehat{Y}_{{2k}}^f+(1-\theta)\widehat{Y}_{{2k+1}}^f)-\nabla^2 f(\widehat{Y}_{{2k}}^f)\right)(\widehat{Y}_{{2k+1}}^f-\widehat{Y}_{{2k}}^f)\rangle\right|
\leq L\|\widehat{Y}_{{2k}}^f- \widehat{Y}_{{2k}}^c\|\|\widehat{Y}_{{2k+1}}^f- \widehat{Y}_{{2k}}^f \|^2.
\end{align*}
Therefore, by Jensen's inequality, we have
\begin{align*}
I_2  \leq & 4^{p/2-1}(J_1 + J_2 +J_3+J_4),
\end{align*}
where
\begin{align*}
J_1=&\E\left[
\left|\sum_{k=0}^{n-1} e^{4k\gamma h}\langle \widehat{Y}_{{2k}}^f- \widehat{Y}_{{2k}}^c,\nabla^2 f(\widehat{Y}_{{2k}}^f)\nabla f(\widehat{Y}_{{2k}}^f) \rangle h^2 \right|^{p/2}\right],
\\
J_2 =&\E\left[
\left|\sum_{k=0}^{n-1} e^{4k\gamma h} S \langle \widehat{Y}_{{2k}}^f- \widehat{Y}_{{2k}}^c,\nabla^2 f(\widehat{Y}_{{2k}}^f)(\widehat{Y}_{{2k}}^c-\widehat{Y}_{{2k}}^f) \rangle h^2 \right|^{p/2}\right],
\\
J_3 =&\E\left[
\left|\sqrt{2}\sum_{k=0}^{n-1} e^{4k\gamma h}\langle \widehat{Y}_{{2k}}^f- \widehat{Y}_{{2k}}^c,\nabla^2 f(\widehat{Y}_{{2k}}^f) \Delta W_{2k} \rangle h \right|^{p/2}\right],\\
J_4 =&\E\left[
\left|\sum_{k=0}^{n-1} e^{4k\gamma h}L\|\widehat{Y}_{{2k}}^f- \widehat{Y}_{{2k}}^c\|\|\widehat{Y}_{{2k+1}}^f- \widehat{Y}_{{2k}}^f \|^2 h\ \right|^{p/2}\right].
\end{align*}

For $J_1$, applying the weighted Jensen inequality~\eqref{eq:jensen} 
together with Corollary~\ref{cor: bound of nabla f}, in particular
$$
\sup_{0\leq n\leq 2N}\mathbb E\left[
\|\nabla f(\widehat Y^f_{n})\|^p\right]
\lesssim 2^{p-1}L^p C^p d^{p/2}p^{p/2}+2^{p-1}\|\nabla f(0)\|^p,
$$
we obtain
\begin{align*}
J_1  &\leq  \left(\sum_{k=0}^{n-1} e^{2\gamma kh}h\right)^{p/2-1} h^{p/2} 
 \sum_{k=0}^{n-1} e^{2\gamma kh}h \cdot e^{p\gamma  {k}h}\E\left[ \| \widehat{Y}_{{2k}}^f- \widehat{Y}_{{2k}}^c\|^{p/2}\|\nabla^2 f(\widehat{Y}_{{2k}}^f)\nabla f(\widehat{Y}_{{2k}}^f) \|^{p/2}\right]\\
& \leq  \left(\sum_{k=0}^{n-1} e^{2\gamma kh}h\right)^{p/2-1} h^{p/2} 
 \sum_{k=0}^{n-1} e^{2\gamma kh}h \cdot e^{p\gamma  {k}h}\E\left[ \| \widehat{Y}_{{2k}}^f- \widehat{Y}_{{2k}}^c\|^{p}\right]^{1/2}\mathbb{E}\left[\|\nabla^2 f(\widehat{Y}_{{2k}}^f)\nabla f(\widehat{Y}_{{2k}}^f) \|^{p}\right]^{1/2}\\
 & \leq  \left(\sum_{k=0}^{n-1} e^{2\gamma kh}h\right)^{p/2} 
 \left(\sup_{0\leq k\leq n-1}\E\left[e^{2p\gamma kh} \| \widehat{Y}_{{2k}}^f- \widehat{Y}_{{2k}}^c\|^{p}\right]\right)^{1/2}
 \left(\sup_{0\leq k\leq n-1}\mathbb{E}\left[\|\nabla^2 f(\widehat{Y}_{{2k}}^f)\nabla f(\widehat{Y}_{{2k}}^f) \|^{p}h^p\right]\right)^{1/2}\\
& \leq\frac{1}{4\zeta}  
\sup_{0\leq k\leq n-1}\E\left[e^{2p\gamma  kh}
\| \widehat{Y}_{{2k}}^f- \widehat{Y}_{{2k}}^c\|^{p}\right] 
+ \zeta  \left(\sum_{k=0}^{n-1} e^{2\gamma kh}h\right)^{p}
\sup_{0\leq k\leq n-1}\E\left[\|\nabla^2 f(\widehat{Y}_{{2k}}^f)\nabla f(\widehat{Y}_{{2k}}^f) \|^{p}h^{p}\right]\\
&\lesssim \frac{1}{4\zeta}  
\sup_{0\leq k\leq n-1}\E\left[e^{2p\gamma  kh}
\| \widehat{Y}_{{2k}}^f- \widehat{Y}_{{2k}}^c\|^{p}\right]  
+ \zeta  (2\gamma)^{-p}e^{2p\gamma nh}\cdot L^p 
\cdot\left(2^{p-1}L^p C^p d^{p/2}p^{p/2}+2^{p-1}\|\nabla f(0)\|^p\right) \cdot h^p\\
&\lesssim \frac{1}{4\zeta}  
\sup_{0\leq k\leq n-1}\E\left[e^{2p\gamma  kh}
\| \widehat{Y}_{{2k}}^f- \widehat{Y}_{{2k}}^c\|^{p}\right] 
+ \zeta  \gamma^{-p} L^{2p}C^pd^{p/2}p^{p/2}  h^p e^{2p\gamma  {n}h}
+\zeta  \gamma^{-p} L^{p}\|\nabla f(0)\|^p  h^p e^{2p\gamma  {n}h}
\end{align*}
for any $\zeta>0.$ 

For 
\begin{align*}
    J_2 =&\E\left[
\left|\sum_{k=0}^{n-1} e^{4k\gamma h} S \langle \widehat{Y}_{{2k}}^f- \widehat{Y}_{{2k}}^c,\nabla^2 f(\widehat{Y}_{{2k}}^f)(\widehat{Y}_{{2k}}^c-\widehat{Y}_{{2k}}^f) \rangle h^2 \right|^{p/2}\right],
\end{align*}
applying the weaker bound established earlier in~\eqref{eq:weaker bound of distance}, namely
$$
\E\left[ e^{\gamma p {n} h}
\|\widehat{Y}_{{2n}}^f- \widehat{Y}_{{2n}}^c\|^p \right]
\lesssim
2^{p/2-1}\mathbb{E}\left[\|\widehat{Y}_{0}^f- \widehat{Y}_{0}^c\|^p\right]
+2^{2p} L^p\gamma^{-p}    p^{p/2}d^{p/2}h^{p/2} e^{\gamma p nh},
$$
we obtain
\begin{align*}
J_2
&\leq \E\Bigg[
\left(\sum_{k=0}^{n-1} e^{2\gamma kh}h\right)^{p/2-1}
\sum_{k=0}^{n-1} e^{2\gamma kh}h
\cdot e^{p\gamma kh} S^{p/2}
\|\widehat{Y}_{2k}^f-\widehat{Y}_{2k}^c\|^{p/2}\cdot
\|\nabla^2 f(\widehat{Y}_{2k}^f)
(\widehat{Y}_{2k}^c-\widehat{Y}_{2k}^f)\|^{p/2}
h^{p/2}
\Bigg]\\
&\leq S^{p/2}L^{p/2}h^{p/2}
\left(\sum_{k=0}^{n-1} e^{2\gamma kh}h\right)^{p/2-1}
\sum_{k=0}^{n-1} e^{2\gamma kh}h
\cdot e^{p\gamma kh}
\E\left[
\|\widehat{Y}_{2k}^f-\widehat{Y}_{2k}^c\|^p
\right]\\
&\lesssim S^{p/2}L^{p/2}h^{p/2}
\left(\sum_{k=0}^{n-1} e^{2\gamma kh}h\right)^{p/2-1}
\sum_{k=0}^{n-1} e^{2\gamma kh}h
\Bigg(
2^{p/2-1}
\E\left[
\|\widehat{Y}_0^f-\widehat{Y}_0^c\|^p
\right]+
2^{2p}L^p\gamma^{-p}p^{p/2}d^{p/2}
h^{p/2}e^{p\gamma kh}
\Bigg)\\
&\lesssim
\left(\frac{SL}{2\gamma}\right)^{p/2}
\Bigg(
2^{p/2-1}h^{p/2}e^{p\gamma nh}
\E\left[
\|\widehat{Y}_0^f-\widehat{Y}_0^c\|^p
\right]+
2^{2p}L^p\gamma^{-p}p^{p/2}d^{p/2}
h^pe^{2p\gamma nh}
\Bigg)
\end{align*}
for any $\zeta>0.$

For 
\begin{align*}
    J_3 =&\E\left[
\left|\sqrt{2}\sum_{k=0}^{n-1} e^{4k\gamma h}\langle \widehat{Y}_{{2k}}^f- \widehat{Y}_{{2k}}^c,\nabla^2 f(\widehat{Y}_{{2k}}^f) \Delta W_{2k} \rangle h \right|^{p/2}\right],
\end{align*}
by the Burkholder-Davis-Gundy inequality in \cite{BARLOW1982198}
and the weighted Jensen inequality~\eqref{eq:jensen_4}, 
for any $\zeta>0,$
\begin{align*}
J_3
&\leq (C_{\mathrm{BDG}}p)^{p/4}
\E\Bigg[
\Bigg|
2\sum_{k=0}^{n-1}e^{8k\gamma h}
\|\widehat{Y}_{2k}^f-\widehat{Y}_{2k}^c\|^2
\|\nabla^2 f(\widehat{Y}_{2k}^f)\|^2
h^3
\Bigg|^{p/4}
\Bigg]\\
&\leq
(2C_{\mathrm{BDG}}pL^2h^2)^{p/4}
\left(\sum_{k=0}^{n-1}e^{4k\gamma h}h\right)^{p/4-1}
\sum_{k=0}^{n-1}e^{4k\gamma h}h
\cdot e^{p\gamma kh}
\E\left[
\|\widehat{Y}_{2k}^f-\widehat{Y}_{2k}^c\|^{p/2}
\right]\\
&\leq
(2C_{\mathrm{BDG}}pL^2h^2)^{p/4}
\left(\sum_{k=0}^{n-1}e^{4k\gamma h}h\right)^{p/4}
\left(
\sup_{0\leq k\leq n-1}
\E\left[
e^{2p\gamma kh}
\|\widehat{Y}_{2k}^f-\widehat{Y}_{2k}^c\|^p
\right]
\right)^{1/2}\\
&\leq
\frac{1}{4\zeta}
\sup_{0\leq k\leq n-1}
\E\left[
e^{2p\gamma kh}
\|\widehat{Y}_{2k}^f-\widehat{Y}_{2k}^c\|^p
\right]+
\zeta(2C_{\mathrm{BDG}}pL^2h^2)^{p/2}
\left(\sum_{k=0}^{n-1}e^{4k\gamma h}h\right)^{p/2}\\
&\leq
\frac{1}{4\zeta}
\sup_{0\leq k\leq n-1}
\E\left[
e^{2p\gamma kh}
\|\widehat{Y}_{2k}^f-\widehat{Y}_{2k}^c\|^p
\right]+\zeta C_{\mathrm{BDG}}^{p/2}2^{-p/2}
\gamma^{-p/2}p^{p/2}L^ph^p
e^{2p\gamma nh}.
\end{align*}

For 
\begin{align*}
    J_4 =&\E\left[
\left|\sum_{k=0}^{n-1} e^{4k\gamma h}L\|\widehat{Y}_{{2k}}^f- \widehat{Y}_{{2k}}^c\|\|\widehat{Y}_{{2k+1}}^f- \widehat{Y}_{{2k}}^f \|^2 h\ \right|^{p/2}\right],
\end{align*}
we have

\begin{align*}
J_4
&\leq
\E\Bigg[
\left(\sum_{k=0}^{n-1}e^{2\gamma kh}h\right)^{p/2-1}
\sum_{k=0}^{n-1}e^{2\gamma kh}h
\cdot e^{p\gamma kh}L^{p/2}
\|\widehat{Y}_{2k}^f-\widehat{Y}_{2k}^c\|^{p/2}
\|\widehat{Y}_{2k+1}^f-\widehat{Y}_{2k}^f\|^p
\Bigg]\\
&=
\left(\sum_{k=0}^{n-1}e^{2\gamma kh}h\right)^{p/2-1}
\sum_{k=0}^{n-1}e^{2\gamma kh}h
\cdot e^{p\gamma kh}L^{p/2}\cdot
\E\left[
\|\widehat{Y}_{2k}^f-\widehat{Y}_{2k}^c\|^{p/2}
\|\widehat{Y}_{2k+1}^f-\widehat{Y}_{2k}^f\|^p
\right]\\
&\leq
\left(\sum_{k=0}^{n-1}e^{2\gamma kh}h\right)^{p/2-1}
\sum_{k=0}^{n-1}e^{2\gamma kh}h
\cdot e^{p\gamma kh}L^{p/2}\cdot
\left(
\E\left[
\|\widehat{Y}_{2k}^f-\widehat{Y}_{2k}^c\|^p
\right]
\right)^{1/2}
\left(
\E\left[
\|\widehat{Y}_{2k+1}^f-\widehat{Y}_{2k}^f\|^{2p}
\right]
\right)^{1/2}\\
&\lesssim
6^pL^{p/2}d^{p/2}p^{p/2}h^{p/2}
\left(\sum_{k=0}^{n-1}e^{2\gamma kh}h\right)^{p/2}\cdot
\left(
\sup_{0\leq k\leq n-1}
\E\left[
e^{2p\gamma kh}
\|\widehat{Y}_{2k}^f-\widehat{Y}_{2k}^c\|^p
\right]
\right)^{1/2}\\
&\lesssim
\frac{1}{4\zeta}
\sup_{0\leq k\leq n-1}
\E\left[
e^{2p\gamma kh}
\|\widehat{Y}_{2k}^f-\widehat{Y}_{2k}^c\|^p
\right]+
\zeta\,6^{2p}L^pd^pp^ph^p
\left(\sum_{k=0}^{n-1}e^{2\gamma kh}h\right)^p\\
&\lesssim
\frac{1}{4\zeta}
\sup_{0\leq k\leq n-1}
\E\left[
e^{2p\gamma kh}
\|\widehat{Y}_{2k}^f-\widehat{Y}_{2k}^c\|^p
\right]+
\zeta\,6^{2p}(2\gamma)^{-p}
L^pd^pp^ph^pe^{2p\gamma nh}
\end{align*}
 for any $\zeta>0.$

In summary,
\begin{align}\label{eq: before two cases}
&\E\left[e^{2p\gamma nh}
\|\widehat{Y}_{2n}^f-\widehat{Y}_{2n}^c\|^p\right]\notag\\
\leq&
3^{p/2-1}6^{p/2}
\left[
\E\left[\|\widehat{Y}_0^f-\widehat{Y}_0^c\|^p\right]
+I_1+4^{p/2-1}(J_1+J_2+J_3+J_4)
\right]\notag\\
\lesssim&
3^{p-1}2^{p/2}
\Bigg\{
\E\left[\|\widehat{Y}_0^f-\widehat{Y}_0^c\|^p\right]+
2^{p-2}\cdot\frac{1}{\zeta}
\sup_{0\leq k\leq n-1}
\E\left[
e^{2p\gamma kh}
\|\widehat{Y}_{2k}^f-\widehat{Y}_{2k}^c\|^p
\right]\notag\\
&\quad+
2^{3p/2-3}
\left(\frac{SL}{2\gamma}\right)^{p/2}
h^{p/2}e^{p\gamma nh}
\E\left[
\|\widehat{Y}_0^f-\widehat{Y}_0^c\|^p
\right]+
e^{2p\gamma nh}h^p\cdot I
\Bigg\},
\end{align}
where
\begin{align*}
I={}&
\gamma^{-p/2}L^p p^{p/2}d^{p/2}+2^{p-2}\Bigg(
\zeta\gamma^{-p}L^{2p}C^pd^{p/2}p^{p/2}+
\zeta\gamma^{-p}L^p\|\nabla f(0)\|^p\\
&\quad+
2^{2p}\left(\frac{SL}{2\gamma}\right)^{p/2}
L^p\gamma^{-p}p^{p/2}d^{p/2}+
\zeta C_{\mathrm{BDG}}^{p/2}2^{-p/2}
\gamma^{-p/2}p^{p/2}L^p+
\zeta 6^{2p}(2\gamma)^{-p}
L^pd^pp^p
\Bigg).
\end{align*}

When $\zeta$ is a constant, $\|\nabla f(0)\| = \mathcal{O}(L\sqrt{d})$, and 
$h = \mathcal{O}(\gamma^2 S^{-2})$, we have
\begin{align*}
    I=\mathcal{O}\left(\left(\frac{L^2\sqrt{d}+Ld}{\gamma}p\right)^p\right).
\end{align*}

In the following, we consider two cases.

In the first case, suppose that $\widehat{Y}_{0}^f= \widehat{Y}_{0}^c$. 
Then we have
\begin{align*}
   & \E\left[ e^{2p\gamma n h }
\|\widehat{Y}_{{2n}}^f- \widehat{Y}_{{2n}}^c\|^p \right]
\lesssim
 3^{p-1}2^{3p/2-2}
 \frac{1}{\zeta}  
\sup_{0\leq k\leq n-1}\E\left[e^{2p\gamma  kh}
\| \widehat{Y}_{{2k}}^f- \widehat{Y}_{{2k}}^c\|^{p}\right]
+ 3^{p-1}2^{p/2}  e^{2p\gamma  {n}h} h^p\cdot I,
\end{align*}
which implies that
\begin{align*}
   & \E\left[ 
\|\widehat{Y}_{{2n}}^f- \widehat{Y}_{{2n}}^c\|^p \right]
\lesssim
 3^{p-1}2^{3p/2-2}
 \frac{1}{\zeta}  
\sup_{0\leq k\leq n-1}\E\left[
\| \widehat{Y}_{{2k}}^f- \widehat{Y}_{{2k}}^c\|^{p}\right]
+ 3^{p-1}2^{p/2}  h^p\cdot I.
\end{align*}
Taking the supremum over $0 \leq n \leq N$, we obtain
\begin{align*}
   & \sup_{0\leq n\leq N}\E\left[ 
\|\widehat{Y}_{{2n}}^f- \widehat{Y}_{{2n}}^c\|^p \right]
\lesssim
 3^{p-1}2^{3p/2-2}
 \frac{1}{\zeta}  
\sup_{0\leq k\leq N}\E\left[
\| \widehat{Y}_{{2k}}^f- \widehat{Y}_{{2k}}^c\|^{p}\right]
+ 3^{p-1}2^{p/2}   h^p\cdot I.
\end{align*}
Choosing $\zeta = 2^{3p/2-1}3^{p-1}$ gives
\begin{align*}
   &\sup_{0\leq n\leq N} \E\left[ 
\|\widehat{Y}_{{2n}}^f- \widehat{Y}_{{2n}}^c\|^p \right]
\lesssim
 3^{p-1}2^{p/2+1}   h^p\cdot I.
\end{align*}

For the second case, suppose that $\widehat{Y}_{0}^f\neq \widehat{Y}_{0}^c$. 
Taking the supremum over $0 \leq n \leq N$ on both sides of~\eqref{eq: before two cases},
\begin{align*}
      &  \sup_{0\leq n\leq N}\E\left[ e^{2p\gamma n h }
\|\widehat{Y}_{{2n}}^f- \widehat{Y}_{{2n}}^c\|^p \right]
\lesssim
3^{p-1}2^{3p/2-2}\cdot\frac{1}{\zeta}  
\sup_{0\leq k\leq N}\E\left[e^{2p\gamma  kh}
\| \widehat{Y}_{{2k}}^f- \widehat{Y}_{{2k}}^c\|^{p}\right]\\
&\quad + \left(3^{p-1}2^{p/2}+3^{p-1}2^{3p/2-3}\zeta \gamma^{-p} L^p S^p h^p
 e^{p \gamma nh}\right)\mathbb{E}\left[\|\widehat{Y}_{0}^f- \widehat{Y}_{0}^c\|^p\right]
+ 3^{p-1}2^{p/2}e^{2p\gamma  {n}h} h^p\cdot I.
\end{align*}
Similarly, we take $\zeta = 2^{3p/2-1}3^{p-1}$,
\begin{align*}
      & \sup_{0\leq n\leq N} \E\left[ e^{2p\gamma n h }
\|\widehat{Y}_{{2n}}^f- \widehat{Y}_{{2n}}^c\|^p \right]\\
\lesssim
&2\left(3^{p-1}2^{p/2}+3^{p-1}2^{3p/2-3}\zeta \gamma^{-p} L^p S^ph^p
 e^{p \gamma nh}\right)\mathbb{E}\left[\|\widehat{Y}_{0}^f- \widehat{Y}_{0}^c\|^p\right]
+ 3^{p-1}2^{p/2+1}e^{2p\gamma  {n}h} h^p\cdot I.
\end{align*}
In particular, when $n = N$, we have
\begin{align*}
 &\E\left[\|\widehat{Y}_{{2N}}^f- \widehat{Y}_{{2N}}^c\|^p \right]\\
\lesssim
&2\left(3^{p-1}2^{p/2}e^{-p \gamma T/2}+3^{p-1}2^{3p/2-3}\zeta \gamma^{-p} L^p S^ph^p
 \right)e^{-p \gamma T/2}\mathbb{E}\left[\|\widehat{Y}_{0}^f- \widehat{Y}_{0}^c\|^p\right]
+ 3^{p-1}2^{p/2+1} h^p\cdot I.
\end{align*}

Finally, we obtain the following result.

There exist constants $C_3, C_4, C_{(1)}, C_{(2)} > 0$
with
$$
C_{(1)}=\mathcal{O}\left(\frac{L\sqrt{d} }{2S-\lambda}\right),\quad C_{(2)} = \mathcal{O}\left(\frac{L^2\sqrt{d}+Ld}{2S-\lambda}\right),
$$
and for all
 $$0<h\leq h_{(1)}:=C_3\min\left\{S^{-1},\tilde{\alpha} L^{-2},\tilde{\alpha} \|\nabla f(0)\|^{-2},\tilde{\alpha}^{-1}
(2S-\lambda)^{-1},\tilde{\alpha} S^{-2}, \|\nabla f(0)\|^{-2},(1-\lambda/S)^2
\right\},$$
the following holds.

\medskip

\noindent
\textbf{(1) when $ \widehat{Y}_{0}^f= \widehat{Y}_{0}^c$:} 
\begin{align*}
   \sup_{0\leq n\leq N} \E\left[
\|\widehat{Y}_{{2n}}^f- \widehat{Y}_{{2n}}^c\|^p \right]
\leq C_{(1)}^p p^{p/2}h^{p/2},
\quad 
\sup_{0\leq n\leq N} \E\left[ 
\|\widehat{Y}_{{2n}}^f- \widehat{Y}_{{2n}}^c\|^p \right]
\leq
C_{(2)}^p p^{p}h^{p}.
\end{align*}

\medskip

\noindent
\textbf{(2) when $ \widehat{Y}_{0}^f\neq \widehat{Y}_{0}^c$:} 
\begin{align*}
 \E\left[\|\widehat{Y}_{{2N}}^f- \widehat{Y}_{{2N}}^c\|^p \right]
 \leq
&C_4^p e^{-p \gamma T/2}\mathbb{E}\left[\|\widehat{Y}_{0}^f- \widehat{Y}_{0}^c\|^p\right]
+C_{(2)}^p p^{p}h^{p},
\end{align*}
where $\gamma\in(0,2S-\lambda)$ and $\gamma=\Theta(2S-\lambda)$.

\end{proof}

It is worth noting that the proofs of the preceding two 
lemmas remain valid in the regime $S=0$ and $\lambda=-m<0$,
 and therefore extend to all algorithms and results considered in this paper. 
In particular, we have the following corollary.
\begin{corollary}\label{cor: p moment distance of one-sided}
Assume that $f$ is $L$-smooth, $L$-Hessian smooth, and satisfies the one-sided Lipschitz condition with constant $m$. 
Consider the coupled scheme~\eqref{SDE:discrete couple}. 
Then there exist constants $h_{(2)},C_{(3)} > 0$ such that
$$
h_{(2)} = C_3 \min\left\{
m L^{-2},\,
m\|\nabla f(0)\|^{-2},\,
\tilde{\alpha}^{-1} m^{-1},\,
\|\nabla f(0)\|^{-2}
\right\},
$$
and
$$
C_{(3)} = \mathcal{O}\!\left(\frac{L^2\sqrt{d}+Ld}{m}\right).
$$
For any $0<h\leq h_{(2)}$ and $p\geq 1$, we have
\begin{align*}
 \E\left[\|\widehat{X}_{2N}^f- \widehat{X}_{2N}^c\|^p \right]
 \leq
 C_4^p e^{-p \gamma T/2}
 \mathbb{E}\left[\|\widehat{X}_{0}^f- \widehat{X}_{0}^c\|^p\right]
 + C_{(3)}^p p^{p} h^{p},
\end{align*}
where $\gamma \in (0,m)$ and $\gamma = \Theta(m)$.
\end{corollary}

The proofs of the following two lemmas rely solely on the change-of-measure technique introduced in Subsection~\ref{subsec:Quantum-Accelerated Multilevel Monte Carlo with Change of Measure}. 
In this setting, $\widehat{Y}_{0}^f = \widehat{Y}_{0}^c$, and hence the conclusion 
of Lemma~\ref{lem:p-bounds of distance of two paths}
can be simplified as
\begin{align*}
 \sup_{0 \leq n \leq N}\E\left[
\|\widehat{Y}_{2n}^f - \widehat{Y}_{2n}^c\|^p \right] 
\leq  C_{(1)}^p p^{p/2} h^{p/2}, \quad 
\sup_{0 \leq n \leq N}\E\left[ 
\|\widehat{Y}_{2n}^f - \widehat{Y}_{2n}^c\|^{p} \right] 
\leq  C_{(2)}^p p^{p} h^p.
\end{align*}
Moreover,  we assume that $\varphi$ is $K$-globally Lipschitz and bounded below.
\begin{lemma}\label{lem:bound of R-D}
   For any $p\geq 1$,
there exist  constant $h_{(3)}=\min\left\{h_{(1)},\frac{1}{2ep(8p-1)TS^2C_{(1)}^2}\right\}>0$
 such that,
for all  $0<h\leq h_{(3)}$,
\begin{align*}
    \E\left[\left| \frac{\d \widehat{\mathbb{Q}}^c}{\d \mathbb{P}}\right|^p\right]
    \leq 2.
\end{align*}
\end{lemma}
\begin{proof}
    Since
  $$  \frac{\d \widehat{\mathbb{Q}}^c}{\d \mathbb{P}}= \prod_{n=0}^{N-1}  R \left( \widehat{Y}_{{2n+2}}^c ,\widehat{Y}_{{2n}}^c, S(\widehat{Y}_{{2n}}^f-\widehat{Y}_{{2n}}^c), 2h\right),$$
  we have
  \begin{align*}
    &\E\left[\left| \frac{\d \widehat{\mathbb{Q}}^c}{\d \mathbb{P}}\right|^p\right]
    =\E\left[\prod_{n=0}^{N-1}  \left|R \left( \widehat{Y}_{{2n+2}}^c ,\widehat{Y}_{{2n}}^c, S(\widehat{Y}_{{2n}}^f-\widehat{Y}_{{2n}}^c), 2h\right)\right|^p\right]\\
    =&\E\left[\exp\left(-\sqrt{2}pS\sum_{n=0}^{N-1}\langle\widehat{Y}_{{2n}}^f-\widehat{Y}_{{2n}}^c,\Delta W_{2n}+\Delta W_{2n+1}\rangle-\frac{1}{2}pS^2h\sum_{n=0}^{N-1}\|\widehat{Y}_{{2n}}^f-\widehat{Y}_{{2n}}^c\|^2 \right)\right]\\
      =&\E\Bigg[\exp\left(p(4p-1/2)S^2h\sum_{n=0}^{N-1}\|\widehat{Y}_{{2n}}^f-\widehat{Y}_{{2n}}^c\|^2\right)\\
      &\quad\quad\times\exp\left(-\sqrt{2}pS\sum_{n=0}^{N-1}\langle\widehat{Y}_{{2n}}^f-\widehat{Y}_{{2n}}^c,\Delta W_{2n}+\Delta W_{2n+1}\rangle-4p^2S^2h\sum_{n=0}^{N-1}\|\widehat{Y}_{{2n}}^f-\widehat{Y}_{{2n}}^c\|^2 \right)\Bigg]\\ 
\leq&\E\left[\exp\left(p(8p-1)S^2h\sum_{n=0}^{N-1}\|\widehat{Y}_{{2n}}^f-\widehat{Y}_{{2n}}^c\|^2\right)\right]^{\frac{1}{2}}\\
 &\quad\quad \times\E\left[\exp\left(-2\sqrt{2}pS\sum_{n=0}^{N-1}\langle\widehat{Y}_{{2n}}^f-\widehat{Y}_{{2n}}^c,\Delta W_{2n}+\Delta W_{2n+1}\rangle-8p^2S^2h\sum_{n=0}^{N-1}\|\widehat{Y}_{{2n}}^f-\widehat{Y}_{{2n}}^c\|^2 \right)\right]^{\frac{1}{2}}.
  \end{align*}
If we denote $u_n=-4\sqrt{h}pS\left(\widehat{Y}_{{2n}}^f-\widehat{Y}_{{2n}}^c\right)$ and $Z_n=\frac{1}{\sqrt{2h}}(\Delta W_{2n}+\Delta W_{2n+1})$, 
then $Z_n\sim \mathcal{N}(0,I_d)$. By the Gaussian exponential identity, we have
\begin{align*}
    &\E\left[\exp\left(-2\sqrt{2}pS\langle\widehat{Y}_{{2n}}^f-\widehat{Y}_{{2n}}^c,\Delta W_{2n}+\Delta W_{2n+1}\rangle-8p^2S^2h\|\widehat{Y}_{{2n}}^f-\widehat{Y}_{{2n}}^c\|^2 \right)\right]\\
    =&\E\left[\exp\left(\langle u_n,Z_n\rangle-\frac{1}{2}\|u_n\|^2\right)\right]=1.
\end{align*}
On the other hand,
\begin{align*}
   &\E\left[\exp\left(p(8p-1)S^2h\sum_{n=0}^{N-1}\|\widehat{Y}_{{2n}}^f-\widehat{Y}_{{2n}}^c\|^2\right)\right]\\
    \leq& \sum_{k=0}^\infty \frac{\E\left[\left(p(8p-1)S^2h\sum\limits_{n=0}^{N-1}\|\widehat{Y}_{{2n}}^f-\widehat{Y}_{{2n}}^c\|^2\right)^k\right]}{k!}\\
    \leq&\sum_{k=0}^\infty \frac{N^{k-1}\left(p(8p-1)S^2h\right)^k\sum\limits_{n=0}^{N-1}\E\left[\|\widehat{Y}_{{2n}}^f-\widehat{Y}_{{2n}}^c\|^{2k}\right]}{k!}.
\end{align*}
Using  Lemma~\ref{lem:p-bounds of distance of two paths} and the Stirling's approximation
$k!\geq \sqrt{2\pi}k^{k+1/2}e^{-k}$ for any $k\geq 1$, we have
\begin{align*}
   &\E\left[\exp\left(p(8p-1)S^2h\sum_{n=0}^{N-1}\|\widehat{Y}_{{2n}}^f-\widehat{Y}_{{2n}}^c\|^2\right)\right]\\
        \leq&1+
    \sum_{k=1}^\infty \frac{N^{k-1}\left(p(8p-1)S^2h\right)^k\cdot
    N C_{(1)}^{2k} (2k)^{k} h^{k}}{\sqrt{2\pi}k^{k+1/2}e^{-k}}\\
   \leq &1+\sum_{k=1}^\infty \frac{\left(e p(8p-1)TS^2C_{(1)}^2 h\right)^k}{\sqrt{2\pi k}}
   < 2,
\end{align*}
when $0<h\leq h_{(1)}$ and $e p(8p-1)TS^2C_{(1)}^2 h<1/2$.

Therefore, for any $p\geq 1$,
there exist constant $h_{(3)}=\min\left\{h_{(1)},\frac{1}{2ep(8p-1)TS^2C_{(1)}^2}\right\}>0$ such that,
for all  $0<h\leq h_{(3)}$,
\begin{align*}
    \E\left[\left| \frac{\d \widehat{\mathbb{Q}}^c}{\d \mathbb{P}}\right|^p\right]
    \leq 2.
\end{align*}
\end{proof}

\begin{lemma}\label{lem:bound of two path with R-D}
  There exist constants $C_5>0$ and
   \begin{align*}
    h_{(4)}=C_5\min\Bigg\{\frac{1}{S},\,\frac{\tilde{\alpha}}{S^2},\,\frac{\tilde{\alpha}}{L^2},\,\frac{\tilde{\alpha}}{2S-\lambda},\,\frac{\min\{1,\tilde\alpha\}}{\|\nabla f(0)\|^2},\,\left(1-\frac{\lambda}{S}\right)^2,\,\frac{(2S-\lambda)^2}{T^2S^2L^2d\,},\,\frac{2S-\lambda}{\sqrt{T}S(L^2+L\sqrt{d})}\,\Bigg\}
    \end{align*}   such that
     for any $p\geq 1$ and $0<h\leq h_{(4)}$, we have
\begin{align*}
        &\E\left[\left|\varphi(\widehat{Y}_{2N}^f)\frac{\d\widehat{\mathbb{Q}}^f}{\d \mathbb{P}}-\varphi(\widehat{Y}_{2N}^c)\frac{\d\widehat{\mathbb{Q}}^c}{\d \mathbb{P}}\right|^p\right]^{1/p}
    =\mathcal{O}\left(\sqrt{T}\cdot\frac{SLd(L+\sqrt{d})}{2S-\lambda}h\right).
\end{align*}
\end{lemma}

\begin{proof}
    \begin{align*}
        &\E\left[\left|\varphi(\widehat{Y}_{2N}^f)\frac{\d\widehat{\mathbb{Q}}^f}{\d \mathbb{P}}-\varphi(\widehat{Y}_{2N}^c)\frac{\d\widehat{\mathbb{Q}}^c}{\d \mathbb{P}}\right|^p\right]\\
\leq &3^{p-1} \E\left[\left|\varphi(\widehat{Y}_{2N}^f)\frac{\d\widehat{\mathbb{Q}}^f}{\d \mathbb{P}}-\varphi(\widehat{Y}_{2N}^f)\right|^p\right]
+3^{p-1}\E\left[\left|\varphi(\widehat{Y}_{2N}^f)-\varphi(\widehat{Y}_{2N}^c)\right|^p\right]
+3^{p-1}\E\left[\left|\varphi(\widehat{Y}_{2N}^c)-\varphi(\widehat{Y}_{2N}^c)\frac{\d\widehat{\mathbb{Q}}^c}{\d \mathbb{P}}\right|^p\right]\\
\leq &3^{p-1}K^p\E\left[\|\widehat{Y}_{2N}^f-\widehat{Y}_{2N}^c\|^p\right]
+3^{p-1} \E\left[\left|\varphi(\widehat{Y}_{2N}^f)\right|^{2p}\right]^{1/2}\E\left[\left|\frac{\d\widehat{\mathbb{Q}}^f}{\d \mathbb{P}}-1\right|^{2p}\right]^{1/2}\\
&\quad +3^{p-1} \E\left[\left|\varphi(\widehat{Y}_{2N}^c)\right|^{2p}\right]^{1/2}\E\left[\left|\frac{\d\widehat{\mathbb{Q}}^c}{\d \mathbb{P}}-1\right|^{2p}\right]^{1/2}.
   \end{align*}

Denote
\begin{align*}
    \mathcal{H}=-\sqrt{2}S\sum_{n=0}^{N-1}\langle\widehat{Y}_{{2n}}^f-\widehat{Y}_{{2n}}^c,\Delta W_{2n}+\Delta W_{2n+1}\rangle
    -\frac{1}{2}S^2h\sum_{n=0}^{N-1}\|\widehat{Y}_{{2n}}^f-\widehat{Y}_{{2n}}^c\|^2,
\end{align*}
Then Taylor expansion gives that
$$e^x=1+e^{\xi(x)}x,\ \text{for some}\ \xi(x)\ \text{with}\ |\xi(x)|<|x|,$$
combine with H{\"o}lder's inequality,  we have
\begin{align*}
 &\E\left[\left|1-\frac{\d \widehat{\mathbb{Q}}^c}{\d \mathbb{P}}\right|^{2p}\right] 
=\E\left[\left( \exp(\xi(\mathcal{H}))|\mathcal{H}|\right)^{2p}\right]
\leq \E\left[\left| \exp(\xi(\mathcal{H}))\right|^{4p}\right]^{1/2}
\E\left[\left|\mathcal{H}\right|^{4p}\right]^{1/2}.
\end{align*}

First, by Lemma~\ref{lem:bound of R-D},
\begin{align*}
    \E\left[\left| \exp(\xi(\mathcal{H}))\right|^{4p}\right]
    \leq \E\left[\max\left(\exp(4p\mathcal{H}),1\right)\right]
    \leq  \E\left[\exp(4p\mathcal{H})\right]+1\lesssim 3.
\end{align*}

By Jensen's inequality and the Burkholder--Davis--Gundy inequality in~\cite{BARLOW1982198},
there exists a constant $C'>0$, independent of $d$, such that
when $h<\min\{h_{(1)}, \Theta(S^{-1}C_{(2)}^{-1}T^{-1/2}d^{1/2})\}$
\begin{align*}
    \E\left[\left|\mathcal{H}\right|^{4p}\right]
\leq &
2^{4p-1}\cdot 2^{2p}S^{4p}\E\left[\left|\sum_{n=0}^{N-1}\langle\widehat{Y}_{{2n}}^f-\widehat{Y}_{{2n}}^c,\Delta W_{2n}+\Delta W_{2n+1}\rangle\right|^{4p}\right]
+2^{4p-1}\cdot 2^{-4p}S^{8p}\E\left[\left|\sum_{n=0}^{N-1}\|\widehat{Y}_{{2n}}^f-\widehat{Y}_{{2n}}^c\|h\right|^{4p}\right]
\\
\leq&2^{6p-1}S^{4p}(C')^p p^{2p}\E\left[\left(\sum_{n=0}^{N-1}\|\widehat{Y}_{{2n}}^f-\widehat{Y}_{{2n}}^c\|^2 dh\right)^{2p}\right]
+2^{-1}S^{8p}T^{4p-1}\E\left[\sum_{n=0}^{N-1}\|\widehat{Y}_{{2n}}^f-\widehat{Y}_{{2n}}^c\|^{8p}h\right]
\\
\leq&2^{6p-1}S^{4p}(C')^p p^{2p}T^{2p-1}\E\left[\sum_{n=0}^{N-1}\|\widehat{Y}_{{2n}}^f-\widehat{Y}_{{2n}}^c\|^{4p}d^{2p} h\right]
+2^{-1}S^{8p}T^{4p-1}\E\left[\sum_{n=0}^{N-1}\|\widehat{Y}_{{2n}}^f-\widehat{Y}_{{2n}}^c\|^{8p}h\right]
\\
\lesssim&2^{6p-1}S^{4p}(C')^p p^{2p}T^{2p-1}\cdot T d^{2p} \cdot C_{(2)}^{4p}(4p)^{4p}h^{4p}
+2^{-1}S^{8p}T^{4p-1}\cdot T\cdot C_{(2)}^{8p}(8p)^{8p}h^{8p}
\\
\lesssim&2^{14p}S^{4p}(C')^pp^{6p}T^{2p}C_{(2)}^{4p}d^{2p}h^{4p}.
\end{align*}
So
\begin{align*}
    \E\left[\left|1-\frac{\d \widehat{\mathbb{Q}}^c}{\d \mathbb{P}}\right|^{2p}\right] 
    =\mathcal{O}(T^pS^{2p}d^{p}C_{(2)}^{2p}h^{2p}).
\end{align*}
Similarly,
\begin{align*}
    \E\left[\left|1-\frac{\d \widehat{\mathbb{Q}}^f}{\d \mathbb{P}}\right|^{2p}\right] 
   =\mathcal{O}(T^pS^{2p}d^{p}C_{(2)}^{2p}h^{2p}).
\end{align*}
Hence, 
by Lemma~\ref{lem:bounds of p moment} and Lemma~\ref{lem:p-bounds of distance of two paths},
\begin{align*}
   & \E\left[\left|\varphi(\widehat{Y}_{2N}^f)\frac{\d\widehat{\mathbb{Q}}^f}{\d \mathbb{P}}-\varphi(\widehat{Y}_{2N}^c)\frac{\d\widehat{\mathbb{Q}}^c}{\d \mathbb{P}}\right|^p\right]
    =\mathcal{O}(C_{(2)}^{p}h^{p})+
    \mathcal{O}(T^{p/2}S^{p}d^{p/2}C_{(2)}^{p}h^{p})
    =\mathcal{O}\left(\left(\sqrt{T}\cdot\frac{SLd(L+\sqrt{d})}{2S-\lambda}\right)^ph^p\right),
\end{align*}
and
\begin{align*}
        &\E\left[\left|\varphi(\widehat{Y}_{2N}^f)\frac{\d\widehat{\mathbb{Q}}^f}{\d \mathbb{P}}-\varphi(\widehat{Y}_{2N}^c)\frac{\d\widehat{\mathbb{Q}}^c}{\d \mathbb{P}}\right|^p\right]^{1/p}
    =\mathcal{O}\left(\sqrt{T}\cdot\frac{SLd(L+\sqrt{d})}{2S-\lambda}h\right).
\end{align*}
  There exist constants $C_5>0$ and
   \small{ \begin{align*}
    h_{(4)}=&C_5\min\Bigg\{S^{-1},\tilde{\alpha} L^{-2},\tilde{\alpha} \|\nabla f(0)\|^{-2},\tilde{\alpha}^{-1}
(2S-\lambda)^{-1},\tilde{\alpha} S^{-2}, \|\nabla f(0)\|^{-2},
(1-\lambda/S)^2,
T^{-1}S^{-2}C_{(1)}^{-2},S^{-1}C_{(2)}^{-1}T^{-1/2}d^{1/2}
\Bigg\}    \\
=&C_5\min\Bigg\{\frac{1}{S},\,\frac{\tilde{\alpha}}{S^2},\,\frac{\tilde{\alpha}}{L^2},\,\frac{\tilde{\alpha}}{2S-\lambda},\,\frac{\min\{1,\tilde\alpha\}}{\|\nabla f(0)\|^2},\,\left(1-\frac{\lambda}{S}\right)^2,\,\frac{(2S-\lambda)^2}{T^2S^2L^2d\,},\,\frac{2S-\lambda}{\sqrt{T}S(L^2+L\sqrt{d})}\,\Bigg\}
    \end{align*}}
\normalsize    such that
     for any $p\geq 1$ and $0<h\leq h_{(4)}$ this holds.
\end{proof}

\section*{Appendix B: Technical Details for Section~\ref{sec:heavy-tailed}}
\label{Appendix B: Technical Details for heavy tailed}
\addcontentsline{toc}{section}{Appendix B: Technical Details for Section~\ref{sec:heavy-tailed}}

\begin{lemma}\label{lem:construct of chi}
Let $0<R_1<R_2$. There exists a function $\chi:[0,\infty)\to[0,1]$ such that
\begin{itemize}
    \item $\chi(r)=1$ for  $r\in[0,R_1]$;
    \item $\chi(r)=0$ for $r\in[R_2,\infty)$;
    \item  $\chi^{(k)}(R_1)=\chi^{(k)}(R_2)=0$ for $k=1,2,3$ and $\chi'(r)\leq 0$ on $(R_1,R_2)$.
\end{itemize}
\end{lemma}

\begin{proof}
We construct $\chi$ explicitly.

First, introduce the rescaled variable
$$
t=\frac{r-R_1}{R_2-R_1}\in[0,1], \quad r\in [R_1,R_2].
$$
It is therefore sufficient to construct a function $p:[0,1]\to[0,1]$ such that
\begin{itemize}
     \item $p(0)=1$ and $p(1)=0$;
    \item  $p^{(k)}(0)=p^{(k)}(1)=0$ for $k=1,2,3$;
    \item  $p'(t)\leq 0$ on $(0,1)$.
\end{itemize}

Given such a function $p$, we can define $\chi$ as 
\begin{align*}
\chi(r)=
\begin{cases}
1, & 0\leq r\leq R_1,\\
p\left(\frac{r-R_1}{R_2-R_1}\right), & R_1<r<R_2,\\
0, & r\geq R_2.
\end{cases}
\end{align*}

Since we impose eight linear conditions on $p$ (the values of $p,p',p'',p'''$ at $0,1$),
it is natural to seek $p$ as a polynomial of degree 7:
\begin{align*}
    p(t)=a_0+a_1t+a_2t^2+a_3t^3+a_4t^4+a_5t^5+a_6t^6+a_7t^7.
\end{align*}

The conditions at $t=0$ give $a_0=1$ and $a_1=a_2=a_3=0$ immediately.
Hence
$$
p(t)=1+a_4t^4+a_5t^5+a_6t^6+a_7t^7.
$$

Consider the conditions at $t=1$ and we get
\begin{align*}
    \begin{cases}
        1+a_4+a_5+a_6+a_7=0,\\
        4a_4+5a_5+6a_6+7a_7=0,\\
        12a_4+20a_5+30a_6+42a_7=0,\\
        24a_4+60a_5+120a_6+210a_7=0.
    \end{cases}
\end{align*}

Solving this system yields
$$
a_4=-35,\quad a_5=84,\quad a_6=-70,\quad a_7=20.
$$
Therefore,
$$
p(t)=1-35t^4+84t^5-70t^6+20t^7.
$$

Define
\begin{align*}
\chi(r)=
\begin{cases}
1, & 0\leq r\leq R_1,\\
1-35\left(\dfrac{r-R_1}{R_2-R_1}\right)^4
+84\left(\dfrac{r-R_1}{R_2-R_1}\right)^5
-70\left(\dfrac{r-R_1}{R_2-R_1}\right)^6
+20\left(\dfrac{r-R_1}{R_2-R_1}\right)^7,
& R_1<r<R_2,\\
0, & r\geq R_2.
\end{cases}
\end{align*}
By construction,
$$
\chi(r)=1\quad \text{for }r\in[0,R_1],
\quad
\chi(r)=0\quad \text{for }r\in[R_2,\infty).
$$
Moreover, since 
$p(0)=1$ and $p(1)=0$,
$\chi$ is continuous at $r=R_1$ and $r=R_2$.

By the chain rule, for $r\in(R_1,R_2)$,
\begin{align*}
\chi'(r)=\frac{1}{R_2-R_1}p'(t),\quad
\chi''(r)=\frac{1}{(R_2-R_1)^2}p''(t),\quad
\chi'''(r)=\frac{1}{(R_2-R_1)^3}p'''(t),
 \end{align*}
where $t=(r-R_1)/(R_2-R_1)$.
Since $p'(0)=p''(0)=p'''(0)=0$, we have
\begin{align*}
    \lim_{r \to R_1^+} \chi'(r) =\lim_{r \to R_1^+} \chi''(r)=\lim_{r \to R_1^+} \chi'''(r)= 0.
\end{align*}
Similarly, $p'(1)=p''(1)=p'''(1)=0$ gives 
\begin{align*}
    \lim_{r \to R_2^-} \chi'(r) =\lim_{r \to R_2^-} \chi''(r)=\lim_{r \to R_2^-} \chi'''(r)= 0.
\end{align*}
It's clear that $\chi$ is constant outside $(R_1,R_2)$, 
therefore,
$$
\chi^{(k)}(R_1)=\chi^{(k)}(R_2)=0,
\quad k=1,2,3.
$$

Moreover, note that
$$
p'(t)=-140t^3+420t^4-420t^5+140t^6
      =-140t^3(1-t)^3<0
$$
for  $t\in(0,1)$.
Consequently, for every $r\in(R_1,R_2)$,
$$
\chi'(r)=\frac{1}{R_2-R_1}p'(t)\leq 0.
$$
This also guarantees that $0\leq \chi(r)\leq 1$ for $r\geq 0$.
\end{proof}

\begin{lemma}[Tail of $\pi$]\label{lem:pi}
As $|x|\to\infty$,
$$
\pi(x)\sim \frac{c\Gamma(1+a)\sin(\pi a/2)}{\pi}|x|^{-(1+a)}.
$$
\end{lemma}

\begin{proof}
First, by integration by parts,
  \begin{align*}
    \pi(x) =&\frac{1}{\pi}\int_0^\infty e^{-c t^a}\cos(tx) \d t\\
    =&\frac{1}{\pi}\int_0^\infty e^{-c t^a}\frac{\partial }{\partial t} \left(\frac{\sin(tx)}{x}\right)\d t\\
    =&\frac{ca}{\pi x}\int_0^\infty e^{-c t^a} t^{a-1}\sin(tx) \d t.
  \end{align*}

For $x>0$, substitute $u=tx$ to get
$$
\pi(x)=\frac{ca}{\pi}x^{-(1+a)}\int_0^\infty e^{-c(u/x)^a}u^{a-1}\sin u \d u.
$$
By dominated convergence theorem, we have 
\begin{align*}
  \lim_{x\rightarrow\infty} \frac{\pi(x)}{x^{-(1+a)}}
  =&\lim_{x\rightarrow\infty}\frac{ca}{\pi}\int_0^\infty e^{-c(u/x)^a}u^{a-1}\sin u \d u\\
  =&\frac{ca}{\pi}\int_0^\infty u^{a-1}\sin u \d u\\
  =&\frac{ca}{\pi}\Gamma(a)\sin(\pi/2).
\end{align*}

Similarly for $x<0$, hence 
$$
\pi(x)\sim \frac{c\Gamma(1+a)\sin(\pi a/2)}{\pi}|x|^{-(1+a)}.
$$
\end{proof}

\begin{lemma}[Tail of $\pi'$]\label{lem:pi'}
As $|x|\to\infty$,
$$
\pi'(x)\sim -\frac{c\Gamma(2+a)\sin(\pi a/2)}{\pi}
\frac{\operatorname{sign}(x)}{|x|^{2+ a}}.
$$
\end{lemma}

\begin{proof}
  Integration by parts, we have
  \begin{align*}
    \pi'(x)=&-\frac{1}{\pi}\int_0^\infty t e^{-c t^a}\sin(tx) \d t\\
    =&-\frac{1}{\pi}\int_0^\infty  e^{-c t^a}  \frac{\partial }{\partial t}\left(-\frac{t\cos(tx)}{x}+\frac{\sin(tx)}{x^2}\right) \d t\\
    =&\frac{ca}{\pi x}\int_0^\infty e^{-c t^a} t^a \cos(tx) \d t
    -\frac{ca}{\pi x^2}\int_0^\infty e^{-c t^a} t^{a-1}\sin(tx) \d t.
  \end{align*}

For $x>0$, use $u=tx$ and dominated convergence theorem, we have
\begin{align*}
    \lim_{x\rightarrow\infty} \frac{\pi'(x)}{x^{-(2+a)}}=&
    \frac{ca}{\pi }\int_0^\infty u^a \cos u \d u
    -\frac{ca}{\pi}\int_0^\infty  u^{a-1}\sin u \d u
    =-\frac{c\Gamma(2+a)\sin(\pi a/2)}{\pi}.
\end{align*}

Similarly for $x<0$ and we can obtain
$$
\pi'(x)\sim -\frac{c\Gamma(2+a)\sin(\pi a/2)}{\pi}
\frac{\operatorname{sign}(x)}{|x|^{2+a}}.
$$
\end{proof}

\begin{lemma}[Tail of $\pi''$]\label{lem:pi''}
As $|x|\to\infty$,
$$
\pi''(x)\sim \frac{c\Gamma(3+a)\sin\frac{\pi a}{2}}{\pi}|x|^{-(3+a)}.
$$
\end{lemma}

\begin{proof}
Using the second derivative of the inverse Fourier formula,
$$
\pi''(x)=-\frac{1}{\pi}\int_0^\infty t^2 e^{-c t^a}\cos(tx)\d t.
$$

  Integration by parts, we have
\begin{align*}
  \pi''(x)=&-\frac{1}{\pi}\int_0^\infty t^2 e^{-c t^a}\cos(tx)\d t\\
  =&-\frac{1}{\pi}\int_0^\infty e^{-c t^a} \frac{\partial }{\partial t}\left(\frac{t^2\sin(tx)}{x}+\frac{2t\cos(tx)}{x^2}-\frac{2\sin(tx)}{x^3}\right) \d t\\
  =&-\frac{ca}{\pi}\Big\{\frac{1}{x} \int_0^\infty e^{-c t^a}t^{a+1}\sin(tx)dt
+\frac{2}{x^2}\int_0^\infty e^{-c t^a}t^{a}\cos(tx)dt
-\frac{2}{x^3}\int_0^\infty e^{-c t^a}t^{a-1}\sin(tx)dt\Big\}.
\end{align*}

For $x>0$, use $u=tx$ and dominated convergence theorem, we have
\begin{align*}
    \lim_{x\rightarrow\infty} \frac{\pi''(x)}{x^{-(3+a)}}=&
  -\frac{ca}{\pi}\left(\int_0^\infty u^{a+1}\sin u \d u
+2\int_0^\infty u^{a}\cos u \d u
-3\int_0^\infty u^{a-1}\sin u \d u \right).
\end{align*}

Since
\begin{align*}
\int_0^\infty u^{a+1}\sin u \d u&=\Gamma(2+a)\sin\frac{\pi(2+a)}{2}=-(1+a)\Gamma(1+a)\sin\frac{\pi a}{2},\\
\int_0^\infty u^{a}\cos u \d u &=\Gamma(1+a)\cos\frac{\pi(1+a)}{2}=-\Gamma(1+a)\sin\frac{\pi a}{2},\\
\int_0^\infty u^{a-1}\sin u \d u&=\Gamma(a)\sin\frac{\pi a}{2}=\frac{1}{a}\Gamma(1+a)\sin\frac{\pi a}{2}.
\end{align*}
We can compute that 
$$
\pi''(x)\sim  \frac{c\Gamma(3+a)\sin\frac{\pi a}{2}}{\pi}x^{-(3+a)}.
$$
as $x\to+\infty$.

Similarly for $x<0$, we have 
$$
\pi''(x)\sim  \frac{c\Gamma(3+a)\sin\frac{\pi a}{2}}{\pi} |x|^{-(3+a)},\qquad |x|\to\infty.
$$
\end{proof}

\begin{lemma}[Tail of $\pi'''$]\label{lem:pi'''}
As $|x|\to\infty$,
$$
\pi'''(x)\sim -\frac{c\Gamma(4+a)\sin\bigl(\tfrac{\pi a}{2}\bigr)}{\pi}
\operatorname{sign}(x)|x|^{-(4+a)}.
$$
\end{lemma}

\begin{proof}
Using the second derivative of the inverse Fourier formula,
$$
\pi'''(x)=\frac{1}{\pi}\int_0^\infty t^3 e^{-c t^a}\sin(tx)\d t.
$$

  Integration by parts, we have
\begin{align*}
  \pi'''(x)=&\frac{1}{\pi}\int_0^\infty t^3 e^{-c t^a}\sin(tx)\d t\\
  =&\frac{1}{\pi}\int_0^\infty e^{-c t^a}
\frac{\partial}{\partial t}\left(
-\frac{t^3\cos(tx)}{x}
+ \frac{3t^2\sin(tx)}{x^2}
+ \frac{6t\cos(tx)}{x^3}
- \frac{6\sin(tx)}{x^4}
\right)\d t\\
 =&\frac{ca}{\pi}\Big[-\frac{1}{x} \int_0^\infty e^{-c t^a}t^{a+2}\cos(tx)\d t
+\frac{3}{x^2}\int_0^\infty e^{-c t^a}t^{a+1}\sin(tx)\d t\\
 &\quad \quad+\frac{6}{x^3}\int_0^\infty e^{-c t^a}t^{a}\cos(tx)\d t
 -\frac{6}{x^4}\int_0^\infty e^{-c t^a}t^{a-1}\sin(tx)\d t \Big].
\end{align*}

For $x>0$, use $u=tx$ and dominated convergence theorem, we have
\begin{align*}
    \lim_{x\rightarrow\infty} \frac{\pi'''(x)}{x^{-(4+a)}}=&
  \frac{ca}{\pi}\left(-\int_0^\infty u^{a+2}\cos u \d u
+3\int_0^\infty u^{a+1}\sin u\d u
+ 6\int_0^\infty u^{a}\cos u \d u
 -6\int_0^\infty u^{a-1}\sin u \d u \right).
\end{align*}

Since
\begin{align*}
   \int_0^\infty u^{a+2}\cos u \d u &=\Gamma(3+a)\cos\frac{\pi(3+a)}{2}=(1+a)(2+a)\Gamma(1+a)\sin\frac{\pi a}{2},\\
\int_0^\infty u^{a+1}\sin u \d u&=\Gamma(2+a)\sin\frac{\pi(2+a)}{2}=-(1+a)\Gamma(1+a)\sin\frac{\pi a}{2},\\
\int_0^\infty u^{a}\cos u \d u &=\Gamma(1+a)\cos\frac{\pi(1+a)}{2}=-\Gamma(1+a)\sin\frac{\pi a}{2},\\
\int_0^\infty u^{a-1}\sin u \d u&=\Gamma(a)\sin\frac{\pi a}{2}=\frac{1}{a}\Gamma(1+a)\sin\frac{\pi a}{2}.
\end{align*}
We can compute that 
$$
\pi'''(x)\sim  -\frac{c\Gamma(4+a)\sin\frac{\pi a}{2}}{\pi}x^{-(4+a)}.
$$
as $x\to+\infty$.

Similarly for $x<0$, we have 
$$
\pi'''(x)\sim \frac{c\Gamma(4+a)\sin\frac{\pi a}{2}}{\pi} |x|^{-(4+a)},\qquad |x|\to\infty.
$$
\end{proof}

\end{document}